\documentclass[11pt]{article}

\usepackage[normalem]{ulem}
\usepackage[justification=justified,singlelinecheck=false]{caption}

\usepackage[OT1]{fontenc}
\usepackage{amsthm,amsmath,amssymb}
\usepackage{natbib}
\usepackage{multirow}
\usepackage[pdftex]{graphicx}
\usepackage{subfigure}
\usepackage{makecell}
\usepackage{booktabs}
\usepackage{array}
\usepackage{fullpage}
\usepackage{url}
\usepackage{algorithm}
\usepackage{algorithmic}
\usepackage{bm}
\usepackage{smile}
\usepackage{color}
\usepackage{color,hyperref}
\definecolor{darkblue}{rgb}{0.0,0.0,0.45}

\usepackage{natbib}
\usepackage{multirow}

\numberwithin{equation}{section}
\numberwithin{theorem}{section}

\newtheorem{condition}{Condition}

\DeclareMathOperator*{\arginf}{arginf}

\begin{document}

\title{\Huge Smooth Projected Density Estimation}

\author{Heather Battey\thanks{Current addresses: Department of Mathematics, Imperial College London, London, SW7 2AZ; e-mail: \texttt{h.battey@imperial.ac.uk}; ORFE, Princeton University, Princeton, NJ 08544; e-mail: \texttt{hbattey@princeton.edu}}\\ \small{\emph{University of Bristol}}  \and Han Liu\thanks{ORFE, Princeton University, Princeton, NJ 08544; e-mail: \texttt{hanliu@princeton.edu}}\\ \small{\emph{Princeton University}}}

\normalsize{ }

\maketitle

\begin{abstract}
We introduce and analyse a new nonparametric estimator of a multi-dimensional density. Our smooth projection estimator (SPE) is defined by a least-squares projection of the sample onto an infinite dimensional mixture class via an undersmoothed nonparametric pilot estimate, which acts as a structural filter to regularise the solution. The undersmoothing is required to optimise the convergence rate of the SPE, which is jointly determined by that of the pilot estimator to the true density in squared $\LL_{2}$ norm and by that of the pilot distribution function to the empirical distribution function in uniform norm. Our procedure was conceived with a view to exploiting well known results in convex analysis and their connection to mixture densities. In the context of our work, this translates to the observation that the infinite dimensional minimisation problem, implicit in the construction of the SPE, possesses a solution of dimension at most $n+1$, where $n$ is the sample size. The SPE thus enjoys practical advantages such as computational efficiency, ease of storage and rapid evaluation at a new data point.
\end{abstract}

\section{Introduction}\label{sectionIntro}

Nonparametric density estimation is a fundamental problem in statistics. Despite its long history, the research field has remained an active one thanks to the broad applicability of density estimators for a range of modern statistical procedures. In spite of the flexibility offered by nonparametric estimators, in many practical applications it is convenient to have, not only an estimate of the density over a set of evaluation points, but also a succinct representation of the density function that is easily stored and rapidly evaluated at a new data point. Density estimates arising from nonparametric mixture models offer both the flexibility of a nonparametric estimator and the succinctness of a parametric one. Their drawback is that they are unable to exploit structure in the data generating process that is either detected or assumed in order to reduce estimation error. Assumed structure may simply be that true density function belongs to a pre-specified smoothness class, as is standard in nonparametric estimation problems. Alternatively, one may seek to exploit shape, topological or graphical structure. For instance, when the dimension of the density to be estimated is large relative to the sample size, the only way to achieve consistency is through the exploitation of assumed sparsity in some suitable domain, such as the conditional independence graph.

In this paper, we propose a flexible procedure, the \emph{smooth projection estimator} (SPE) that yields a succinct parametric representation and that is able to exploit structure. The SPE is constructed with a view to achieving the advantages of both parametric and nonparametric procedures whilst circumventing the negative features associated with each (see Table 1), and is defined as a least squares projection of the sample onto an infinite dimensional class of mixture densities via a nonparametric pilot estimator. The latter acts as a structural filter to regularise the solution. Minimisation of the least squares criterion function over the infinite dimensional mixture class gives rise to a mixture law that is supported on at most $n+1$ points, thus the solution is finite dimensional in finite samples. 

\begin{table}[ht]
{\centering
\caption{Advantages and disadvantages of nonparametric and parametric density estimators.}
\begin{tabular}{l|l|l}
\toprule
                                          & Nonparametric (e.g.~kernel estimator) & Parametric (e.g.~finite mixture model) \\
\midrule
\multirow{3}{*}{Pros}                  & $\bullet$ Flexible                    & $\bullet$ Ease of storage \\  
																						 & $\bullet$ Easy to exploit structural information  & $\bullet$ Rapid evaluation at new data point \\    
																						 &                                 & $\bullet$ Fast rate of convergence \\    
\midrule
\multirow{1}{*}{Cons}               &  $\bullet$ Slow rate of convergence    & $\bullet$ Hard to exploit structural information \\  
																			 &  $\bullet$ Large storage requirements         & $\bullet$ Hard to justify a model \\  
																						 &  $\bullet$  Hard to evaluate at new data point       \\
																																										 &   $\bullet$ Hard to choose tuning parameters   &               \\   
\bottomrule
\end{tabular}
}
\end{table}

\section{Smooth projection estimator}\label{sectionMethodology}

The notation of Section \ref{subSecNotation} will be used in the construction of the SPE and in the theoretical derivations appearing in subsequent sections.

\subsection{Notation}\label{subSecNotation}

Let $\mathcal{F}_{d}$ be the class of all Lebesgue densities on $\mathbb{R}^{d}$, and let $Y_{1},\ldots,Y_{n}$ be a sample of $n$ i.i.d.~copies of $Y$ drawn from distribution $P$ on $\mathbb{R}^{d}$ with Lebesgue density $f_{0}\in\mathcal{F}_{d}$. The corresponding empirical measure is $P_{n}$, defined for any Borel set $A$ as $P_{n}(A)=\frac{1}{n}\sum_{i=1}^{n}\ind\{Y_{i}\in A\}$.

With notation inspired by \citet{SamworthYuan2012}, for an arbitrary probability density or mass function $g$, an arbitrary probability measure $Q$ and an arbitrary class of Lebsegue densities $\mathcal{C}$, all on $\mathbb{R}^{d}$, define the projection operator
\[
\psi^{*}_{\mathcal{C}}(Q,g):=\argmin_{f\in\mathcal{C}}\int_{\mathbb{R}^{d}}(g-f)^{2}dQ. 
\] 
Of particular interest, are $Q\in\{P,P_{n}\}$, and $g\in\{f_{0},f_{n},\widehat{f}^{P}\}$, where $\widehat{f}^{P}$ is a nonparametric pilot estimate based on $Y_{1},\ldots,Y_{n}$ and $f_{n}$ is the collection of $1/n$-weighted point masses at $Y_{1},\ldots,Y_{n}$. Also of interest is $\mathcal{C}\in\{\mathcal{F}_{d}^{\mathcal{G}},\mathcal{F}_{d}^{S},\bar{\mathcal{F}}_{d}^{S}\}$, where 
\begin{equation}\label{mixInfinite}
\mathcal{F}_{d}^{\mathcal{G}}:=\left\{f\in\mathcal{F}_{d}: f(y)=\int_{\Theta}f(y;\theta)dG(\theta): \;\;G\in\mathcal{G} \right\},
\end{equation}
is an infinite dimensional mixture class parameterised by $G\in\mathcal{G}$, the space of probability measures on $\Theta$. For notational simplicity, we assume that the density $f_{\theta}=f(\cdot;\theta)$ is parameterised by a single parameter vector $\theta\in\RR^{d}$. The finite dimensional analogue of \ref{mixInfinite} is
\[
\mathcal{F}_{d}^{S}:=\left\{f\in\mathcal{F}_{d}: f(y)=\sum_{s=1}^{S}\pi_{s} f(y;\theta_{s}): \hspace{1pt} \pi\in\Delta^{S},\hspace{1pt}\theta_{1},\ldots,\theta_{S}\in\Theta\right\}
\]
where $\Delta^{S}$ is the $S$-dimensional unit simplex. Finally, $\bar{\mathcal{F}}_{d}^{S}$ is the special case of $\mathcal{F}_{d}^{S}$ in which the mixture components $\{f_{\theta}: \theta\in\Theta\}$ are taken to be the $d$-dimensional spherical Gaussian densities with mean vector $\theta$. More explicitly
\[
\bar{\mathcal{F}}_{d}^{S}:=\left\{f\in\mathcal{F}_{d}: f(y)=\sum_{s=1}^{S}\pi_{s} \phi(\cdot;\mu_{s},\underline{q}I_{d}): \hspace{1pt} \pi\in\Delta^{S},\hspace{1pt}\mu_{1},\ldots,\mu_{S}\in\mathcal{M}, \underline{q}>0\right\},
\]
where $\mathcal{M}=[-M,M]^{d}$, $\underline{q}>0$ is fixed and $\phi(\cdot;\mu,\Sigma)$ is the Gaussian density function with mean vector $\mu$ and covariance matrix $\Sigma$. Proposition \ref{propContainment} illustrates that there is no loss of generality by choosing $\underline{q}\hspace{1pt}I_{d}$ rather than some other covariance matrix with equal diagonal elements no smaller than $\underline{q}$. For notational compactness, we write $\Xi=(\Delta^{S}\times\Theta^{S})$, where $\Theta^{S}$ is the $S$-element cartesian product $\Theta\times \Theta \times \cdots$. Write $\xi=(\xi_{1},\ldots,\xi_{S},\xi_{S+1}^{T},\ldots,\xi_{2S}^{T})^{T}=(\pi^{T}, \theta_{1}^{T},\ldots, \theta_{S}^{T})^{T}$, so $\xi$ is a vector of $2S$ parameters, the first $S$ of which are valued in $[0,1]$ and the last $S$ of which are valued in $\RR^{d}$, hence $\Xi\subset\RR^{S(1+d)}$. This allows us to write $f_{\xi}=\sum_{s=1}^{S}\pi_{s}f_{\theta_{s}}$. For the purpose of the theoretical results in Section \ref{statProperties}, introduce $\widehat{\mathbb{M}}_{n}(\xi)=\int\bigl(\widehat{f}^{P} - f_{\xi}\bigr)^{2}dP_{n}$ for $f_{\xi}\in\mathcal{F}_{d}^{S}$, and similarly define $\mathbb{M}_{0}(\xi)=\int\bigl(f_{0} - f_{\xi}\bigr)^{2}dP$. The notation introduced above is summarised in Table \ref{tableNotation}, together with some additional notation convenient for the proofs in Appendix \ref{sectionAppendixB}.

$\mathbb{L}_{p}(\mathcal{V}):=\mathbb{L}_{p}(\mathcal{V},\text{Leb})$, denotes the space of Lebesgue $p$-integrable functions on $\mathcal{V}$, where $1\leq p<\infty$ and $\mathbb{L}_{\infty}(\mathcal{V})$ denotes the space of uniformly bounded functions on $\mathcal{V}$. We make use of the following notation for weakly differentiable functions from \citet{Ziemer1989}. $\alpha=(\alpha_{1},\ldots,\alpha_{d})$ is a multi-index of non-negative integers, $|\alpha|=\sum_{j=1}^{d}\alpha_{j}$, and $\alpha!=\alpha_{1}!\alpha_{2}!\cdots \alpha_{d}!$. If $x=(x_{1},\ldots,x_{d})\in\RR^{d}$ we will let $x^{\alpha}=x_{1}^{\alpha_{1}}x_{2}^{\alpha_{2}}\cdots x_{d}^{\alpha_{d}}$. The partial derivative operators are denoted by $D_{i}=\partial/\partial x_{i}$ for $1\leq i\leq d$ and the higher order derivatives are denoted by
\[
D^{\alpha}=D_{1}^{\alpha_{1}}\cdots D_{d}^{\alpha_{d}}=\frac{\partial^{|\alpha|}}{(\partial x_{1})^{\alpha_{1}}\cdots (\partial x_{d})^{\alpha_{d}}}.
\]
$u$, which belongs to the space of locally integrable functions on the open set $\mathcal{V}$, is the $\alpha^{th}$ weak derivative of a function $f$ if $\int_{\mathcal{V}}\varphi(x)u dx=(-1)^{|\alpha|}\int_{\mathcal{V}}f(x)D^{\alpha}\varphi(x)dx$ for all $\varphi\in \mathcal{C}^{\infty}_{0}(\mathcal{V})$, the space of infinitely differentiable functions with compact support on $\mathcal{V}$. We write the $\alpha^{th}$ weak derivative of $f$ as $u=D^{\alpha}f$. 
For $p\geq 1$ and $k$ a non-negative integer, the Sobolev space is defined as in \citet[][page 43]{Ziemer1989} as $\mathfrak{W}_{k,p}(\mathcal{V}):=\mathfrak{W}_{k,p}(\mathcal{V},\text{Leb})=\LL_{p}(\mathcal{V},\text{Leb})\cap \{g:D^{\alpha}g\in\LL_{p}(\mathcal{V},\text{Leb}), |\alpha|\leq k\}$. 

The following norms are used throughout. $\|v\|_{\ell_{p}}:=\bigl(|v_{1}|^{p}+\cdots+|v_{d}|^{p}\bigr)^{1/p}$, $\|v\|_{\ell_{\infty}}=\max\{v_{1},\ldots, v_{d}\}$. $\|\cdot\|_{\mathbb{L}_{p}}:=\|\cdot\|_{\mathbb{L}_{p}(\text{Leb})}$ and $\|\cdot\|_{\mathbb{L}_{p}(P)}$ are defined as $\|f\|_{\mathbb{L}_{p}(\text{Leb})}=(\int |f(x)|^{p}dx)^{1/p}$ and $\|f\|_{\mathbb{L}_{p}(P)}=(\int |f(x)|^{p}P(dx))^{1/p}$ respectively, whilst $\|\cdot\|_{\LL_{\infty}}$ is the supremum norm, $\|f\|_{\LL_{\infty}}=\sup_{x\in\mathbb{R}^{d}}|f(x)|$. $\mathfrak{W}_{k,p}(\mathcal{V})$ is equipped with the Sobolev norm $\|f\|_{k,\LL_{p}(\mathcal{V})}=\sum_{|\alpha|\leq k}\|D^{\alpha}f\|_{\LL_{p}(\mathcal{V})}$.

\begin{table}[ht]
\caption{Notation.}
{\centering
\begin{tabular}{l l}
\hline
$g$ & Arbitrary probability density or mass function on $\mathbb{R}^{d}$. \\
$w_{g,i}$ & Weight assigned to observation $i$ by density $g$. \\
$\psi^{*}_{\mathcal{C}}(Q,g)$ & $\argmin_{f\in\mathcal{C}}\int(g-f)^{2}dQ$, where $Q$ is a probability measure on $\mathbb{R}^{d}$. \\
$\Delta^{S}$ & The $S$-dimensional unit simplex. \\
$\Theta$ & A compact finite-dimensional parameter space. \\
$\Theta^{S}$ & The $S$-element cartesian product $\Theta\times\Theta\times\cdots$. \\
$\Xi$ & ($\Delta^{S}\times \Theta^{S}$). \\
$\mathcal{G}$ & The set of all probability measures on $\Theta$. \\
$\mathcal{M}$ & $[-M,M]^{d}$, $M<\infty$. \\
$\mathcal{F}_{d}$ & The class of all Lebesgue densities on $\mathbb{R}^{d}$. \\
$\mathcal{F}^{\mathcal{G}}_{d}$ & $\left\{f\in\mathcal{F}_{d}: f(y)=\int_{\Theta}f(y;\theta)dG(\theta): \;\;G\in\mathcal{G} \right\}$. \\
$\mathcal{F}_{d}^{S}$ & $\{f_{\xi}:\xi\in\Xi\}=\left\{f\in\mathcal{F}_{d}: f(y)=\sum_{s=1}^{S}\pi_{s} f(y;\theta_{s}): \hspace{1pt} \pi\in\Delta^{S},\hspace{1pt}\theta\in\Theta^{S}\right\}$. \\
$\bar{\mathcal{F}}_{d}^{S}$ & $\left\{f\in\mathcal{F}_{d}: f(y)=\sum_{s=1}^{S}\pi_{s} \phi(y;\mu_{s},\underline{q}I_{d}): \hspace{1pt} \pi\in\Delta^{S},\hspace{1pt}\mu_{s}\in\mathcal{M} \;\forall s\in\{1,\ldots,S\}, \underline{q}>0\right\}$. \\
$\mathcal{F}_{d}^{LC}$ & $\left\{f\in\mathcal{F}_{d}: \quad h=\log(f) \text{ is a concave function}\right\}$. \\
$f_{0}$ & The true density function.\\
$f_{n}$ & The collection of ($1/n$)-weighted point masses at $Y_{1},\ldots,Y_{n}$.\\
$\widehat{f}^{P}$ & A generic nonparametric pilot density estimator. \\
$\widehat{f}^{H}$ & A histogram estimator. \\
$\widehat{f}^{k}_{h}$ & A kernel density estimator with bandwidth $h$. \\
$\widehat{f}^{*}_{n}$ & $\psi^{*}_{\mathcal{F}_{d}^{S}}(P_{n},\widehat{f}^{P})$. \\
$\mathbb{M}_{0}(\xi)$, $\widehat{\mathbb{M}}_{n}(\xi)$ & $\int\bigl(f_{0} - f_{\xi}\bigr)^{2}dP$ and $\int\bigl(\widehat{f}^{P} - f_{\xi}\bigr)^{2}dP_{n}$.\\
$\Phi_{Q}(g-f)$ & $\int (g-f)^{2} dQ$, where $Q$ is a probability measure on $\mathbb{R}^{d}$. \\
\hline
\end{tabular}

}\label{tableNotation}
\end{table}

\subsection{Motivation for the SPE}

As above, let $\widehat{f}^{P}$ be a pilot estimator for $f_{0}$. The choice of $\widehat{f}^{P}$ is made, baring in mind Condition \ref{conditionPilot} below, either on computational grounds, or with a view to exploiting assumed or detected structure on $f_{0}$, such as conditional independence relations amongst the marginals of $f_{0}$, unimodality, or other shape restrictions. The resulting estimate is characterised by the advantages and disadvantages of the left hand column of Table 1; for instance, it may be non-differentiable and awkward to evaluate and store. To remove such undesirable features, the pilot estimate is projected onto a mixture class. One possible choice of structural constraint is forest structure, as in Example \ref{exampleConstraint}.

\begin{example}\label{exampleConstraint}
For random variables $X$ and $X'$ independent conditional on $Z$, the joint density of $Y=(X,X',Z)$ is $f_{Y}=f_{X|Z}f_{X'|Z}f_{Z}$. Analogously, any multidimensional joint density possessing sparsity in its conditional independence graph can be expressed in terms of lower dimensional conditional and marginal densities.
\end{example}

\begin{remark}
\emph{\citet{LaffertyLiuWasserman2012}} provide a review of nonparametric graph estimation. Alternatively, graphical structure is sometimes justified by the scientific problem underlying the statistical one \emph{\citep[e.g.][]{HuberGraphsGenetics}}.
\end{remark}

Example \ref{exampleConstraint} exposes a weakness of classical mixture models, which do not lend themselves naturally to graphical sparsity constraints. By contrast, the na\"ive nature of the histogram construction makes it well suited for imposing graphical structure as well as simple shape structure such as unimodality; further explanation is provided in Example \ref{exampleGraphical}.

\begin{example}\label{exampleGraphical}
The general $d$-dimensional histogram is defined as
\begin{equation}\label{densityHist}
\widehat{f}^{H}(y)=\frac{1}{n \prod _{j=1}^{d} h_{n,j}}\sum_{i=1}^{n}{1{\hskip -2.5 pt}\hbox{\text{\emph{I}}}} \{Y_{i}\in A_{n}(y)\}
\end{equation}
for $A_{n}(y)$ the set containing $y$ in the partition $\mathcal{P}_{n}$ of $\mathbb{R}^{d}$, where $\mathcal{P}_{n}$ is defined through an anchor point \emph{(}taken as the origin without loss of generality\emph{)} and a collection of bin widths $\{h_{n,j}: j=1,\ldots, d\}$.

Let $X$, $X'$, $Z$ and $Y$ be as in Example \ref{exampleConstraint}. Estimating $f_{Z}$ using a histogram and conditioning on realisations of $Z$ falling in bin $B$ allows construction of
\[
\widehat{f}^{H}_{(X,X')|B}((x,x')|z\in B)=\widehat{f}^{H}_{X|B}(x|z\in B)\widehat{f}^{H}_{X'|B}(x'|z\in B)\widehat{P}^{H}_{Z}(B),
\]
where $\widehat{f}^{H}_{X|B}$ is used to denote the histogram estimate of $f_{X|B}$ and $\widehat{P}^{H}_{Z}(B)=\int_{B}\widehat{f}^{H}_{Z}(z)dz$ is the estimated probability of $Z$ falling in bin $B$. Analogous estimates for all bins in $\widehat{f}^{H}_{Z}$ are used to construct a histogram estimate of the joint density.
\end{example}

\subsection{Construction of the SPE}

The (infinite dimensional) smooth projection estimator is defined as
\begin{equation}\label{eqSmoothProjection}
\widehat{f}^{*}_{n}:=\psi^{*}_{\mathcal{F}_{d}^{\mathcal{G}}}(P_{n},\widehat{f}^{P}).
\end{equation}
Given mixture components $\{f_{\theta}: \theta\in\Theta\}$, an $f\in\mathcal{F}_{d}^{\mathcal{G}}$ is completely parameterised by a $G\in \mathcal{G}$, hence 
\begin{equation}\label{eqInfMinProblem}
\min_{f\in\mathcal{F}_{d}^{\mathcal{G}}}\int(\widehat{f}^{P}-f)^{2}dP_{n} = \min_{G\in\mathcal{G}}\int\bigl(\widehat{f}^{P}-\int_{\Theta} f_{\theta}dG(\theta)\bigr)^{2}dP_{n}.
\end{equation}
\begin{lemma}\label{minimumExists}
Let $\Theta$ be a compact finite dimensional parameter space. The minimum of equation \eqref{eqInfMinProblem} exists, and there exists a mixing distribution $\widehat{G}$ whose support contains no more than $n+1$ points such that $\psi^{*}_{\mathcal{F}_{d}^{\mathcal{G}}}(P_{n},\widehat{f}^{P})=\int f_{\theta}d\widehat{G}(\theta)$ achieves this minimum.
\end{lemma}
Lemma \ref{minimumExists} is very similar to the well-known result of \citet{Lindsay1983} for mixture likelihoods. For convenience, the argument of the proof is reproduced (with the relevant modifications) in Appendix \ref{sectionAppendixB}. The implication of Lemma \ref{minimumExists} is that a minimiser $\widehat{G}$ necessarily takes the form $\widehat{G}=\sum_{s=1}^{S}\pi_{s}\delta(\theta_{s})$ with $S\leq n+1$, where $\theta_{1},\ldots,\theta_{S}$ are distinct elements of $\Theta$, $\delta(\theta)$ assigns probability one to any set containing $\theta$, and $\pi=(\pi_{1},\ldots, \pi_{S})$ belongs to the unit $S$-simplex, $\Delta^{S}$. The implication of this result is that the original infinite dimensional minimisation problem is equivalent to the finite dimensional minimisation problem
\begin{equation}\label{eqFiniteMinProblem}
\min_{\pi\in\Delta^{S},\theta\in\Theta^{S}}\int\bigl(\widehat{f}^{P}-f_{\pi,\theta}\bigr)^{2}dP_{n} = \min_{\xi\in\Xi}\int(\widehat{f}^{P}-f_{\xi})^{2}dP_{n},
\end{equation}
justifying the finite dimensional SPE $\widehat{f}_{n}^{*}:=\psi^{*}_{\mathcal{F}_{d}^{S}}(P_{n},\widehat{f}^{P})$ with $S\leq n+1$.
\begin{remark}
Lemma \ref{minimumExists} specifies that the support size of $\widehat{G}$ is no larger than $n+1$. From a practical point of view, it is desirable to take $S$ smaller than $n+1$. Simulations reveal decreasing marginal improvements from increasing $S$ and the SPE performs well for $S$ much smaller than $n+1$ \emph{(}see \S \ref{sectionSimulation}\emph{)}.
\end{remark}

\subsection{Statistical properties of the SPE}\label{statProperties}

Conditions \ref{conditionTrueDensity}, \ref{conditionPilot}, \ref{conditionMixtures} and \ref{conditionSE} provide requirements on, respectively, $f_{0}$, $\widehat{f}^{P}$, $\mathcal{F}_{d}^{S}$, and ($f_{0},\hspace{1pt}\widehat{f}^{P},\hspace{1pt}\mathcal{F}_{d}^{S}$) simultaneously that are sufficient for the theoretical results reported in Theorems \ref{thmRate} and \ref{thmConsistency}.

\begin{condition}\label{conditionTrueDensity} \emph{[on true density].}
$f_{0}\in \LL_{\infty}(\RR^{d})\cap \mathfrak{W}_{1,1}(\RR^{d},\text{\emph{Leb}})$.
\end{condition}
Condition \ref{conditionTrueDensity} is a weak one, requiring only that the true density be uniformly bounded with integrable first partial derivatives.

\begin{condition}\label{conditionPilot}\emph{[on pilot estimator].}
$\Pr\bigl(\widehat{f}^{P}\in\mathcal{D}\bigr)\longrightarrow 1$, where $\mathcal{D}$ is a $P$-Donsker class of functions. Furthermore, for $r_{n}$ and $s_{n}$ positive deterministic sequences satisfying $r_{n}\searrow 0$ and $s_{n}\searrow 0$ as $n\rightarrow \infty$, $\widehat{f}^{P}$ satisfies $\mathbb{E}\|\widehat{f}^{P}-f_{0}\|^{2}_{\mathbb{L}_{2}}=O(r_{n})$ and $\mathbb{E}\sup_{t\in\mathbb{R}^{d}}\bigl|\widehat{F}^{P}(t)-F_{n}(t)\bigr|=O(s_{n})$, where $\widehat{F}^{P}$ is the distribution function corresponding to the density function $\widehat{f}^{P}$.
\end{condition}

The rates of convergence $r_{n}$ and $s_{n}$ determine the rate of convergence of the least squares criterion function in Theorem \ref{thmRate}. A particularly interesting case is that in which $r_{n}$ and $s_{n}$ are $O(n^{-1/2})$; this delivers a parametric rate in Theorem \ref{thmRate}. For estimation of densities in $\mathfrak{W}_{\ell,\infty}(\RR^{d})$, the minimax rate is $\mathbb{E}\|\widehat{f}^{P}-f_{0}\|^{2}_{\mathbb{L}_{2}}=O(n^{-2\ell/(2\ell+d)})$ \citep[][]{Ibragimov1983}, which is $O(n^{-1/2})$ for $d\leq 2\ell$. The shifted histogram \citep[e.g.][]{DasGupta2008, Scott1992} with AMISE minimising bin width and the multivariate kernel density estimator with product kernel and AMISE-minimising bandwidth $h=O(n^{-1/(2t+d)})$ both achieve this minimax rate \citep[e.g.][Theorems 5.3 and 6.4]{Scott1992}. The minimax rate of convergence for log-concave density estimation is $\mathbb{E}\|\widehat{f}^{P}-f_{0}\|^{2}_{\mathbb{L}_{2}} = O(n^{-2/(d+1)})$ \citep[][]{KimSamworth2014}, which is $O(n^{-1/2})$ for $d\leq 3$.

Condition \ref{conditionPilot} allows a rate in Theorem \ref{thmRate} below that is faster than the rate of convergence of $\widehat{f}^{P}$ to $f_{0}$ in $\LL_{2}$ norm. The rate is optimised by optimally trading off the rate of convergence of $\widehat{f}^{P}$ to $f_{0}$ in squared $\LL_{2}$ norm and the rate of convergence of $\widehat{F}^{P}$ to $F_{n}$ in supremum norm. This involves a choice of tuning parameter in the pilot estimation stage that converges faster than that typically used for optimal density estimation using the pilot estimator alone. In other words, optimal rates of convergence for SPE involve undersmoothing at the pilot estimation stage, where the degree of undersmoothing is dictated by the precise choice of pilot estimator. We illustrate this for the kernel density estimator.

Let $k_{h}(x)=\frac{1}{h^{d}}k(x/h)$, where $k:\RR^{d}\rightarrow \RR$, and define the kernel-smoothed empirical distribution function as $\widehat{F}_{n,h}^{k}(x)=\int_{-\infty}^{x_{1}}\cdots \int_{-\infty}^{x_{d}}\widehat{f}_{n,h}^{k}(y)dy$ where
\begin{equation}\label{eqKernelDensityEstimator}
\widehat{f}_{n,h}^{k}
(x)=P_{n}\ast k_{h}(x)=\frac{1}{nh^{d}}\sum_{i=1}^{n}k\Bigl(\frac{x-X_{i}}{h}\Bigr) \quad x\in\RR^{d}.
\end{equation}
In Proposition \ref{propositionKernel} we provide a range of bandwidths for which the third requirement of \ref{conditionPilot} is fulfilled. In Proposition \ref{propositionKernel}, $\sigma^{2}(b)$ is the supremum over $0<h<b$ and $x\in\RR^{d}$ of the variance of $\bigl(K\bigl((x-X)/h\bigr)-\ind\{X\leq x\}\bigr)$, where $K(z)=\int_{-\infty}^{z_{1}}\cdots\int_{-\infty}^{z_{d}} k(y)dy$ for $z=(z_{1},\ldots,z_{d})\in\RR^{d}$.
\begin{proposition}\label{propositionKernel}
Let $f_{0}\in\mathfrak{W}_{\ell,\infty}(\RR^{d})$, $\ell\geq 0$, and let $k:\RR^{d}\rightarrow \RR$ be such that $\int_{\RR^{d}}k(z)dz=1$, $\int_{\RR^{d}}|k(z)||z^{\alpha}|dz<\infty$ for all $\alpha$ such that $|\alpha|=\ell+1$, where $z^{\alpha}=z_{1}^{\alpha_{1}}z_{2}^{\alpha_{2}}\cdots z_{d}^{\alpha_{d}}$, and \emph{$\ind\{\ell>0\}\sum_{|\alpha|\leq \ell}\int_{\RR^{d}}z^{\alpha}k(z)dz=0$}. Then, for $b_{n}$ a sequence of constants $0<b_{n}<1$ such that $\sigma^{2}(b_{n})\searrow 0$ as $b_{n}\searrow 0$ and $b_{n}=o(n^{-1/2(\ell+1)}\sqrt{\log\log n})$, $\sup_{0<h\leq b_{n}}\EE\|\widehat{F}_{n,h}^{K}-F_{n}\|_{\LL_{\infty}}=o(\sqrt{n^{-1}\log\log n})$.
\end{proposition}
\begin{remark}
In Proposition \ref{propositionKernel}, the requirement that \emph{$\ind\{\ell>0\}\sum_{|\alpha|\leq \ell}\int_{\RR^{d}}z^{\alpha}k(z)dz=0$} simply amounts to using a symmetric kernel if $\ell=1$, whilst if $f_{0}$ possesses more smoothness, the same clause prescribes the use of higher-order kernels, thereby allowing slower convergence of the bandwidth sequence to deliver the same rate of convergence of $\EE\|\widehat{F}_{n,h}^{K}-F_{n}\|_{\LL_{\infty}}$, and hence a faster rate of convergence for $\EE\|\widehat{f}^{P}-f_{0}\|_{\LL_{2}}^{2}$.
\end{remark}

Using a bandwidth $h$ of order $b_{n}=o(n^{-1/2(\ell+1)}\sqrt{\log\log n})$, and examining the bias and variance terms of the multivariate product kernel estimator \citep[][Theorem 6.4]{Scott1992}, the dominating term is in the variance, and is of order $o((\log\log n)^{-1/2}n^{(d/2(\ell+1))-1})=o(n^{-1/2})$ if $d\leq \ell+1$. Thus, with a pilot estimator constructed as in the setting of Proposition \ref{propositionKernel} with $d\leq \ell+1$, the SPE criterion function achieves a rate of convergence of $o_{p}(\sqrt{n^{-1}\log\log n})$.

The first requirement of Condition \ref{conditionPilot} is not restrictive and is satisfied by e.g.~the kernel density estimator with product Gaussian kernel, as established in Proposition \ref{propositionDonsker}.

\begin{proposition}\label{propositionDonsker}
Suppose $f_{0}\in\mathfrak{W}_{\ell,2}(\RR^{d},\emph{\text{Leb}})$, $d\leq \ell$. Let $k_{h}(v)=\frac{1}{h^{d}}\prod_{j=1}^{d}k(v_{j}/h)$, where $k$ is the Gaussian Kernel $k(x)=\frac{1}{\sqrt{2\pi}}\exp\{-x^{2}/2\}$. With $\widehat{f}^{P}$ the kernel density estimator $\widehat{f}_{n,h_{n}}^{k}$ defined in equation \eqref{eqKernelDensityEstimator} with $h_{n}$ any bandwidth converging no faster than $O_{p}(n^{-1/2d})$, $\Pr(\widehat{f}^{P}\in\mathcal{D})\longrightarrow 1$ as $n\longrightarrow \infty$, where $\mathcal{D}$ is a $P$-Donsker class of functions.
\end{proposition}

\begin{condition}\label{conditionMixtures}\emph{[on mixture class].}
$\mathcal{F}_{d}^{S}=\{f_{\xi}:\xi \in\Xi\}$ with $\Xi$ a compact set. $\mathcal{F}_{d}^{S}\subseteq \LL_{\infty}(\RR^{d})\cap \mathfrak{W}_{1,1}(\RR^{d},\text{\emph{Leb}})$. Finally, for all $\xi,\xi'\in\Xi$, there exists a $K<\infty$ such that $\|f_{\xi}-f_{\xi'}\|_{\mathbb{L}_{2}(\text{\emph{Leb}})}\leq  K \|\xi-\xi'\|_{\ell_{1}}$. 
\end{condition}

Since $\Xi=\Delta^{S}\times\Theta^{S}$ with $\Delta^{S}$ the unit $S$-simplex, compactness of $\Xi$ follows if $\Theta$ is compact, which is also a requirement of Lemma \ref{minimumExists}. A sufficient condition for the Lipschitz requirement on $\mathcal{F}_{d}^{S}$ is that $A_{f_{\xi}}\in\mathbb{L}_{2}(\mathbb{R}^{d},\text{Leb})$ where 
\begin{equation}\label{suffLipschitz}
A_{f_{\xi}}(y):=\Bigl\|\sup_{\xi\in\Xi}\nabla_{\xi}^{T}f_{\xi}(y)\Bigr\|_{\ell_{\infty}}.
\end{equation}
The previous statement follows by the mean value theorem because
\begin{eqnarray*}
|f_{\xi}(y)-f_{\xi'}(y)|&   =  &\left|\left(\nabla_{\xi}^{T}f_{\bar{\xi}}\right)(y)(\xi-\xi') \right| \quad \bar{\xi}\in \text{\emph{conv}}(\xi,\xi')\\
											& \leq & \Bigl\|\sup_{\xi\in\Xi}\nabla_{\xi}^{T}f_{\xi}(y)\Bigr\|_{\ell_{\infty}}\|\xi-\xi'\|_{\ell_{1}}.
\end{eqnarray*}
$A_{f_{\xi}}\in\mathbb{L}_{2}(\mathbb{R}^{d},\text{Leb})$ thus ensures $\|f_{\xi}-f_{\xi'}\|_{\mathbb{L}_{2}(\text{Leb})}\leq K\|\xi-\xi'\|_{\ell_{1}}$ for all $\xi,\xi'\in\Xi$ with $K=\left\|A_{f_{\xi}}\right\|_{\mathbb{L}_{2}(\text{Leb})}$.

\begin{condition}\label{conditionSE}\emph{[on ($f_{0},\hspace{1pt}\widehat{f}^{P},\hspace{1pt}\mathcal{F}_{d}^{S}$) simultaneously].}
For any $\delta>0$,
\begin{equation}\label{conditionStochEqui}
\sup_{\xi,\xi'\in\Xi: \; \|\xi-\xi'\|_{\ell_{1}}<\delta}\Bigl|\frac{1}{n}\sum_{i=1}^{n}\bigl[f_{0}(Y_{i})(f_{\xi'}(Y_{i})-f_{\xi}(Y_{i})\bigr)\bigr]-\mathbb{E}\bigl[f_{0}(Y_{i})\bigl(f_{\xi'}(Y_{i})-f_{\xi}(Y_{i})\bigr)\bigr]\Bigr| = O_{p}(v_{n}),
\end{equation}
where $v_{n}=\max\{n^{-1/2},r_{n},s_{n}\}$ with $r_{n}$ and $s_{n}$ defined as in Condition \ref{conditionPilot}.
\end{condition}
When $\mathcal{F}_{d}^{S}$ is the space of spherical Gaussian mixtures and when $f_{0}$ satisfies Condition \ref{conditionTrueDensity}, Conditions \ref{conditionMixtures} and \ref{conditionSE} are satisfied with a rate of $O_{p}(1/\sqrt{n})$ in equation \eqref{conditionStochEqui}, which a fortiori is $O_{p}(v_{n})$. 
\begin{proposition}\label{propositionNormalMixtures}
Under Condition \ref{conditionTrueDensity} on $f_{0}$ and Condition \ref{conditionPilot} on $\widehat{f}^{P}$, $\bar{\mathcal{F}}_{d}^{S}$ \emph{(}cf.~Table \ref{tableNotation}\emph{)} satisfies Conditions \ref{conditionMixtures} and \ref{conditionSE} with $S\leq n+1$. 
\end{proposition}

\begin{theorem}\label{thmRate}
For $f_{0}$, $f_{\xi}\in\mathcal{F}_{d}^{S}$ $S\leq n+1$ and $\widehat{f}^{P}$ obeying Condition \ref{conditionTrueDensity}, \ref{conditionMixtures} and \ref{conditionPilot} respectively,
\[
\sup_{\xi\in\Xi}\bigl|\widehat{\mathbb{M}}_{n}(\xi) - \mathbb{M}_{0}(\xi)\bigr|=O_{p}(v_{n}).
\]
where $v_{n}=\max\{n^{-1/2},r_{n},s_{n}\}$ with $r_{n}$ and $s_{n}$ defined as in Condition \ref{conditionPilot}.
\end{theorem}

Theorem \ref{thmRate} demonstrates that the empirical criterion function from which the parameters of the SPE are estimated, converges to the ideal criterion function at a rate that is potentially much faster than that of the pilot estimator upon which it is based, provided that the pilot estimator is chosen to satisfy Condition \ref{conditionPilot}. The exact rate is determined by $r_{n}$ and $s_{n}$ from Condition \ref{conditionPilot}. In general, $\psi^{*}_{\mathcal{F}_{d}^{S}}(P,f_{0})$ does not define a unique element of $\mathcal{F}_{d}^{S}$. When it does, and when the projection does not lie on a boundary of the parameter space, the stronger consistency result of Theorem \ref{thmConsistency} is obtained.

\begin{theorem}\label{thmConsistency}
Let $f_{0}$, $\widehat{f}^{P}$ and $\mathcal{F}_{d}^{S}$ be such that \ref{conditionTrueDensity}-\ref{conditionPilot} are fulfilled. Suppose further that $\xi^{*}_{0}$  in $\Xi=(\Delta^{S}\times \Theta^{S})$, defines a unique minimiser of $L_{P}(f_{0};f(\xi))$ and belongs to the interior of $\Xi$. Let $\widehat{\xi}^{*}_{n}$ be a sequence such that $\widehat{\mathbb{M}}_{n}(\xi)\leq \inf_{\xi\in\Xi}\widehat{\mathbb{M}}_{n}+o_{p}(1)$. Then $\widehat{\xi}^{*}_{n} \longrightarrow_{p} \xi^{*}_{0}$.
\end{theorem}

\section{Discussion}\label{sectionDiscussion}

This section focusses primarily on the function performed 
by the structural filtering step of the SPE. The SPE is compared to two other approaches: a direct projection of the ($1/n$)-weighted point masses onto a finite-dimensional mixture class, and a projection of the weighted point masses onto a structurally constrained mixture class.

\subsection{Comparison to direct projection of $f_{n}$ onto an unconstrained mixture class}

In view of the fact that the nonparametric structural filtering step induces a double layer of estimation error, one may question whether this pilot estimation stage is really necessary. Letting $f_{n}$ denote the ($1/n$)-weighted point masses at $Y_{1},\ldots,Y_{n}$, consider the projection of $f_{n}$ onto the class $\bar{\mathcal{F}}_{d}^{S}$ of location mixtures of spherical Gaussian densities. Figures 1 and 2 reveal a 
debilitating feature of the direct projection approach. Specifically, the projected density concentrates its mass on regions of the support where no data are observed.
 
 \bigskip

{\centering
\includegraphics[trim=0.0in 3.8in 0.0in 3.8in, clip, height=0.23\paperwidth]{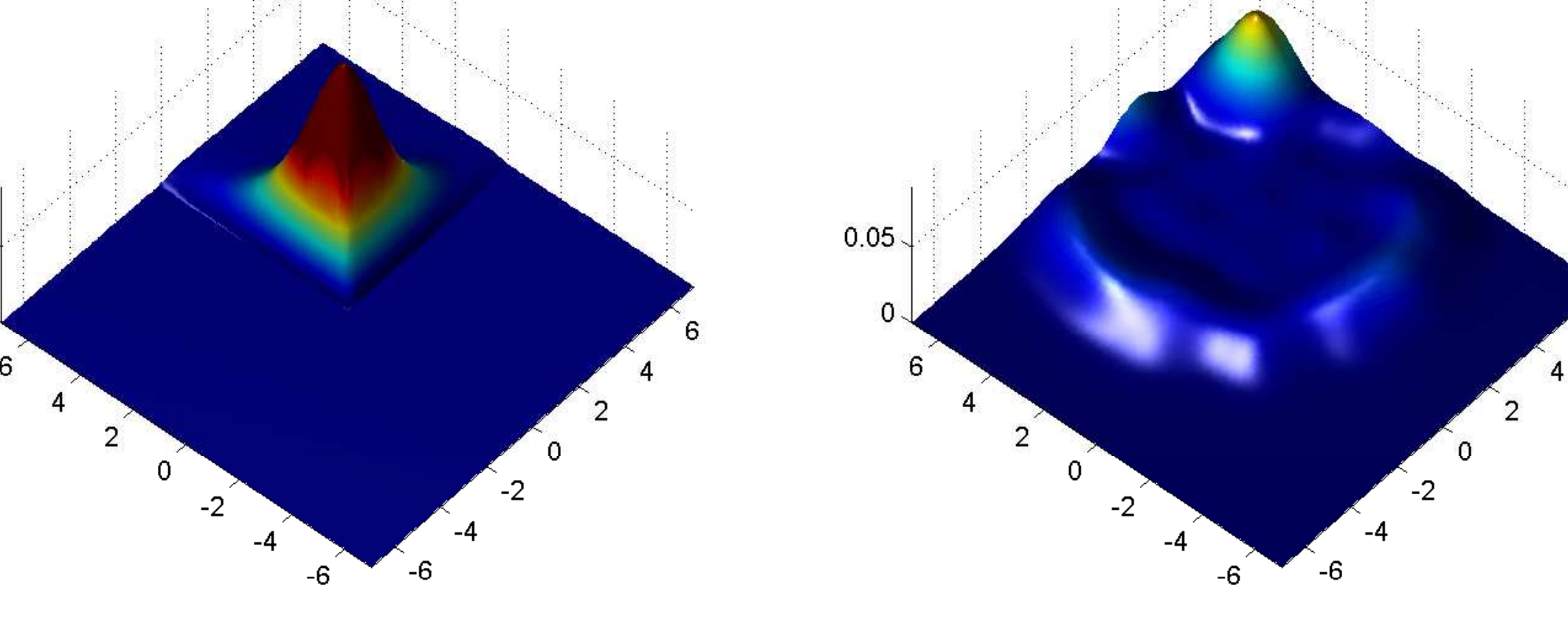}

}

\noindent \small{\textbf{Figure 1.} Left: true gamma density with independent $\Gamma(2,1)$ marginals; projection of the $1/n$-weighted point masses (based on $n=250$ observations) at the data points on the spherical Gaussian mixture class with 64 components.}

\vspace{8pt}

{\centering
\includegraphics[trim=1.5in 3.8in 1.5in 3.8in, clip, height=0.20\paperwidth]{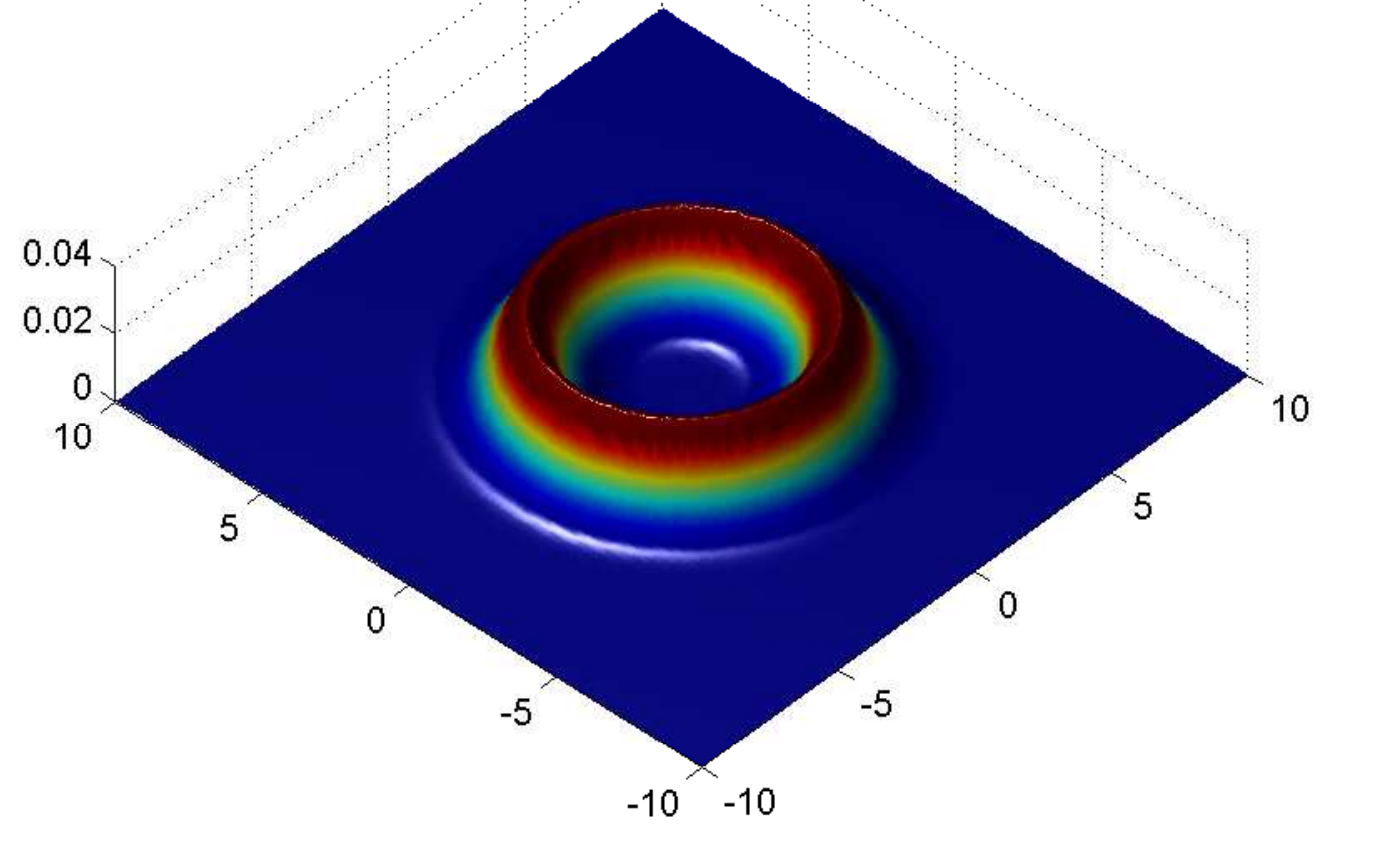}
\includegraphics[trim=1.5in 3.8in 1.8in 3.8in, clip, height=0.19\paperwidth]{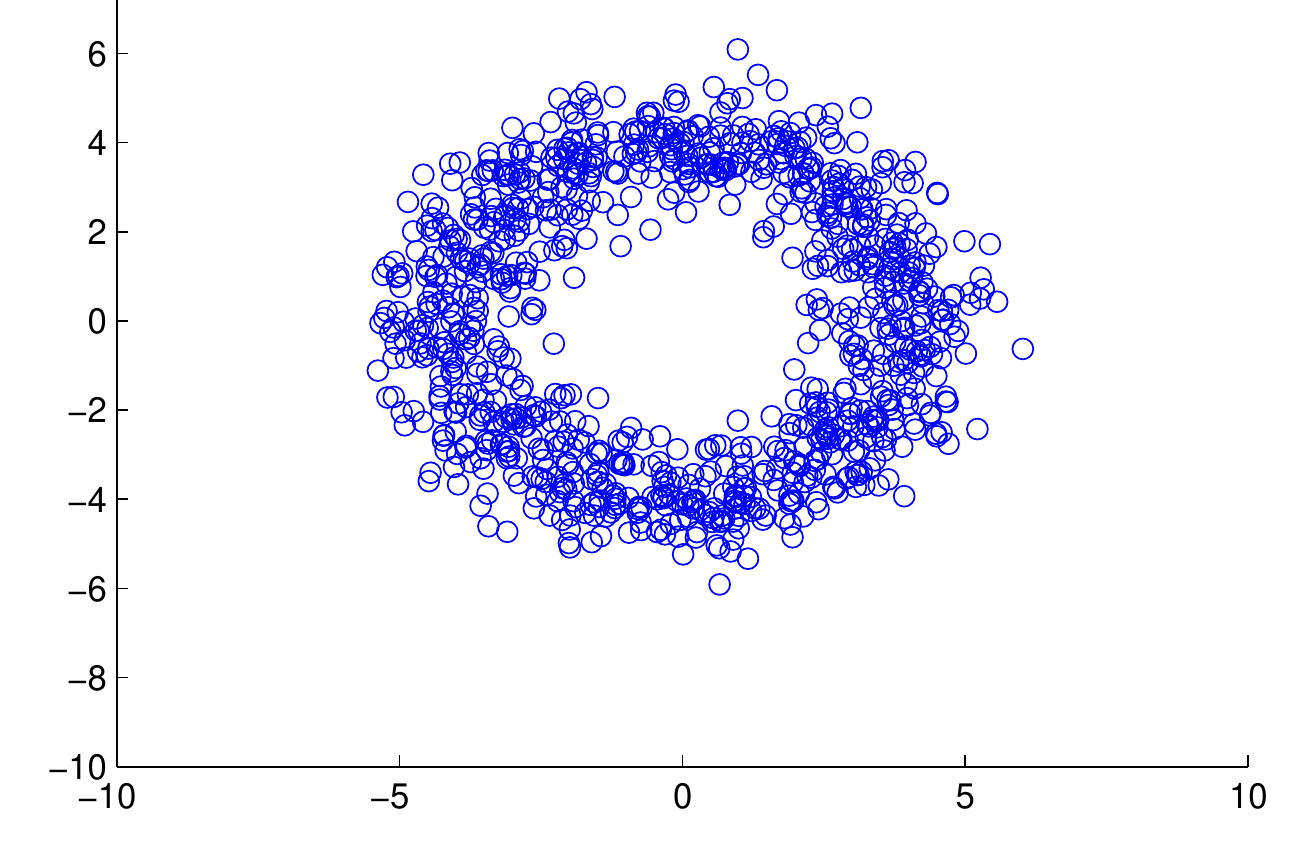}\\
\includegraphics[trim=1.5in 3.8in 1.5in 3.8in, clip, height=0.20\paperwidth]{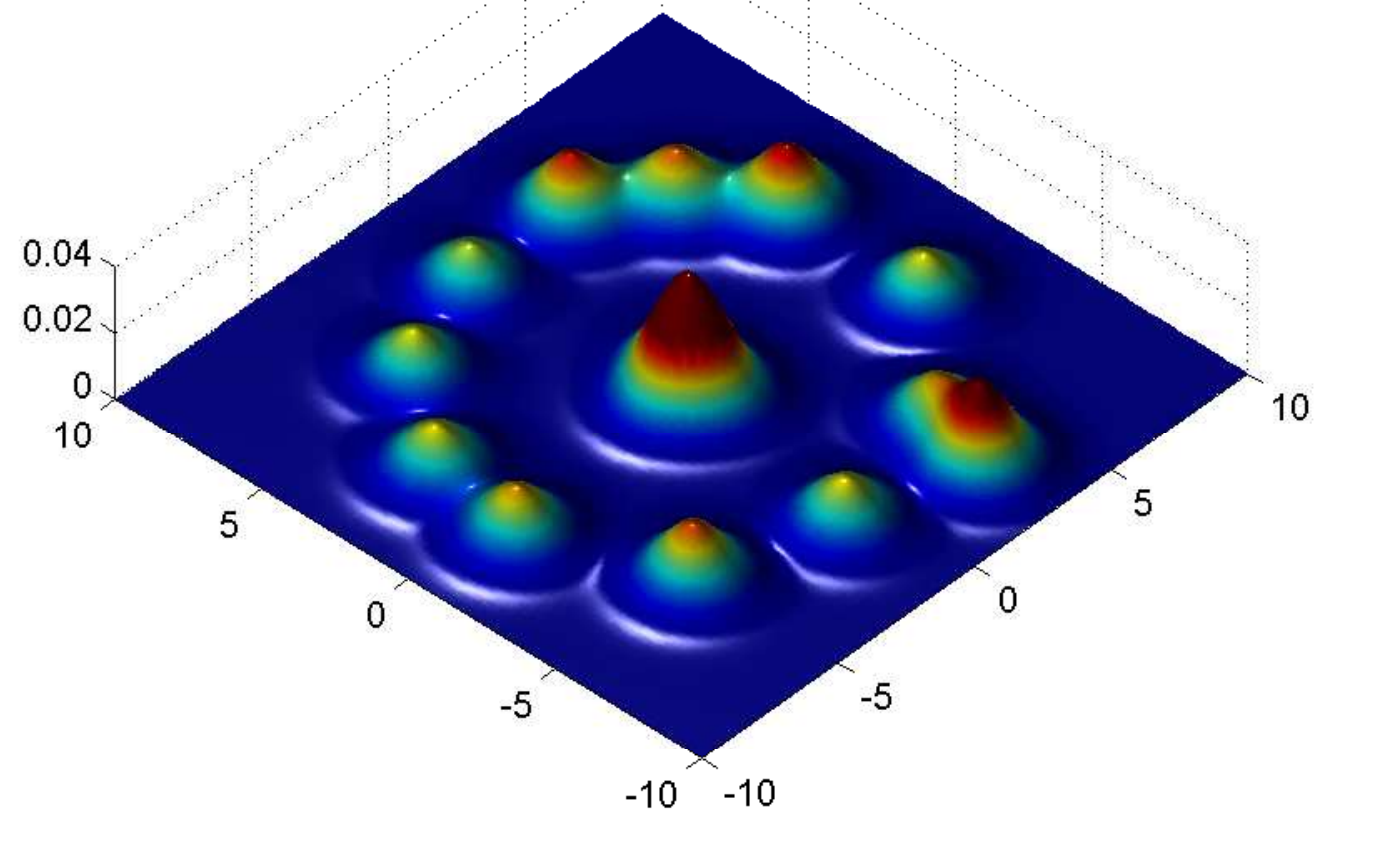}
\includegraphics[trim=1.5in 3.8in 1.8in 3.8in, clip, height=0.20\paperwidth]{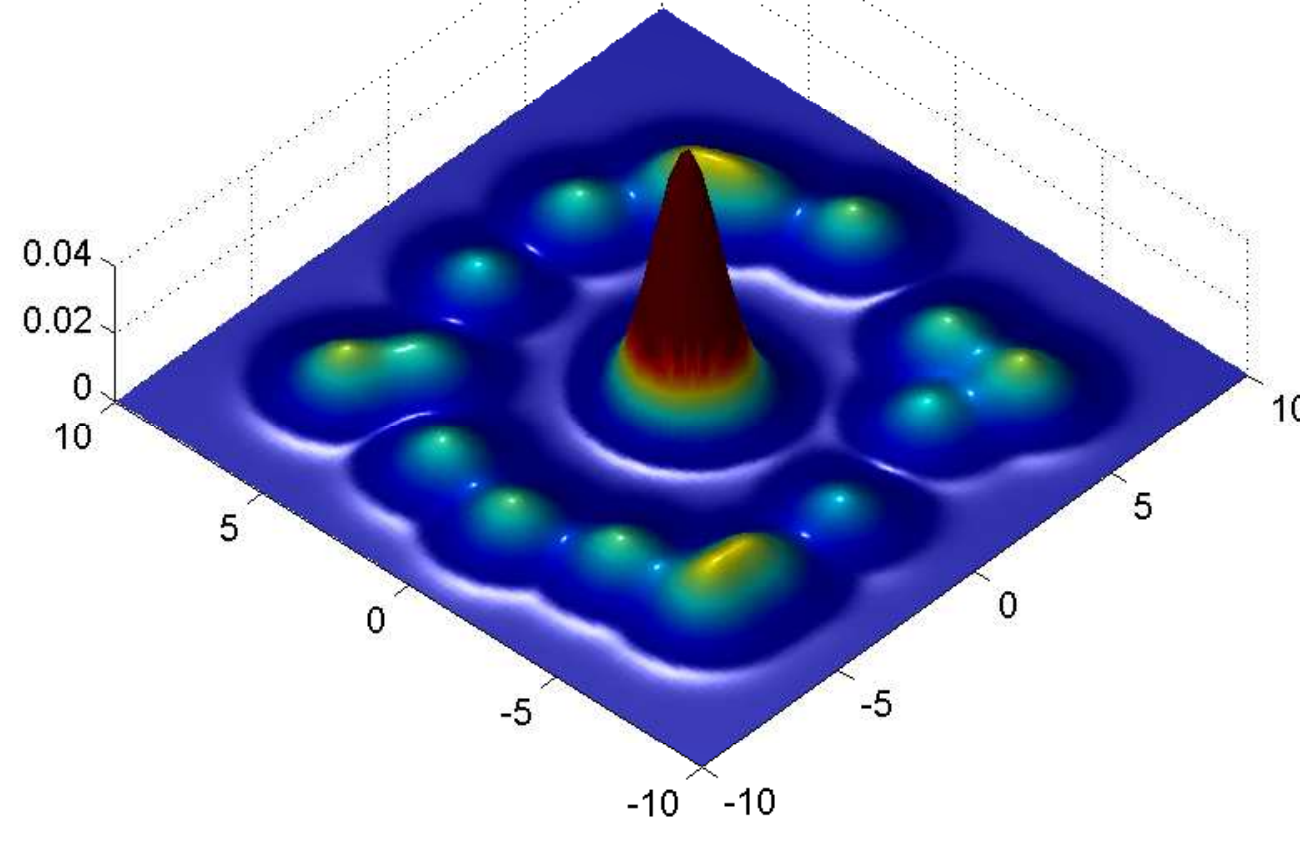}

}

\noindent \small{\textbf{Figure 2.} From top left to bottom right: true density; 1000 random draws from the true density; projection of the $1/n$-weighted point masses at the data points on the spherical Gaussian mixture class with 16 components; projection of the $1/n$-weighted point masses at the data points on the spherical Gaussian mixture class with 25 components.}

\vspace{8pt}

\normalsize{The intuition for Figures 1 and 2 stems from the fact that the criterion function $\Phi_{P_{n}}(g-f)$ is apathetic concerning the ability of a $f\in\mathcal{F}_{d}^{S}$ to well approximate $g$ outside the support of $f_{n}$. Proposition \ref{propositionProjectEmpirical}, gives a concrete description of the nature of the direct projection of $f_{n}$ on the spherical Gaussian mixture class $\bar{\mathcal{F}}_{d}^{S}$ when $f_{0}\in \mathcal{F}_{d}^{LC}$. For log concave $f_{0}$ there exists a nested sequence, $\mathcal{R}_{1}^{f_{0}}\subset \mathcal{R}_{2}^{f_{0}} \subset \cdots$, of closed convex sets such that
\begin{equation}\label{sequenceSets}
\mathcal{R}_{\ell}^{f_{0}}:=\left\{y\in\mathbb{R}^{d}: f_{0}(y)\geq r_{\ell}\right\} \quad r_{\ell}>r_{k} \;\; \forall k>\ell.
\end{equation}
Thus any estimator of $f_{0}\in \mathcal{F}_{d}^{LC}$ should yield a density whose mass is concentrated on these sets. Instead, Proposition \ref{propositionProjectEmpirical} shows that, the direct projection estimator yields, with probability 1, estimates whose mass is concentrated on the complement of successively large convex sets $\mathcal{R}_{k(n)}$ as $n$ grows large.}

\begin{proposition}\label{propositionProjectEmpirical}
Let $f_{0}\in \mathcal{F}_{d}^{LC}$ and $\mathcal{R}_{1}^{f_{0}}\subset \mathcal{R}_{2}^{f_{0}} \subset \cdots$, be defined as in equation \eqref{sequenceSets}.
Suppose further that $f_{0}\notin \bar{\mathcal{F}}_{d}^{S}$ for fixed $S<n$ and write $f_{0}^{*}=\psi^{*}(P;f_{0})$ as 
\[
f_{0}^{*}(y)=\sum_{s=1}^{S}\pi_{0,s}^{*}\phi(y;\mu_{0,s}^{*},\underline{q} I_{d}).
\]
With $g=f_{n}$, for any $n>n_{0}$ \emph{(}$n_{0}<\infty$\emph{)} there exists a $k(n)$ \emph{(}and $r_{k(n)}\searrow 0$ as $n\longrightarrow \infty$\emph{)} such that the $\mu_{n,g}^{*}:=\{\mu_{n,g,1}^{*},\ldots,\mu_{n,g,S}^{*}\}$ that minimise the objective function $\Phi_{P_{n}}(g-f_{\pi_{0}^{*},\mu})$ for $f_{\pi_{0}^{*},\mu}\in \bar{\mathcal{F}}_{d}^{S}$ with $\pi_{0}^{*}=(\pi_{0,1}^{*},\ldots, \pi_{0,S}^{*})\in\Delta^{S}$ known, satisfy one of the following:
\begin{itemize}
\item[(i)] $\mu_{n,g}^{*}\subset \mathbb{R}^{d}\backslash \mathcal{R}_{k(n)}^{f_{0}}$ with probability 1;
\item[(ii)] $\|\mu_{n,g,s}^{*}\|_{\ell_{2}}=\infty$ for all $s\in\{1,\ldots,S\}$.
\end{itemize}
\end{proposition}

\begin{remark}
In the limit as $n\longrightarrow \infty$, scenario (ii) is the only possibility.
\end{remark}

The intuition for the result in Proposition \ref{propositionProjectEmpirical} does not depend on the log concavity of $f_{0}$, only on the fact that $f_{0}$ and $\widehat{f}^{P}$ are both Lebesgue densities whilst $f_{n}$ is not. Whilst $\int f_{n}dP_{n}=1$, $\int \widehat{f}^{P}dP_{n}>1$ as $\widehat{f}^{P}$ is normalised to integrate to one in the most meaningful sense. Replacing $f_{n}$ with a continuous approximation to a mixture of $n$ spikes at the data points, for instance 
\[
f_n(y) = n^{-1}\sum_{i=1}^{n} \phi(y-Y_{i};n^{-1}I_{d})
\] 
where $\phi(y-Y_{i};n^{-1}I_{d})$ is the $d$-dimensional spherical normal with variance $1/n$, imposes Lebesgue integrability and prevents the mass from piling up at points in $\mathbb{R}^{d}$ where no data are observed. However, this continuous approximation is just another example of a pilot estimator, and one that is not convenient from the standpoint of exploiting assumed structure, such as inclusion in a smoothness class.

Proposition \ref{propositionProjectLogConcave} shows that the phenomenon observed in Proposition \ref{propositionProjectEmpirical} is corrected through the use of a pilot density estimator whose mass is concentrated on sufficiently small sets of sufficiently high probability under the true $P$.

\begin{proposition}\label{propositionProjectLogConcave}
Let $g=f_{0}$ or $g=\widehat{f}^{P}$ where $f_{0}, \widehat{f}^{P}\notin \bar{\mathcal{F}}_{d}^{S}$ for fixed $S<n$. Let $A^{g}$ be the smallest set in $\text{\emph{supp}}(g)$ satisfying $P(A^{g})\geq 1/2$ \emph{(}for $g=f_{0}$, $A^{g}$ is simply the $\mathcal{R}_{\ell}^{f_{0}}$ on which half the mass of $f_{0}$ is concentrated\emph{)}. With probability 1, the $\mu_{n,g}^{*}:=\{\mu_{n,g,1}^{*},\ldots,\mu_{n,g,S}^{*}\}$ that minimise the objective function $\Phi_{P_{n}}(g-f_{\pi_{0}^{*},\mu})$ for $f_{\pi_{0}^{*},\mu}\in \bar{\mathcal{F}}_{d}^{S}$ with $\pi_{0}^{*}=(\pi_{0,1}^{*},\ldots, \pi_{0,S}^{*})\in\Delta^{S}$ known, satisfy
\begin{itemize}
\item[(i)] $\sum_{s=1}^{S}\pi_{0,s}^{*}\int_{A^{g}}\phi(y;\mu_{n,g,s}^{*},\underline{q}I_{d})dy\geq 1/2$ or
\item[(ii)] $\|\mu_{n,g,s}^{*}\|_{\ell_{2}}=\infty$ for all $s\in\{1,\ldots,S\}$.
\end{itemize}
\end{proposition}

\begin{remark}
Scenario (ii) of Proposition \ref{propositionProjectLogConcave} can be ruled out in practice by initialising the optimisation scheme over a grid of points whose outer edges are dictated by the convex hull of the data.
\end{remark}

\subsection{Direct projection of $f_{n}$ onto a constrained mixture class}

The previous subsection illustrates the importance of performing the nonparametric structural filtering step prior to projection on the mixture class. The fact that the structure exploited in the filtering stage is transformed by the projection is not necessarily a limitation of the SPE, as the filtering step is simply a means to achieve appropriate concentration of mass and regularisation. In view of this, mispecification of the structure in the filtering stage is not of great concern.

If it is deemed important that the final estimate obeys some structural restrictions, an alternative way to proceed is to directly project $f_{n}$ onto $\mathcal{F}_{d}^{S,C}=\mathcal{F}_{d}^{S}\cap \mathcal{F}_{d}^{C}$, a structurally constrained subset of the mixture class. The final estimate will, by construction, possess a succinct parametric representation and obey the constraints.

\subsubsection{Estimation error versus approximation error}

Let $f_{0}\in\mathcal{F}_{d}^{C}$ but $f_{0}\notin\mathcal{F}_{d}^{S}$ and consider $P(\ell\circ \widetilde{f}_{n})-P(\ell\circ \widehat{f}^{*}_{n})$, where $\widehat{f}^{*}_{n}$ is as in Table \ref{tableNotation} and, for the squared error loss function $\ell$,
\[
P(\ell\circ g)=\int(f_{0}-g)^{2}dP \quad \text{and} \quad \widetilde{f}_{n}:=\arginf_{f\in\mathcal{F}^{S,C}}\frac{1}{n}\sum_{i=1}^{n}(\delta_{X_{i}}-f)^{2},
\]
with $\delta_{x}$ the point mass at $x$. A cursory theoretical analysis of the relative performance of the direct projection to that of the SPE is obtained through a standard decomposition into estimation error and approximation error:
\begin{eqnarray*}
& & P(\ell\circ \widetilde{f}_{n}) - P(\ell\circ \widehat{f}^{*}_{n})\\
&=& P(\ell\circ \widetilde{f}_{n}) - \inf_{f\in\mathcal{F}_{d}}P(\ell\circ f) + \inf_{f\in\mathcal{F}_{d}}P(\ell\circ f) - P(\ell\circ \widehat{f}^{*}_{n})\\
&=& \left(\Big[P(\ell\circ \widetilde{f}_{n}) - \inf_{f\in\mathcal{F}_{d}^{S,C}}P(\ell\circ f) \Big] + \Big[\inf_{f\in\mathcal{F}_{d}^{S,C}}P(\ell\circ f) - \inf_{f\in\mathcal{F}_{d}}P(\ell\circ f)\Big]\right) \\
& & \;\; - \left(\Big[P(\ell\circ \widehat{f}^{*}_{n}) - \inf_{f\in\mathcal{F}_{d}^{S}}P(\ell\circ f) \Big] +\Big[\inf_{f\in\mathcal{F}_{d}^{S}}P(\ell\circ f) - \inf_{f\in\mathcal{F}_{d}}P(\ell\circ f)\Big]\right)\\
&=& \text{I}+\text{II}-(\text{III}+\text{IV}) = (\text{I}-\text{III}) + (\text{II}-\text{IV}).
\end{eqnarray*}
It is always true that II$>$IV because $\mathcal{F}_{d}^{S,C}\subset \mathcal{F}_{d}^{S}$ and $f_{0}\notin \mathcal{F}_{d}^{S}$, therefore a fortiori $f_{0}\notin \mathcal{F}_{d}^{S,C}$. We see that $P(\ell\circ \widetilde{f}_{n}) - P(\ell\circ \widehat{f}^{*}_{n})<0$ if and only if -(I-III)$>$(II-IV), i.e.~if $\widetilde{f}_{n}$ achieves a smaller estimation error and if the difference in estimation error exceeds the difference in approximation error.

An example gives weight to our claim that the model class $\mathcal{F}^{S,C}$ is not sufficiently rich to well approximate many elements of $\mathcal{F}_{d}^{C}$. By contrast $\mathcal{F}_{d}^{S}$ is a rich model class which engenders small values of
\[
\sup_{f_{0}\in\mathcal{F}_{d}^{C}}\left(\inf_{f\in\mathcal{F}_{d}^{S}}P(\ell\circ f) - P(\ell\circ f_{0})\right),
\]
despite the fact that the $f\in\mathcal{F}_{d}^{S}$ do not obey the defining features of elements in $\mathcal{F}_{d}^{C}$.
	
Take $\mathcal{F}_{d}^{C}=\mathcal{F}_{d}^{LC}$ and $\mathcal{F}_{d}^{S}=\bar{\mathcal{F}}_{d}^{S}$ (cf.~Table 2). A twice continuously differentiable function $h: \mathbb{R}^{d}\longrightarrow \mathbb{R}$ is concave if and only if $-\nabla \nabla^{T} h(y)$ is positive semi-definite for all $y\in\mathbb{R}^{d}$ \citep[][Theorem 4.5]{Rockafellar1970}. An $f\in\mathcal{F}_{d}^{LC} \cap \bar{\mathcal{F}}_{d}^{S}$ is thus an element of $\bar{\mathcal{F}}_{d}^{S}$ satisfying $-\nabla \nabla^{T} \log f(y)$ positive semi-definite. To simplify notation, let $\phi_{s}(y):=\phi(y;\mu_{s},\underline{q}I_{d})$ for $y\in\mathbb{R}^{d}$. We have
\[
\nabla \phi_{s}(y)=-\phi_{s}(y)\left(\frac{y-\mu_{s}}{\underline{q}}\right)
\]
and
\begin{eqnarray*}
\nabla \nabla^{T} \phi_{s}(y)&=&-\left(\nabla \phi_{s}(y)\left(\frac{y-\mu_{s}}{\underline{q}}\right)^{T}+\phi_{s}\nabla\left(\frac{(y-\mu_{s})}{\underline{q}}\right)^{T}\right) \\
														 &=&\frac{\phi_{s}(y)}{\underline{q}}\left(\left(\frac{y-\mu_{s}}{\sqrt{\underline{q}}}\right)\left(\frac{y-\mu_{s}}{\sqrt{\underline{q}}}\right)^{T} - I_{d}\right).
\end{eqnarray*}
Substituting in
\begin{eqnarray*}
& & -\nabla \nabla^{T}(\log f)(y)\\
&=& \left(\frac{1}{f^{2}}\nabla f (\nabla f)^{T} - \frac{1}{f}\nabla \nabla^{T} f\right)(y) \\
														&=& \frac{1}{f^{2}(y)}\left[\left(\sum_{s=1}^{S}\pi_{s}\nabla\phi_{s}(y)\right)\left(\sum_{s=1}^{S}\pi_{s}\nabla\phi_{s}(y)\right)^{T}\right] - \frac{1}{f(y)}\left[\sum_{s=1}^{S}\pi_{s}\nabla \nabla^{T} \phi_{s}(y) \right] \\
														&=& \frac{1}{f^{2}(y)}\left[\left(\sum_{s=1}^{S}\pi_{s}\phi_{s}(y)\left(\frac{y-\mu_{s}}{\underline{q}}\right)\right)\left(\sum_{s=1}^{S}\pi_{s}\phi_{s}(y)\left(\frac{y-\mu_{s}}{\underline{q}}\right)\right)^{T}\right]\\
														& & \qquad  + \quad \frac{1}{f(y)}\left[\sum_{s=1}^{S}\frac{\pi_{s}}{\underline{q}}\phi_{s}(y)\left(I_{d}-\left(\frac{y-\mu_{s}}{\sqrt{\underline{q}}}\right)\left(\frac{y-\mu_{s}}{\sqrt{\underline{q}}}\right)^{T}\right)\right],
\end{eqnarray*}
we see that $f$ is in $\mathcal{F}_{d}^{LC}\cap \bar{\mathcal{F}}_{d}^{S}$ if and only if the eigenvalues of
\[
\frac{1}{f(y)}\left[\sum_{s=1}^{S}\frac{\pi_{s}}{\underline{q}}\phi_{s}(y)\left(I_{d}-\left(\frac{y-\mu_{s}}{\sqrt{\underline{q}}}\right)\left(\frac{y-\mu_{s}}{\sqrt{\underline{q}}}\right)^{T}\right)\right]
\]
are non-negative for all $y\in\mathbb{R}^{d}$. This constraint is easily imposed (at high computational expense) at the optimisation stage via a constraint-violation penalty in the objective function (see Section \ref{sectionNumerical} for details of implementation in the unconstrained case).

Take $f_{0}\in\mathcal{F}_{d}^{LC}$ as the density of the uniform distribution on the $r$-radius disk, i.e.
\[
f_{0}(y)=\left\{\begin{array}{cl} 
					(\pi r^{2})^{-1} & \text{if } y_{1}^{2}+y_{2}^{2}<r^{2} \\
					0 & \text{otherwise}.
\end{array} \right.
\]
The choice of this example is motivated by knowledge of the fact that log-concave densities are necessarily unimodal, whilst it is clear that the only unimodal members of $\bar{\mathcal{F}}_{d}^{S}$ are the mixtures whose component means are sufficiently close together, with the exact proximity depending on $\underline{q}$. The result is that, although multimodal members of $\bar{\mathcal{F}}_{d}^{S}$ produce a more uniform dispersion of mass over the support of $f_{0}$, the restriction that the estimate lies in $\bar{\mathcal{F}}_{d}^{S}\cap \mathcal{F}_{d}^{LC}$ leads to an extremely poor estimate of $f_{0}$, which concentrates its mass over a small region of the support of $f_{0}$. Figure 3.~depicts a typical realisation of the log-concave mixture projection and the SPE based on a log-concave MLE pilot \citep{SamworthReadPaper}.

\vspace{8pt}

{\centering
(A)\includegraphics[trim=1.5in 4.0in 1.5in 4.2in, clip, height=0.20\paperwidth]{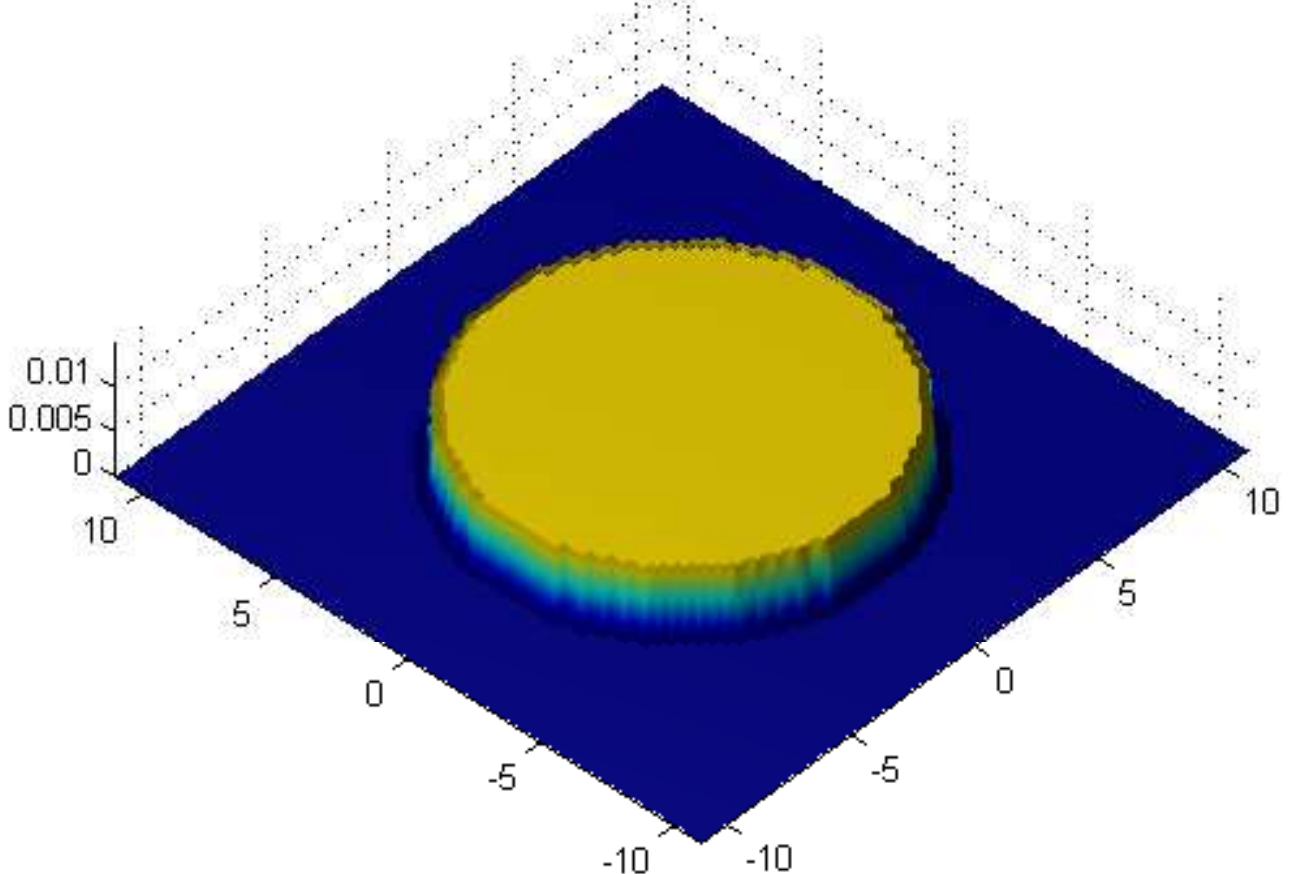}
(B)\includegraphics[trim=1.5in 4.0in 1.8in 4.2in, clip, height=0.20\paperwidth]{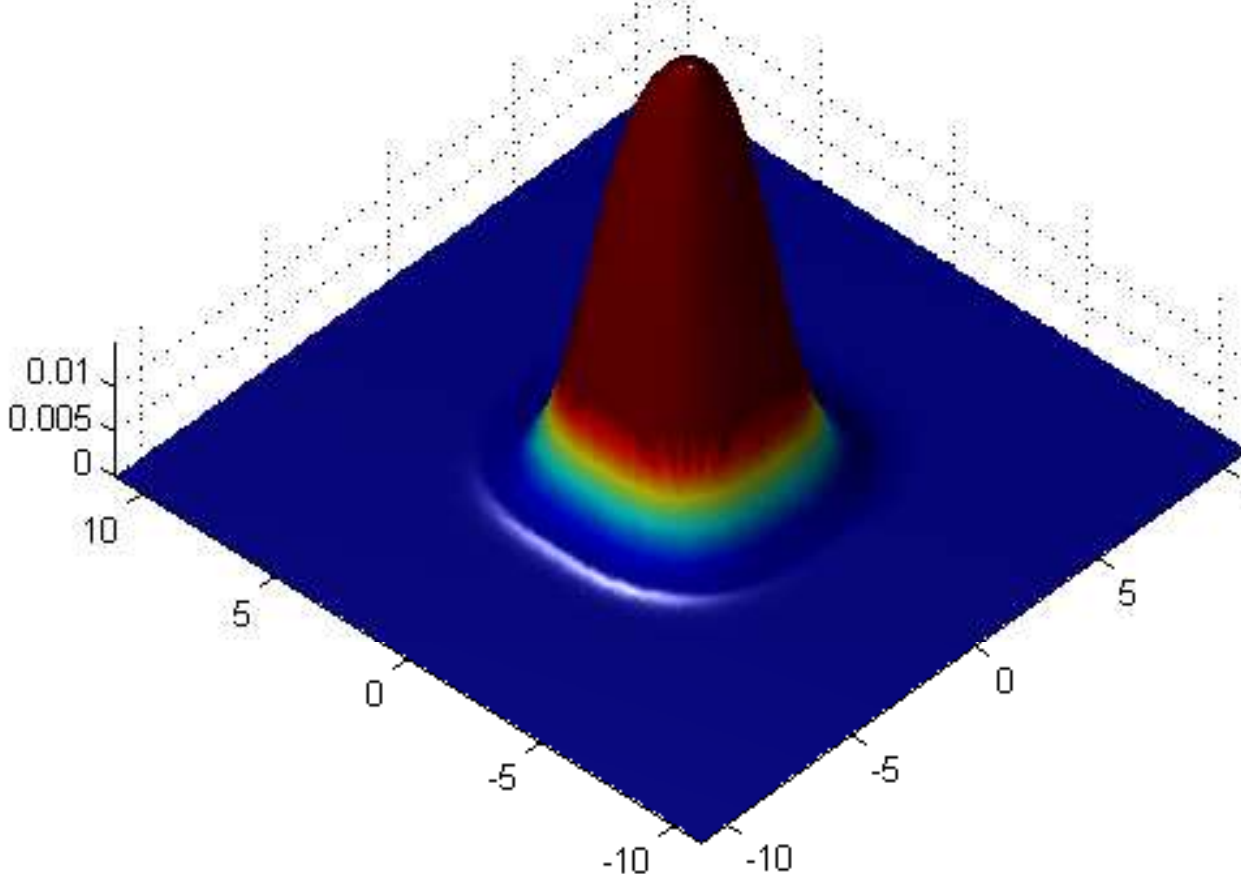}\\
(C)\includegraphics[trim=1.5in 4.0in 1.5in 4.2in, clip, height=0.20\paperwidth]{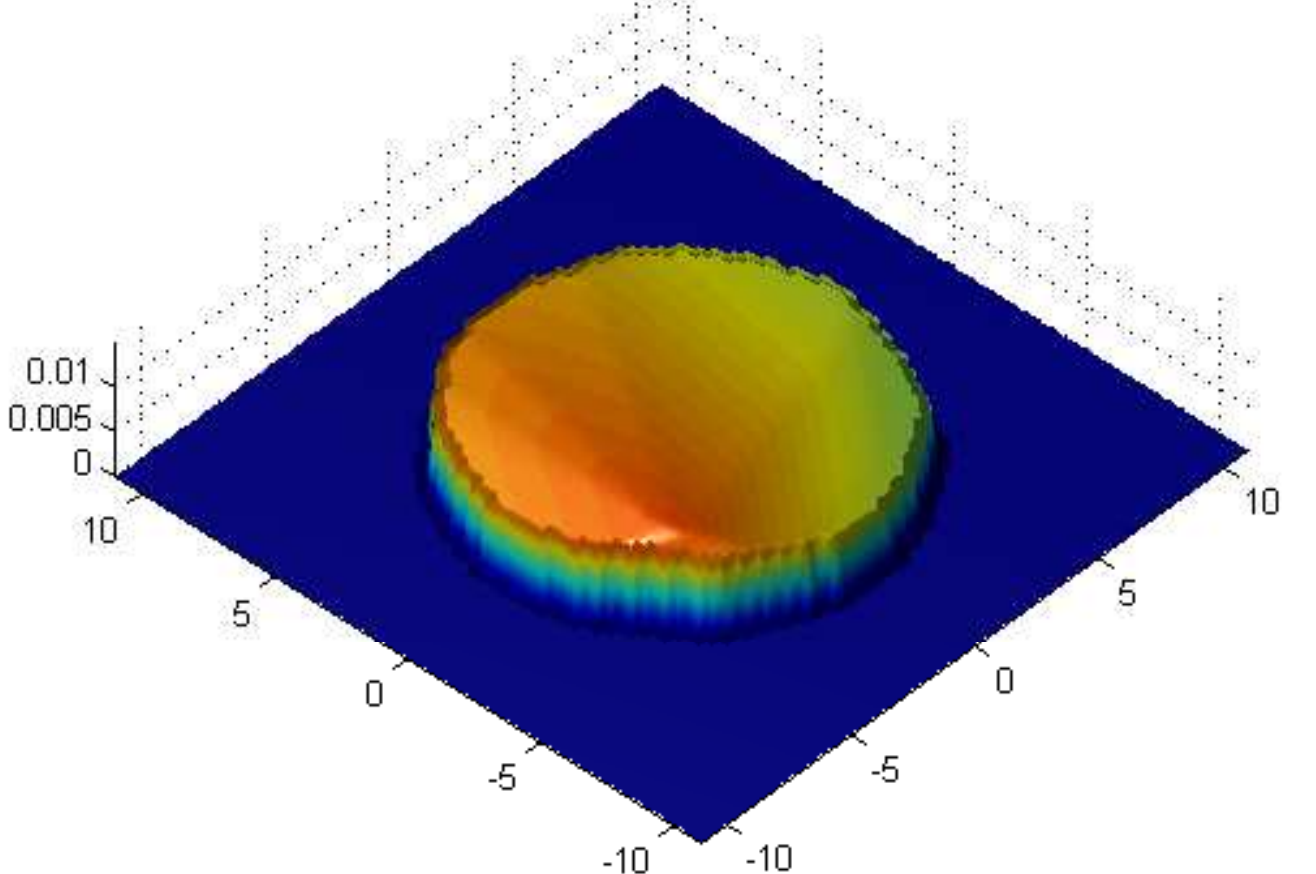}
(D)\includegraphics[trim=1.5in 4.0in 1.8in 4.2in, clip, height=0.20\paperwidth]{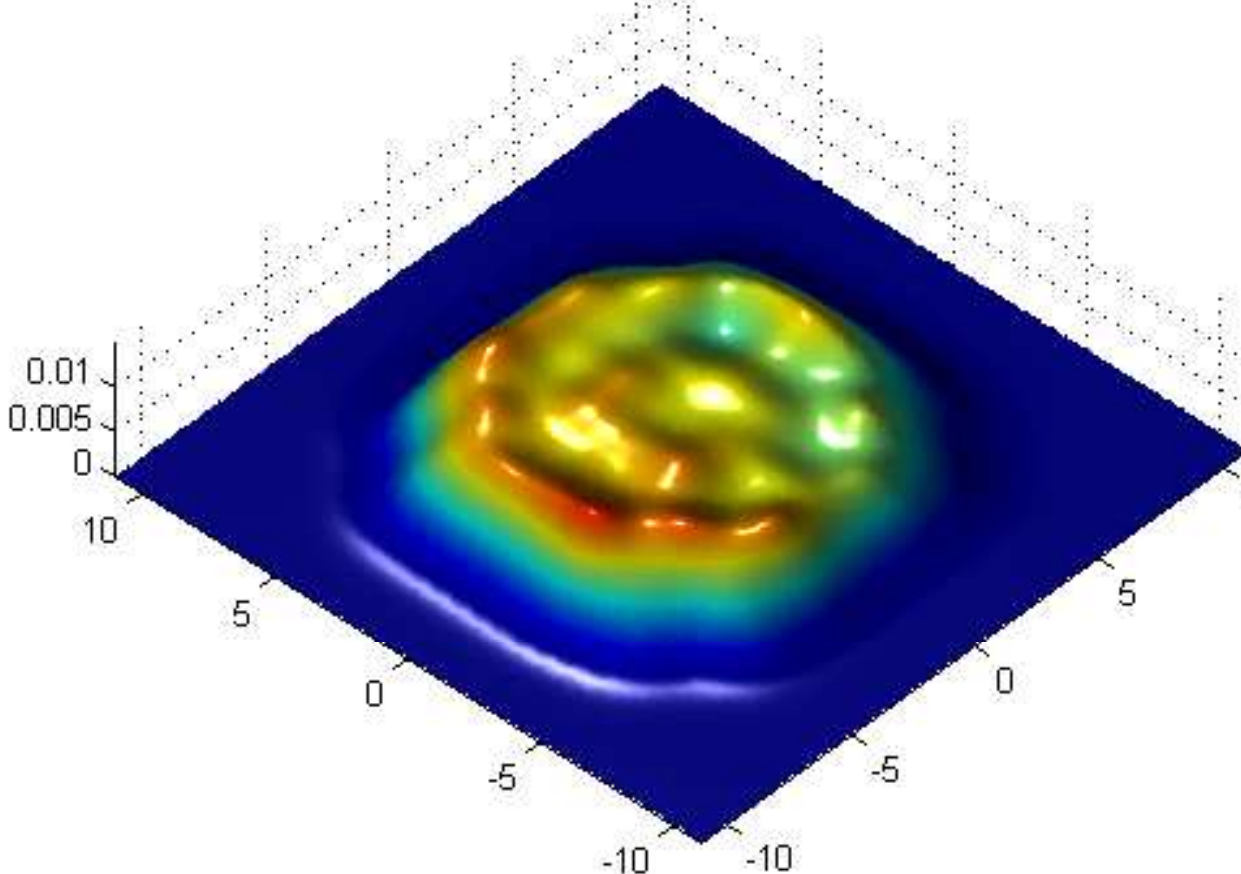}

}

\noindent \small{\textbf{Figure 3.} (A) true density (density of the uniform distribution on the radius 7 disk); (B) projection of $f_{n}$ onto the log-concavity constrained mixture class with 49 mixture components; (C) log-concave MLE; (D) SPE with 49 mixture components.}

\subsubsection{Further practical concerns}

\normalsize{Although the log-concavity condition leads to a relatively simple constraint set, other instances of $\mathcal{F}^{S,C}_{d}=\mathcal{F}_{d}^{S}\cap\mathcal{F}_{d}^{C}$ are extremely difficult to establish and computationally demanding to implement. Consider, for example $\mathcal{F}_{d}^{C}$ the class of density functions with a particular conditional independence structure; imposition of such conditional independence structure entails a highly nonlinear set of constraints composed of marginalised conditional mixture densities. Conditional independence structure imposed through a nonparametric filter is not preserved through projection, but nevertheless provides regularisation in multi-dimensional density estimation problems.}

Sampling from $\widehat{f}_{n}^{*}$ is particularly simple due to its constitution as a finite mixture of parametric densities; sampling from parametric densities is straight forward in view of the probability integral transform. By contrast, sampling from nonparametric estimators is typically difficult; consider for instance the log-concave MLE, which involves an accept-reject algorithm \citep[see][Appendix B.3]{SamworthReadPaper}.

\section{Numerical results} \label{sectionNumerical}

This section examines a particular concrete example from the class of SPEs. $\mathcal{F}_{d}^{S}$ is taken to be $\bar{\mathcal{F}}_{d}^{S}$ and $\widehat{f}^{P}$ is chosen from a large set of nonparametric densities, to be discussed below.

\subsection{Simulation performance}\label{sectionSimulation}

\normalsize{The first simulation study comprises 1000 pseudo random samples of size $n=50, 100, 250, 500$ from $f_{0}$, where $f_{0}$ is a bivariate density of one of the following forms:} 
\begin{itemize}
\item[(i)] Normal location-scale mixture density $\pi_{1}\phi_{\mu_{1},\Sigma_{1}}+\pi_{2}\phi_{\mu_{2},\Sigma_{2}}$ with $\pi_{1}=\pi_{2}=1/2$, $\mu_{1}=(1,2)^{T}$, $\mu_{2}=(-1,1)^{T}$, $\Sigma_{1}=[2,\hspace{1pt} -0.5;-0.5,\hspace{1pt} 1.5]$, $\Sigma_{2}=[4, \hspace{2pt} 0.9;0.9, \hspace{2pt} 1.5]$;
\item[(ii)] The bivariate density with independent $\Gamma(2,1)$ marginals; 
\item[(iii)] $\sum_{j=1}^{500}\pi_{j}\phi_{2,\mu_{j},0.7I}$, where $\pi_{1}=\pi_{2}=\ldots=\pi_{500}=1/500$ and $\mu_{j}=(4\cos(t_{j}), 4\sin(t_{j}))^{T}$, where $t_{1},\ldots,t_{500}$ are equally spaced points in $[0, 2\pi]$. 
\end{itemize}
(i) is a skewed unimodal density exhibiting dependence, (ii) is a log-concave density, and (iii) is a density concentrated on a non-convex domain. 

The algorithm used to estimate the weights and the parameters of the mixture representation is described in Section \ref{sectionAlgo}

\normalsize{Fig.~4 displays estimates resulting from a draw of size $n=50$ from the density of (iii), which is displayed in panel (A) of Fig.~1. (A) is the log concave maximum likelihood estimate (LCMLE) \citep{SamworthReadPaper}, (B) is the projection of the LCMLE on $\bar{\mathcal{F}}_{2}^{S}$, with scale parameter $\underline{q}=0.7$ and $S=36$, (C) is the least-squares cross-validated (LSCV) bandwidth KDE, computed in \verb+R+ using the \verb+ks+ package \citep{Duong} and (D) is the projection of the LSCV KDE.}

\vspace{8pt}

{\centering
(A)\includegraphics[trim=2.0in 3.8in 2.0in 3.8in, clip, height=0.25\paperwidth]{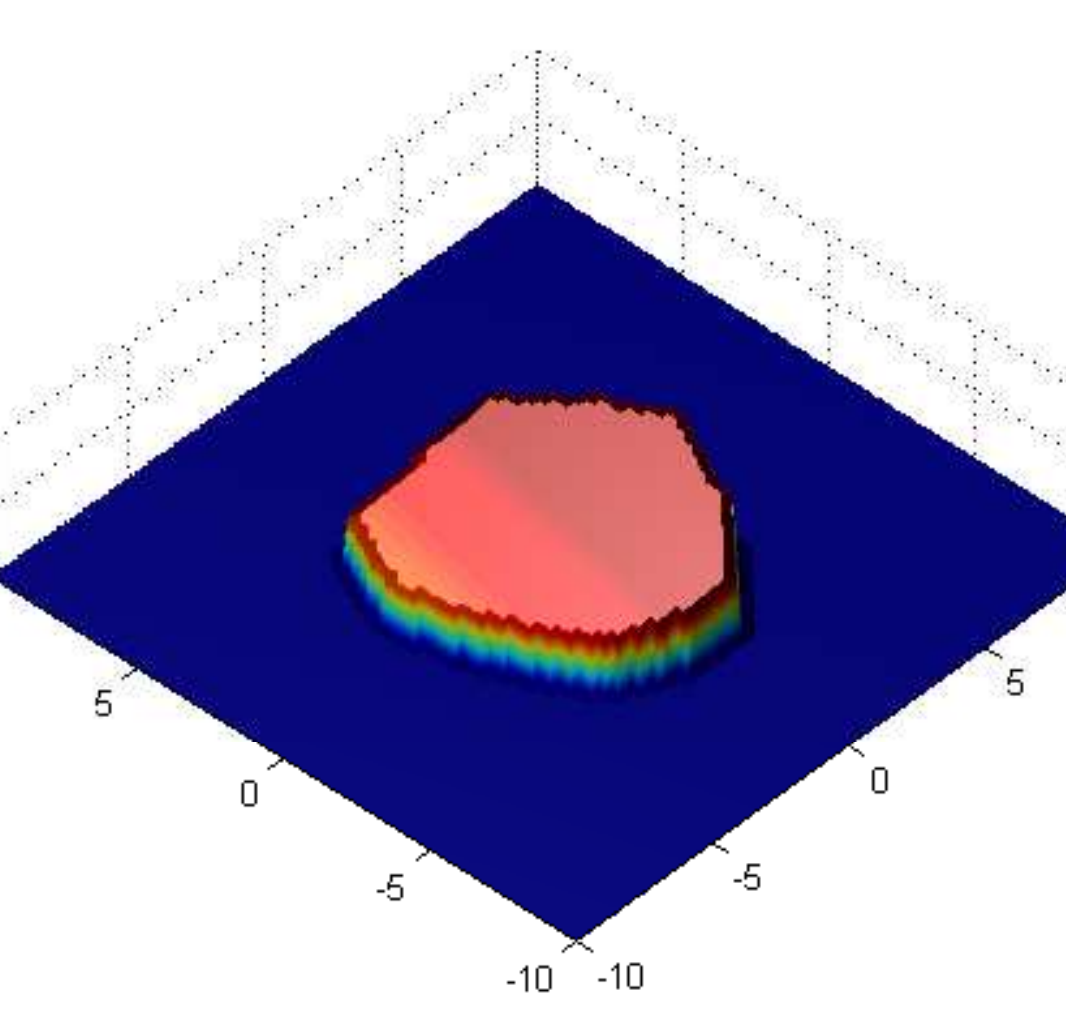} 
(B)\includegraphics[trim=2.0in 3.8in 2.0in 3.8in, clip, height=0.25\paperwidth]{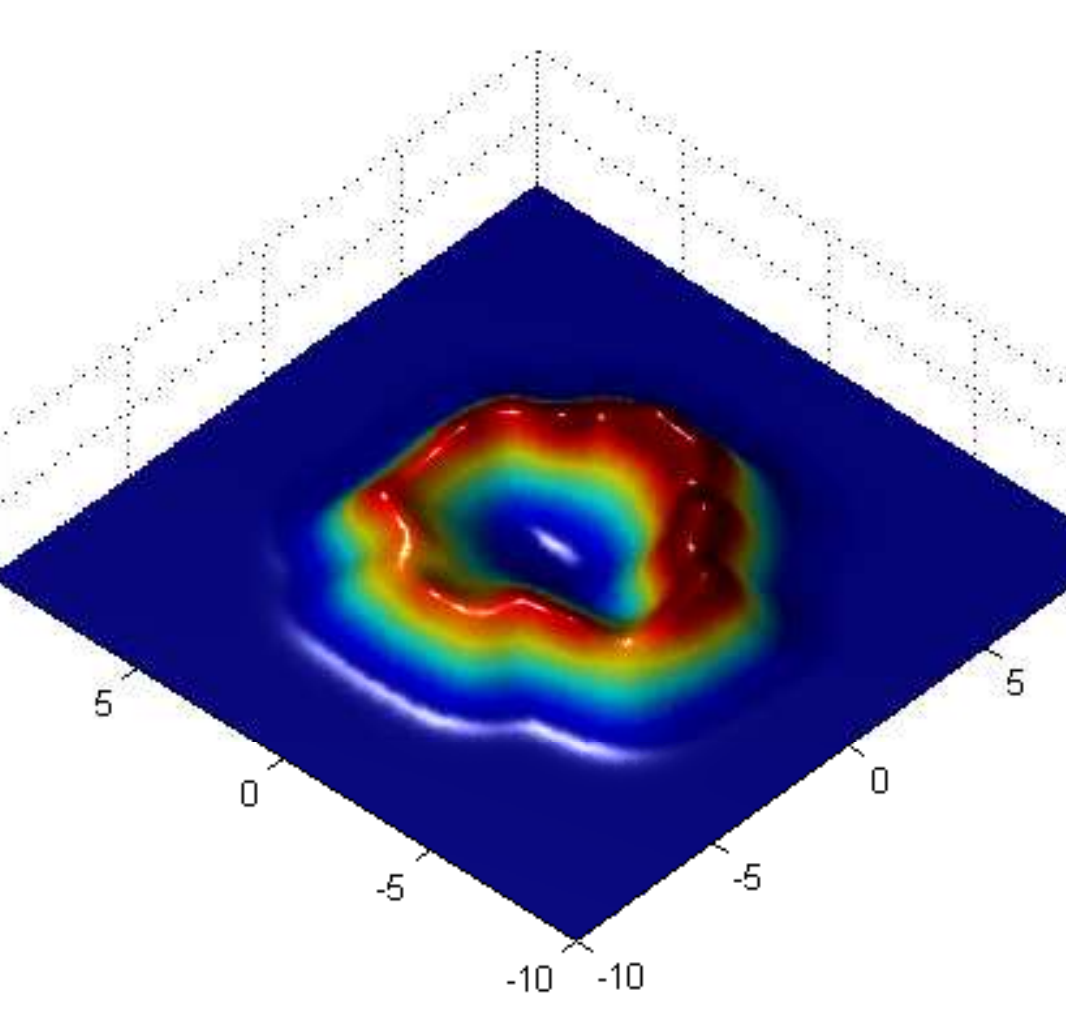} \\
(C)\includegraphics[trim=2.0in 3.8in 2.0in 3.8in, clip, height=0.25\paperwidth]{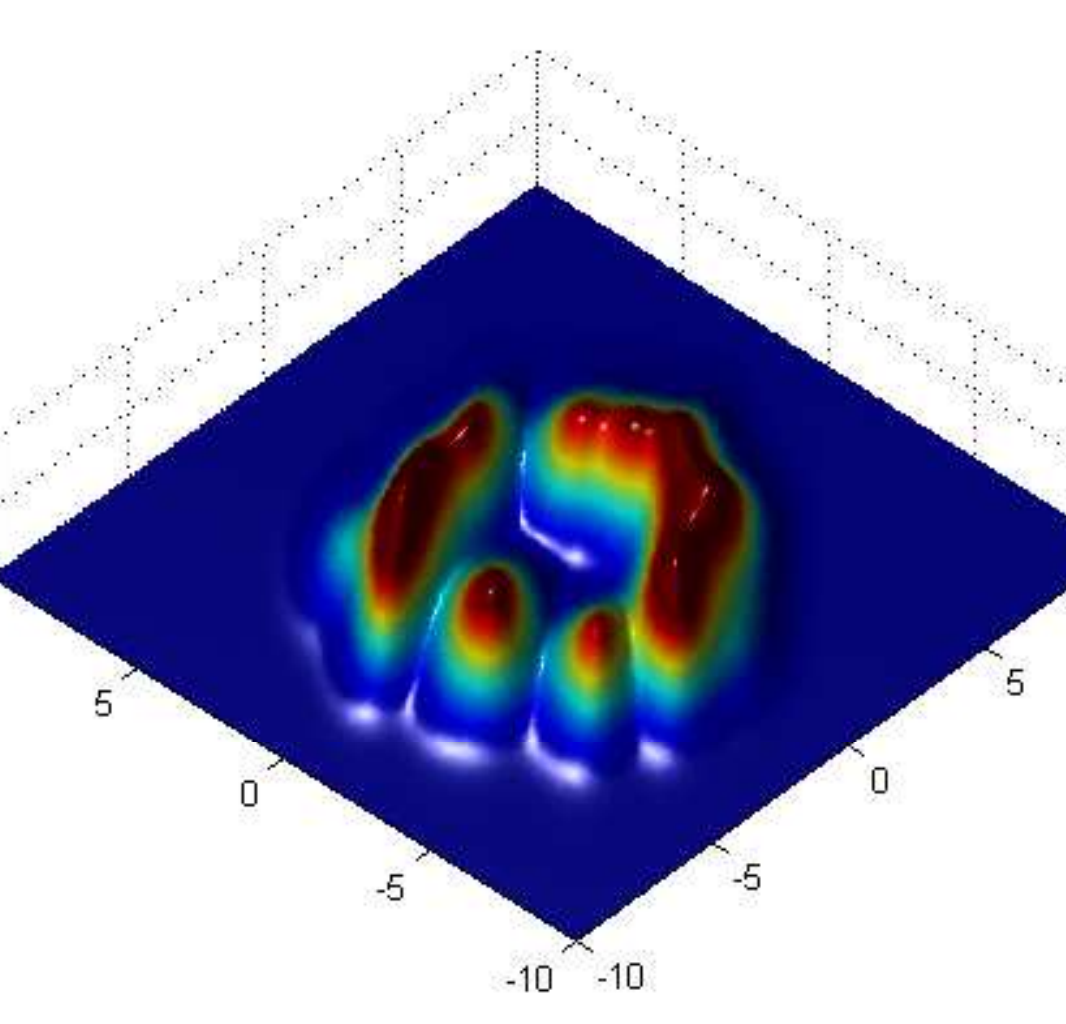} 
(D)\includegraphics[trim=2.0in 3.8in 2.0in 3.8in, clip, height=0.25\paperwidth]{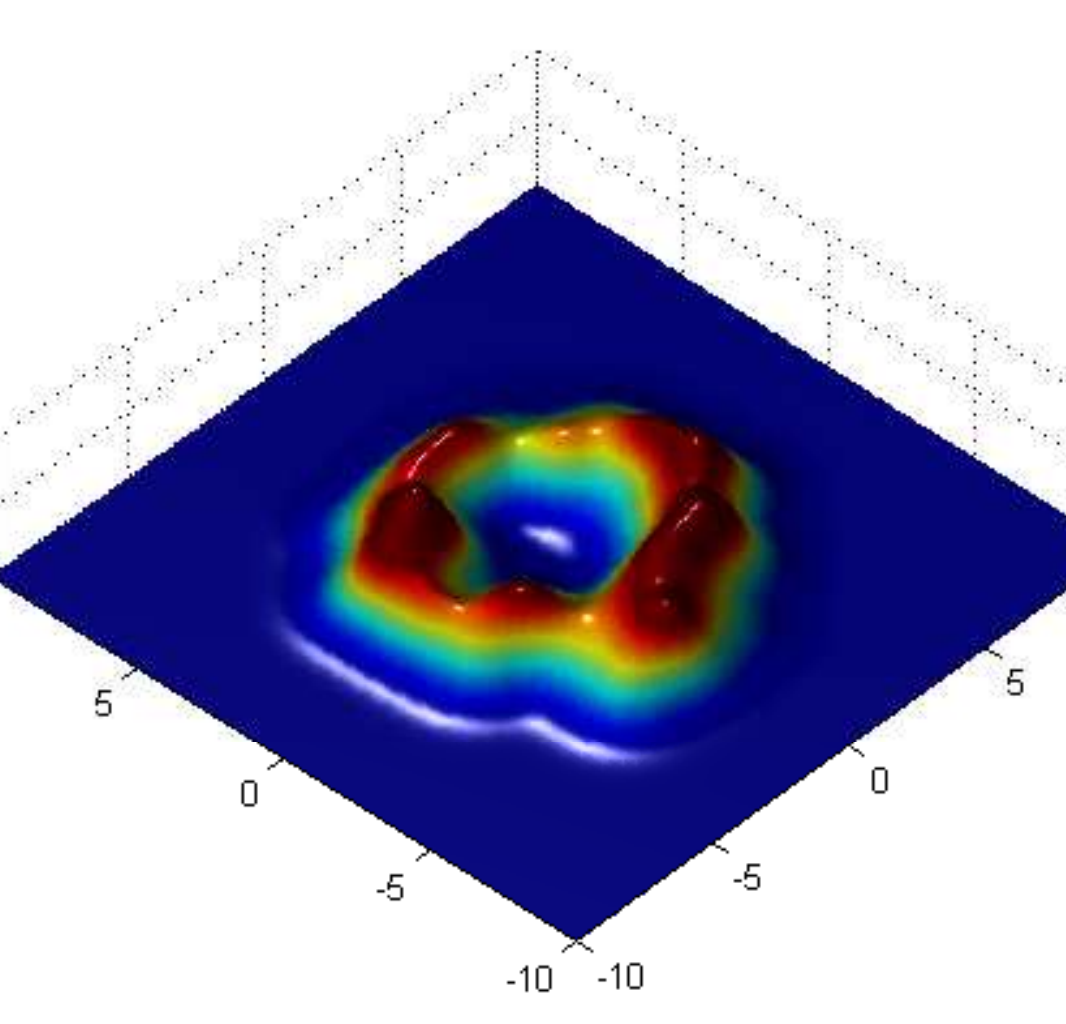}

}
\vspace{4pt}
\noindent \small{\textbf{Figure 4.} Estimates of the density in Fig.~2 (A) constructed as described above.}

\vspace{8pt}

\normalsize{Figures 5-7 illustrate the behavior of the integrated squared error (ISE) with increasing $n\in\{50,100,250,500\}$ for the density of (i)-(iii) respectively. We consider (1) the KDE with LSCV bandwidth (`K-CV'); (2) the projection of K-CV on $\bar{\mathcal{F}}_{2}^{S}$ with $S=36$ and $\underline{q}=0.7$ (K-CVProject); (3) the KDE with 2-stage PI bandwidth (`K-PI'); (4) the projection of K-PI (K-PIProject); (5) the log-concave MLE (`Lcd'); (6) the projection of Lcd (LcdProject); (7) the histogram (`Hist1') with coordinate-wise bin widths $(IQ)_{j}n^{-1/4}$ where $(IQ)_{j}$ is the inter-quartile range of $Y_{j}$; (8) the projection of (7) (`HistProject1'); (9) the perturbed histogram (`pHist1', see below for further details); (10) the projection of (9) (`pProject1'); (11)-(14) as (7)-(10) but with the coordinate-wise bin widths taken as $2(IQ)_{j}n^{-1/4}$; (15) the EM implementation of the Kiefer-Wolfowitz estimator with the same fixed scaling matrix $\underline{q}I_{d}$ across all location mixture components. The so called perturbed histogram is a mixture of 5 histograms, one with anchor point zero and the others with anchor points that are small perturbations from zero.}

For this two-dimensional scenario, EM is a competitive adversary, achieving a strong performance in terms of ISE over all data generating mechanisms whilst yielding a succinct parametric representation of the density. The SPE with histogram pilot estimator is also universally unintimidated over the various experiments.


{\centering
\includegraphics[trim=1.1in 0.0in 0.0in 3.5in, clip, height=0.39\paperwidth]{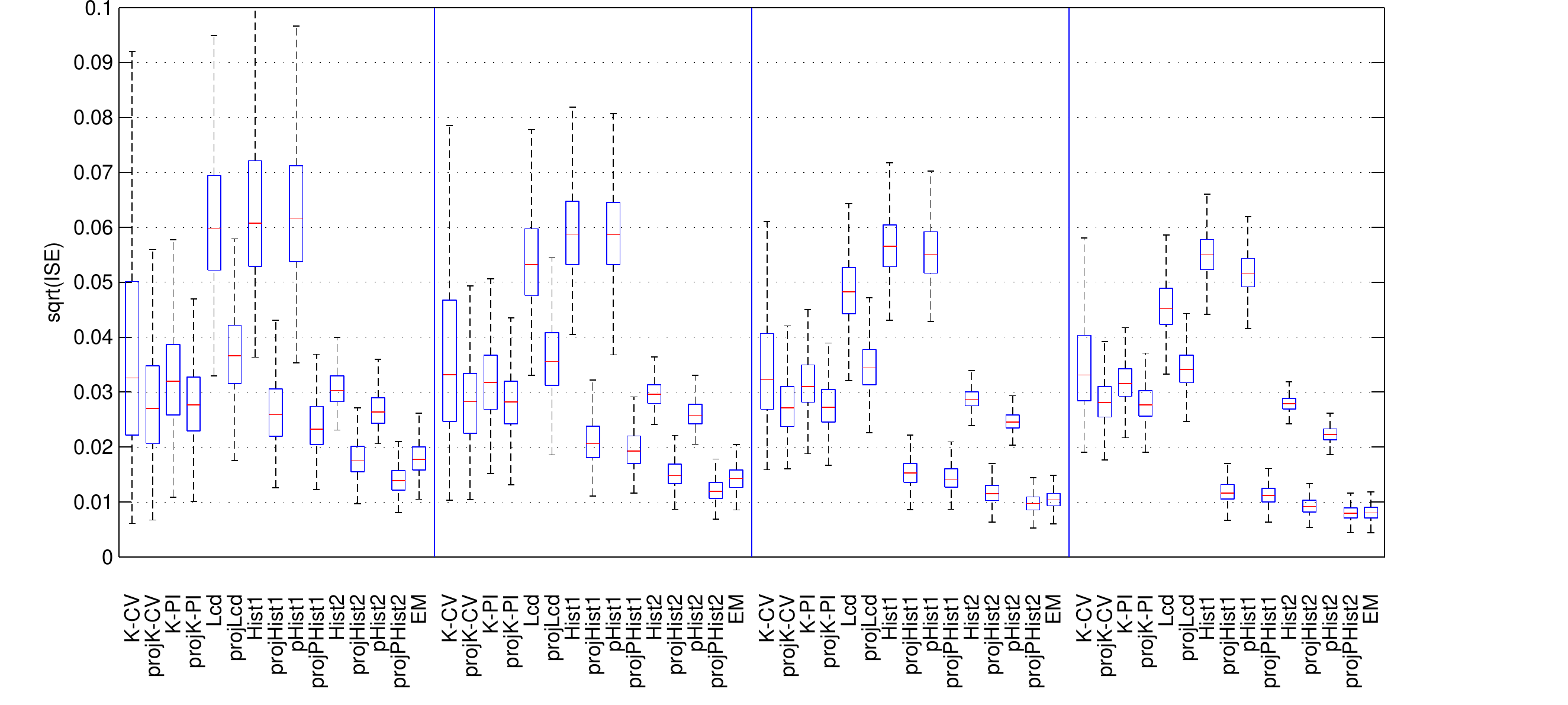} 

}
\noindent \small{\textbf{Figure 5.} Case (i): location-scale normal mixture. Left segment to right segment correspond to $n\in\{50, 100, 250, 500\}$.}

\bigskip

{\centering
\includegraphics[trim=1.1in 0.0in 0.0in 3.5in, clip, height=0.39\paperwidth]{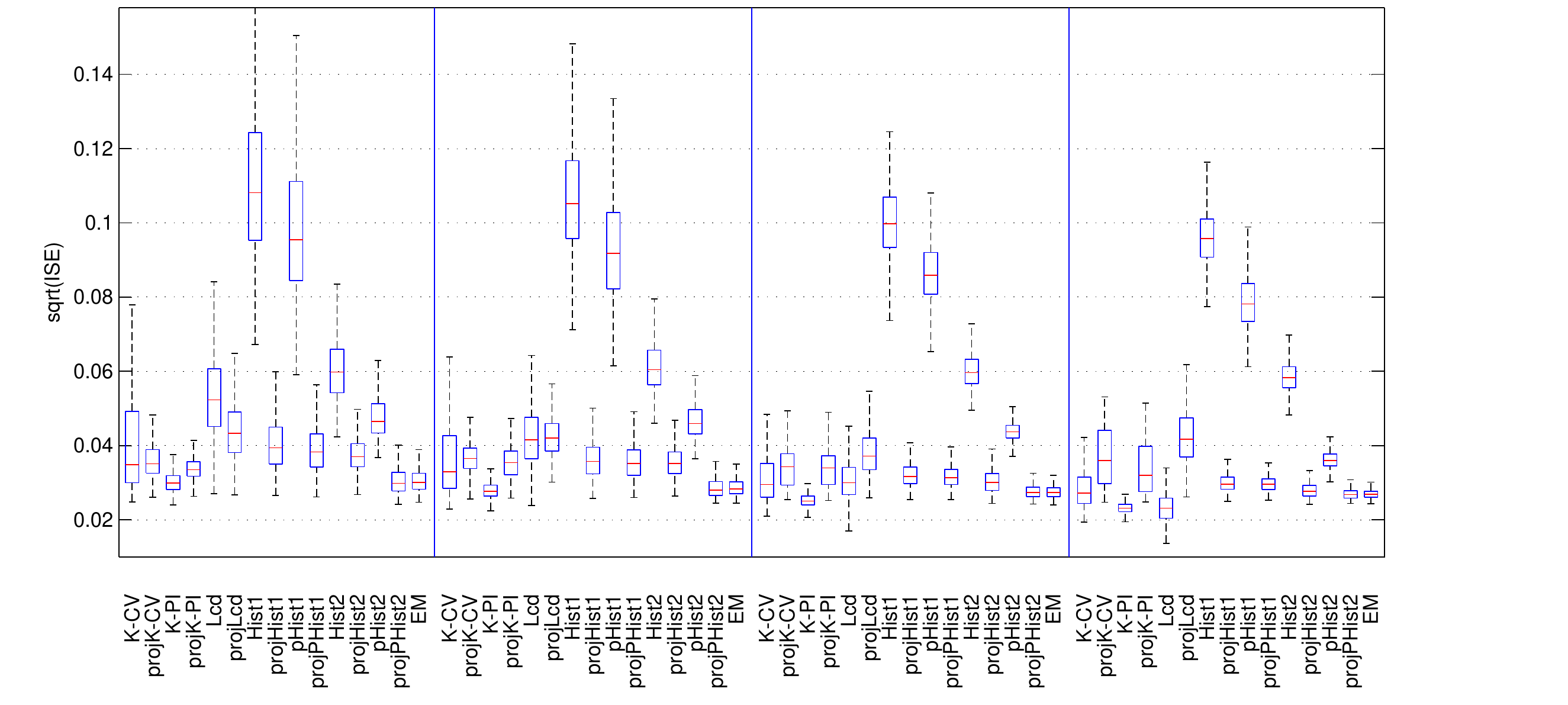} \\

}
\noindent \small{\textbf{Figure 6.} Case (ii): gamma. Left segment to right segment correspond to $n\in\{50, 100, 250, 500\}$.}

\bigskip

{\centering
\includegraphics[trim=1.1in 0.0in 0.0in 3.5in, clip, height=0.39\paperwidth]{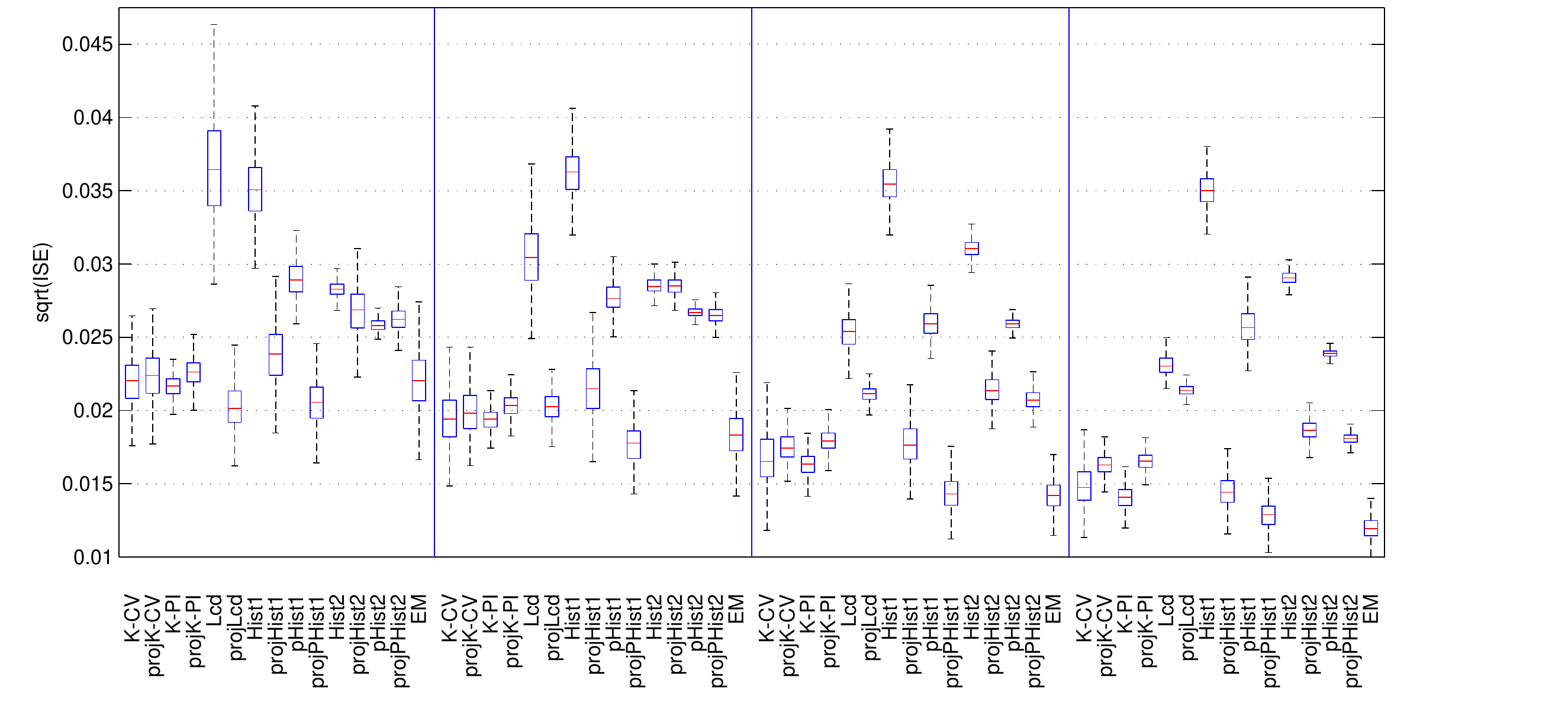} \\

}
\noindent \small{\textbf{Figure 7.} Case (iii): ring. Left segment to right segment correspond to $n\in\{50, 100, 250, 500\}$.}

\vspace{8pt}

\subsubsection{A five-dimensional case with graphical information}

\normalsize{In this section, we illustrate the benefit of exploiting structure in the pilot estimation stage, even though the structure is not ultimately enforced on the final estimate. Over 1000 Monte-Carlo replications, we generate $n\in\{50,100,250,500\}$ observations from a five-dimensional Gaussian graphical model with zero mean vector and inverse covariance matrix}
	\[A=\left[\begin{array}{ccccc} 
		3 & 0 & 0 & 0 & 0 \\
		0 & 5 & 0 & 1 & -1 \\
		0 & 0 & 2 & 0 & 0 \\
		0 & 1 & 0 & 2 & 0 \\
		0 & -1 & 0 & 0& 2
	\end{array}\right]	
	\]
We consider two pilot estimators for the five dimensional density $f_{0}$: a five dimensional histogram estimator (`Hist') and a histogram estimator that exploits the graphical structure in the sample (`graphHist'), constructed as
\[
\widehat{f}^{H}_{Y}=\widehat{f}^{H}_{Y_{5}|B_{2}}\widehat{f}^{H}_{Y_{4}|B_{2}}\widehat{f}^{H}_{Y_{2}}\widehat{f}^{H}_{Y_{1}}\widehat{f}^{H}_{Y_{3}}
\]
where $\widehat{f}^{H}_{Y_{j}|B_{k}}(y_{j}|y_{k}\in B_{k})$ with $B_{k}$ an arbitrary bin for the $k^{th}$ variable, where the bin width is taken as $2(IQ)_{k}n^{-1/2d}=2(IQ)_{k}n^{-1/10}$ where $(IQ)_{k}$ is the inter-quartile range of $Y_{k}$.

The SPE outperms EM for almost every combination of pilot estimator, $n$, and $S$ (see Fig.~8). Moreover, the advantage of exploiting the graphical structure in the nonparametric pilot estimation stage is clearly visible, with the SPE based on the graphical histogram strongly outperforming the SPE based on the agnostic histogram estimator in all situations, despite the agnostic histogram estimator itself having a stronger performance than the graphical histogram. The LSCV KDE performs extremely poorly in this high dimensional setting so is not reported, whilst the five-dimensional log-concave MLE and plug-in KDE are too computationally intensive for testing on a standard desktop computer and are therefore not considered.

{\centering
\includegraphics[trim=1in 3.0in 1in 0.3in, clip, height=7cm, width=10.5cm]{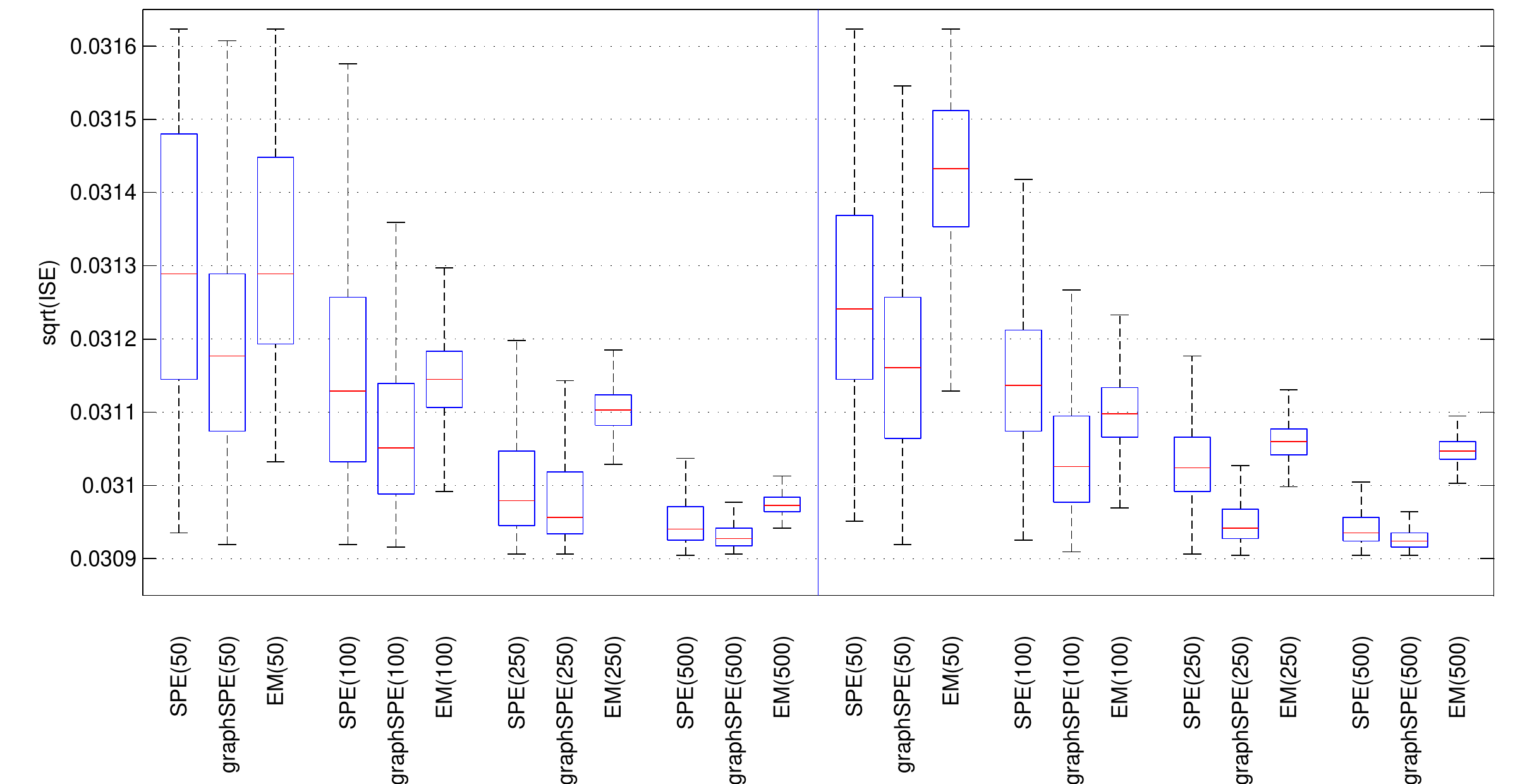}
\includegraphics[trim=1.2in 3.2in 0.5in 3.5in, clip, height=7cm, width=5.7cm]{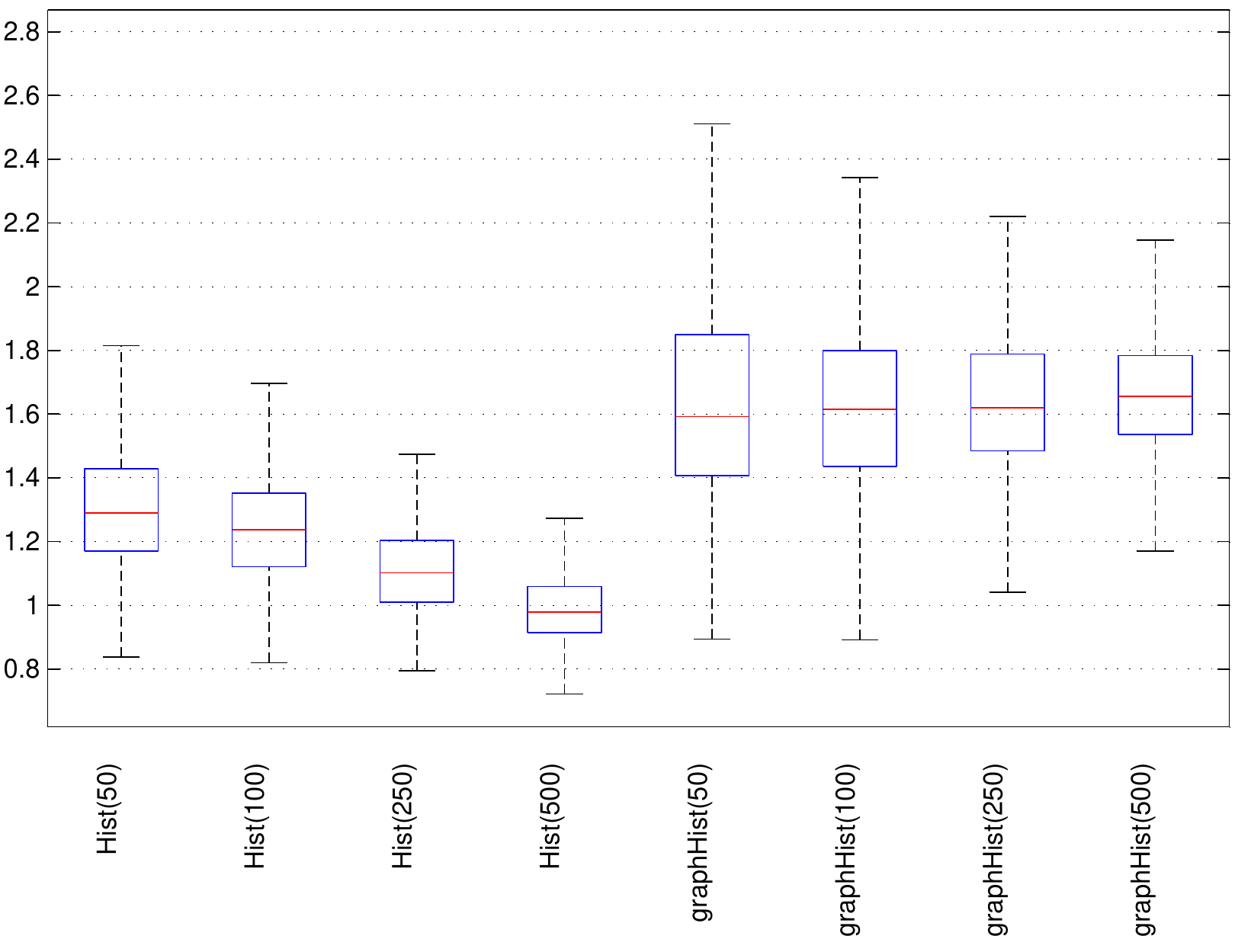}

}
\noindent \small{\textbf{Figure 8.} Left: performance (in terms of square root ISE) of EM, SPE with histogram pilot and SPE with graphical histogram pilot. Left panel corresponds to $2^5=32$ mixture components; right panel corresponds to $3^5=243$ mixture components. Right: Performance (in terms of square root ISE) of the histogram and graphical histogram before projection.}

\subsubsection{Least squares alternation algorithm}\label{sectionAlgo}

\normalsize{Projection of the histogram onto the class of spherical Gaussian mixtures involves solving the non-convex minimisation problem of \eqref{eqFiniteMinProblem}. We estimate $\pi_{1},\ldots, \pi_{S}$ and $\mu_{1},\ldots,\mu_{S}$ using a least-squares alternation algorithm, similar to that proposed by \citet{Yuan}.  The algorithm iteratively minimises \eqref{eqFiniteMinProblem} with respect to the mixing proportions, with the $\{\mu_{s}: s=1,\ldots,S\}$ held fixed (at their estimated values, or an initialisation in the first iteration), and then with respect to the $\{\mu_{s}: s=1,\ldots,S\}$ with the mixing proportions held at their estimated values at the previous iteration. These steps are repeated until convergence. With one set of unknowns held fixed, the minimisation is performed efficiently with standard quadratic program solvers. We suggest taking initial $\pi_{1},\ldots, \pi_{S}$ as the center point of the unit $S$-simplex, i.e.~$(1/S,\ldots, 1/S)$. The solver for $\mu_{1},\ldots,\mu_{S}$ also requires a set of starting values for $\mu_{1},\ldots,\mu_{S}$ in the first iteration; close inspection of the proof of Lemma \ref{lemmaWeakConvergence} reveals that a judicious choice of initial $\mu_{1},\ldots, \mu_{S}$ is to take them equally spaced in $\mathcal{M}=[-M, M]^{d}$. We may use this to guide our choice of $S$, as taking an $S$ that possesses an integer-valued $d^{th}$ root allows us to place the $\mu_{1},\ldots,\mu_{S}$ on a regular grid over an arbitrary pre-defined boundary.} 

\subsubsection{Robustness to choice of tuning parameters}

\normalsize{Figure 9 is based on 100 draws from the density in scenario (iii) for increasing bin widths in a histogram pilot estimation step. Whilst the performance of the perturbed histogram is affected substantially by the choice of bin width, the performance of its projection onto $\bar{\mathcal{F}}_{2}^{S}$ (with $\underline{q}=0.7$ and $S=64$) is substantially more stable. As highlighted in our theoretical results, the bandwidth that is optimal for pilot density estimation is not optimal for estimation with the SPE. More specifically, a degree of undersmoothing is required, as is visible in Figure 9.}

\vspace{8pt}

{\centering
\includegraphics[trim=1.1in 0.0in 0.0in 3.5in, clip, height=0.37\paperwidth]{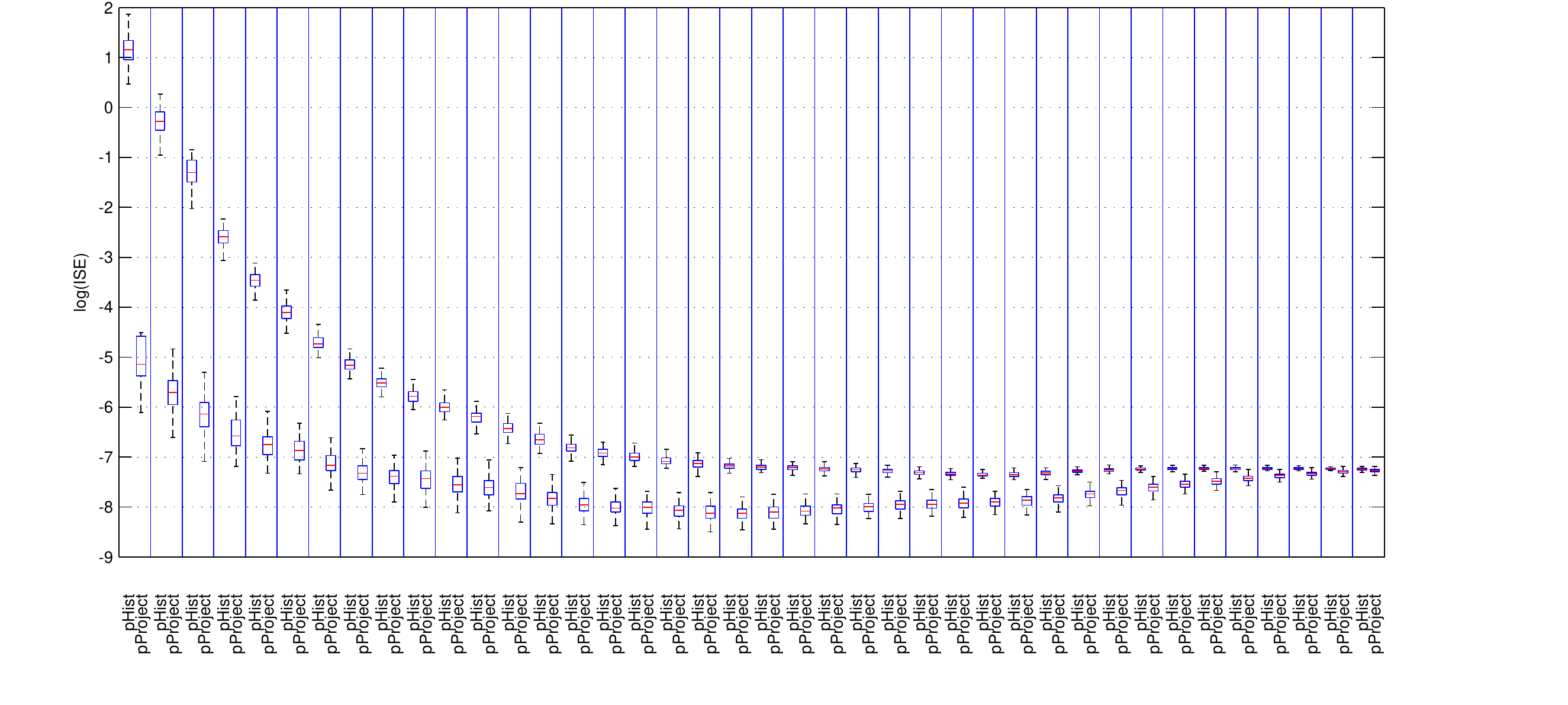}

}

\noindent \small{\textbf{Figure 9.} Robustness to changes in the histogram bin width. From left segment to right segment: increasing constant $c$ in the histogram bin width $h_{n,j}=c(IQ)_{j}n^{-1/2d}$ from 0.05 to 2 in the ring example for a fixed sample size of $n=100$.}

\vspace{8pt}

\normalsize{To illustrate the degree of sensitivity of the SPE to different choices of $\underline{q}$, Fig.~10 provides a heatmap of the log mean ISE, the log median ISE and the log of the variance of the ISE, computed from 100 Monte Carlo replications for each bin width parameter and for each value of $\underline{q}$. Bin widths are increased in 0.05 intervals from 0.05 to 2 and $\underline{q}$ is increased in 0.05 intervals from 0.4 to 2. The red dot indicates the minimum.}

\vspace{8pt}

{\centering
\includegraphics[trim=1.5in 3.5in 1.5in 3.9in, clip, height=0.2\paperwidth]{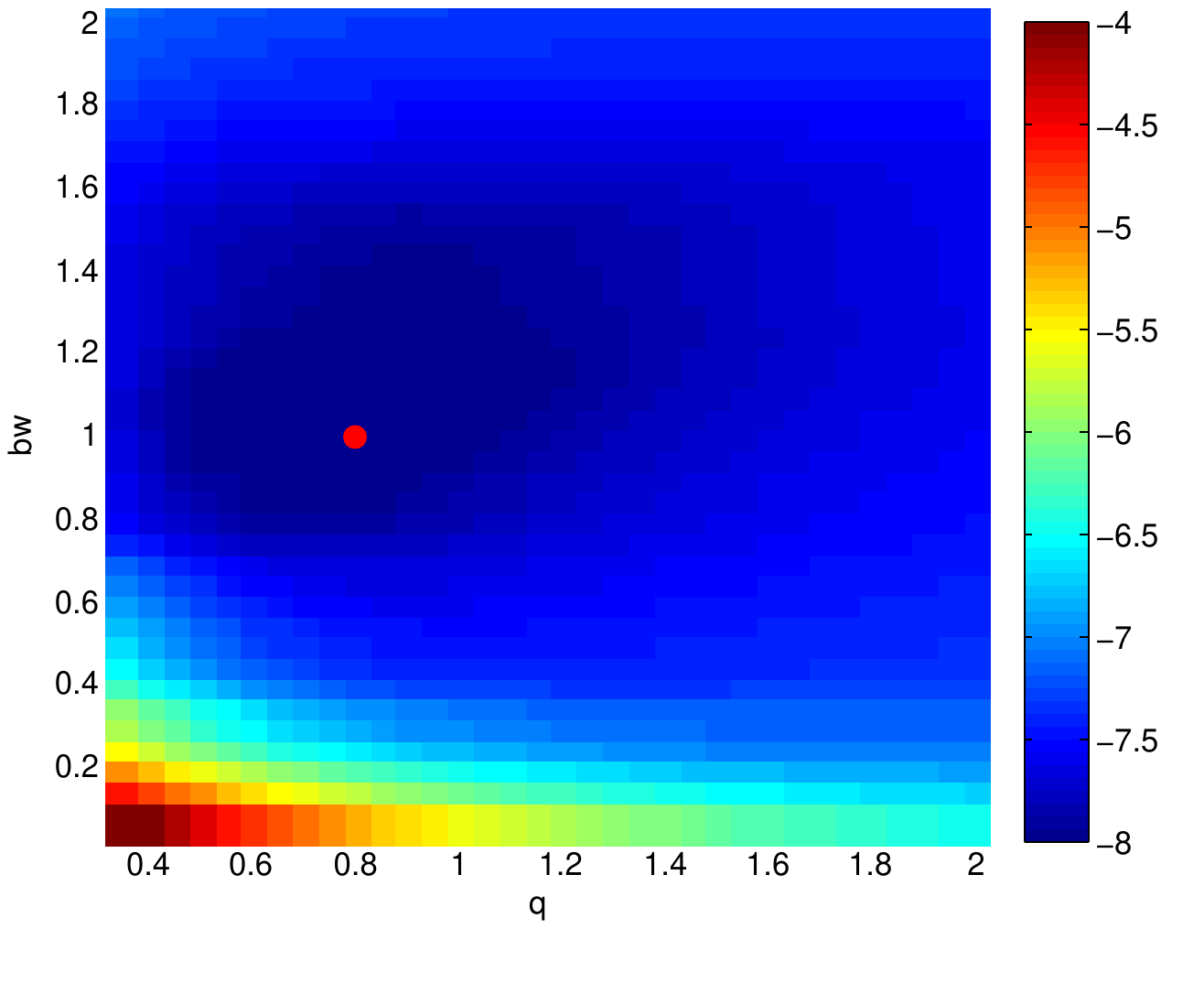}\includegraphics[trim=1.5in 3.5in 1.5in 3.9in, clip, height=0.2\paperwidth]{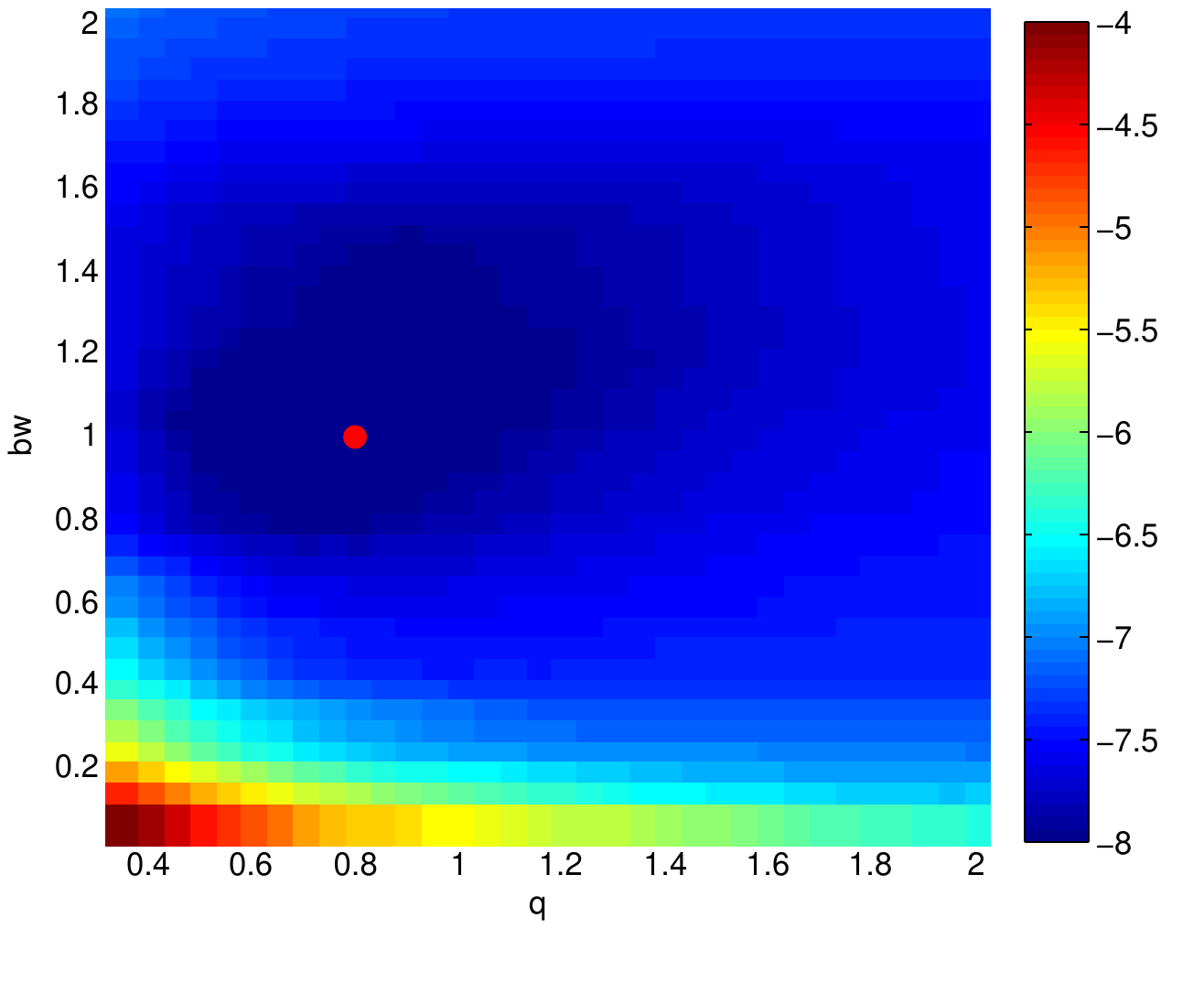}\includegraphics[trim=1.5in 3.5in 1.5in 3.9in, clip, height=0.2\paperwidth]{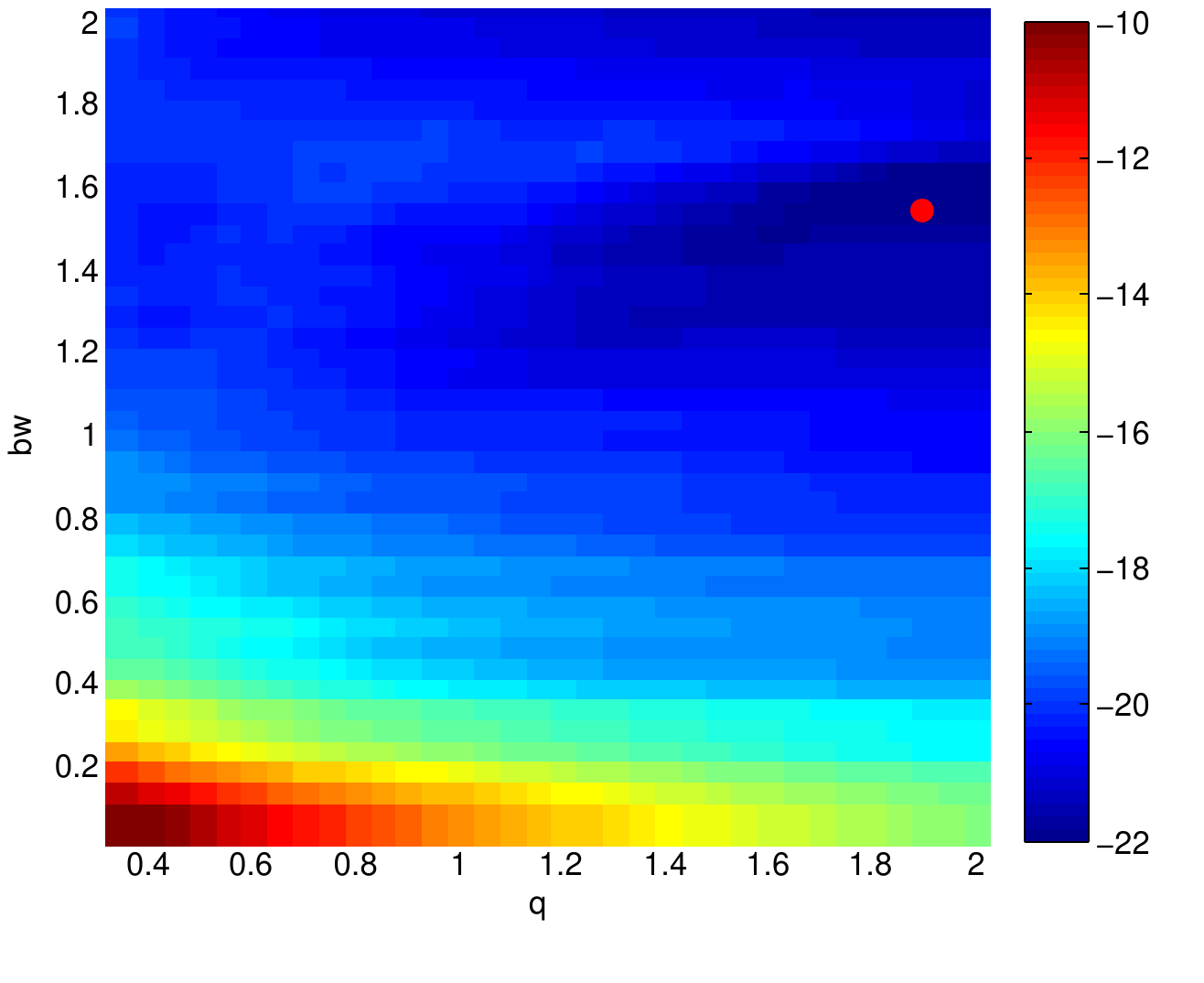}

}

\noindent \small{\textbf{Figure 10.} Joint sensitivity to histogram bin width and $\underline{q}$. From left to right: mean(ISE), median(ISE), var(ISE).}

\subsection{Real data example}\label{sectionRealData}

\normalsize{A dataset that is commonly considered in the context of multivariate density estimation \citep[see e.g.][]{Liebscher, SamworthReadPaper} is the Wisconsin breast cancer (diagnostic) dataset, publically available on the UCI Machine Learning Repository website:} \newline \href{http://archive.ics.uci.edu/ml/datasets/Breast+Cancer+Wisconsin+\%28Diagnostic\%29}{http://archive.ics.uci.edu/ml/datasets/Breast+Cancer+Wisconsin+\%28Diagnostic\%29}.
\newline
\normalsize{It consists of 30 real-valued continuous attributes based on the cell nuclei of 569 breast tumor patients, of which 212 instances are malignant and 357 instances are benign, along with a variable indicating whether the tumor was malignant or benign. The dataset is discussed in further detail in \citet{Street1993}.

{\centering
\includegraphics[trim=1.6in 3.8in 1.6in 3.8in, clip, height=0.29\paperwidth]{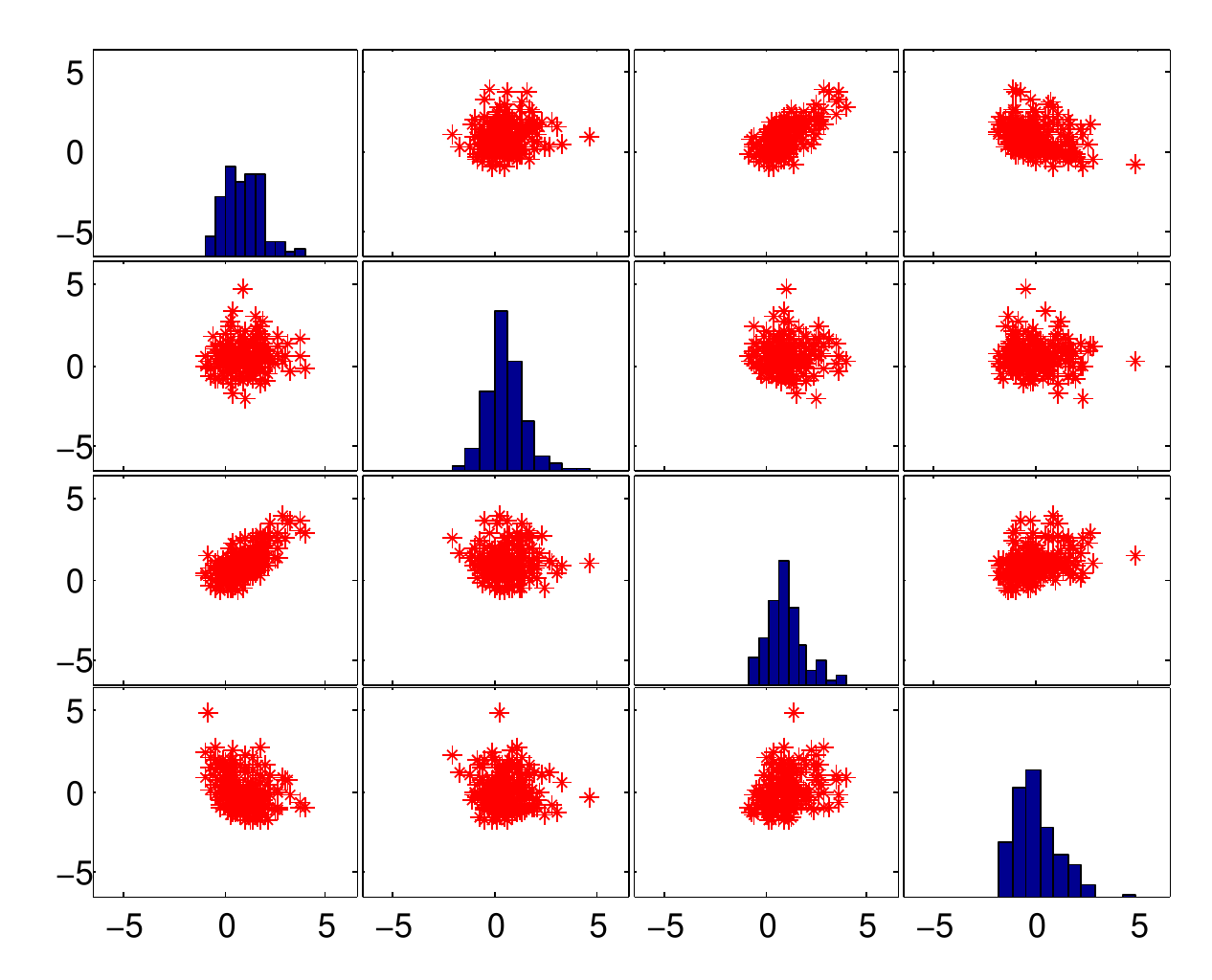}
\includegraphics[trim=1.6in 3.8in 1.6in 3.8in, clip, height=0.29\paperwidth]{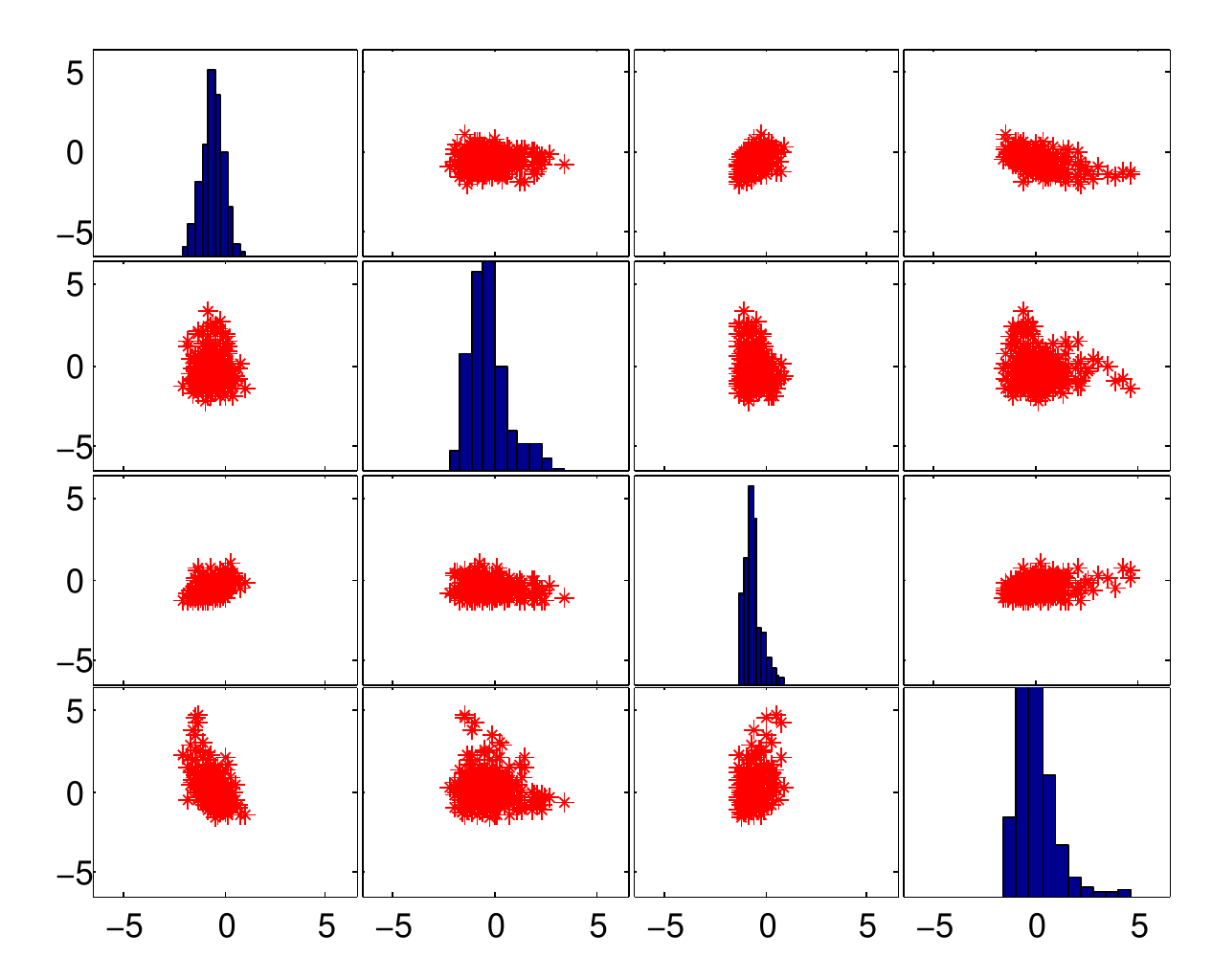}\\

}
\noindent \small{\textbf{Figure 11.} Scatterplot matrices of mean radius, mean texture, mean number of concave portions and mean fractal dimension for malignant (left) and benign (right) groups.}

\normalsize{We consider four variables from this data set: mean radius, mean texture, mean number of concave portions and mean fractal dimension. Scatterplot matrices for the benign and malignant cases are displayed in Fig.~11. Our experiment aims to classify individuals into the malignant or benign class based on observations on these four variables alone, i.e.~the class labels are ignored and are only revealed for testing the performance of the classifier. To this end, we construct a test set of 50 randomly drawn observations from the sample of 569 observations; amongst these test observations, $n^{o}_{M}$ cases correspond to malignant tumours and $n^{o}_{B}$ to benign. On the remaining 519 observations, we construct four-dimensional density estimates based on training sets of $n_{M}$ patients in the malignant group and $n_{B}$ patients in the benign group. These estimates are constructed using three different estimators: the SPE with histogram pilot estimator with bandwidth $2(IQ_{j})h_{n}$ in each coordinate direction and a dictionary of $3^4$ spherical Gaussian densities, the KDE with cross validated bandwidth, and the log-concave maximum likelihood estimator. Notice that $n_{M}^{o}$, $n_{B}^{o}$, $n_{M}$ and $n_{B}$ are random variables.}

\normalsize{Let $\widehat{f}$ denote an arbitrary density estimator and define the posterior probabilities associated with a Bayes classifier by}
\begin{equation}\label{bayesClassifier}
\widehat{P}(1|y)=\frac{\widehat{f}^{1}(y|1)\widehat{P}(1)}{\widehat{f}(y|1)\widehat{P}(1)+\widehat{f}(y|0)\widehat{P}(0)} \qquad \text{and}\qquad \widehat{P}(0|y)=\frac{\widehat{f}^{0}(y|0)\widehat{P}(0)}{\widehat{f}(y|1)\widehat{P}(1)+\widehat{f}(y|0)\widehat{P}(0)}\;.
\end{equation}
where $\hat{f}(y|1)$ and $\hat{f}(y|0)$ are density estimates obtained from the training observations in the malignant and benign group respectively, and $\hat{P}(1)$ and $\hat{P}(0)$ denote the estimated probability of being in the malignant and benign groups respectively; these quantities are obtained using the empirical proportions of malignant and benign cases in the training set (i.e.~$n_{M}/(n_{M}+n_{B})$ and $n_{B}/(n_{M}+n_{B})$ respectively). Letting $\{y_{i}: i=1,\ldots, 50\}$ denote the vector of observations for patient $i$ in the test set, the relative magnitudes of $\hat{P}(1|y_{i})$ and $\hat{P}(0|y_{i})$ determine whether individual $i$ with attribute vector $y_{i}$ is assigned to group 1 (malignant) or group 0 (benign). More specifically, individual $i$ is classified into group $j\in\{0,1\}$ if $\widehat{P}(j|y_{i})>\widehat{P}(\{j\}^{c}|y_{i})$. If the estimated posterior probabilities of a patient being the malignant and benign groups are equal, we use the pessimistic default rule of classifying the patient into the malignant group. A similar experiment to the one described above (but using the first two principal component scores rather than the variables themselves) appears in an early preprint version of \citet{SamworthReadPaper}.

We repeat the above experiment 500 times, recording the misclassification rates for the LSCV KDE and the log-concave MLE in each case, i.e.~letting $m^{o}$ denote the total number of misclassifications, the misclassification rates recorded are $m^{o}/50$, which are then averaged over the 500 experiments.

\vspace{12pt}

{\centering

\begin{tabular}{|c|c|c|c|}
\hline
                                          & KDE(CV) & LogConcave & SmoothProject \\
\hline
Mean $m^{o}/50$        &  0.0934   &  0.1596    &   0.0799    \\
Std dev. $m^{o}/50$ &  0.0424   &  0.0520    &   0.0387    \\
\hline
\end{tabular}

}
\vspace{4pt}
\noindent \small{\textbf{Table 3}: Out-of-sample misclassification rates of the Bayes classifier based on the indicated density estimator.}

\vspace{8pt}
\normalsize{}


\appendix
\section{Further propositions}\label{sectionAppendix}

\begin{proposition}\label{propContainment}
Let
\[
\mathcal{F}^{\dagger}:=\left\{f(\cdot): f(y)=\int_{\mathcal{M}}\phi(y;\mu,q \hspace{1pt}\Omega^{\dagger})dG(\mu), \; G\in \mathcal{G}\right\}
\]
where $q \hspace{1pt}\Omega^{\dagger}$ is an arbitrary covariance matrix whose diagonal elements are all equal to $q$ for $q$ fixed in $[\underline{q},\bar{q}]$. Then $\mathcal{F}^{\dagger}\subseteq \mathcal{F}_{d}^{\mathcal{G}}$.
\end{proposition}

\begin{proof}
Let $\Sigma$ and $\Sigma^{\dagger}$ be two covariance matrices. Consider the function classes 
\[
\mathcal{F}_{\Sigma}:= \left\{f:f(x)=\int \phi(x|\mu, \Sigma) dG(\mu), \; G\in\mathcal{G}\right\} \quad \text{and} \quad \mathcal{F}_{\Sigma^{\dagger}}:= \left\{f:f(x)=\int \phi(x|\mu, \Sigma^{\dagger}) dG(\mu), \; G\in\mathcal{G}\right\}.
\] 
We will first prove that, if $\Sigma - \Sigma^{\dagger}$ is non-negative definite, $\mathcal{F}_{\Sigma^{\dagger}}\subseteq \mathcal{F}_{\Sigma}$. Let $X$ be a random variable drawn from an arbitrary $f\in\mathcal{F}_{d}^{\mathcal{G}}$, then $X \stackrel{d}{=} Z + W$ where $Z\sim N (\mu, \Sigma)$, $W\sim G\in\mathcal{G}$ and $\stackrel{d}{=}$ means equality in distribution. Its moment generating function is
\begin{eqnarray*}
\mathbb{E}[\exp\{t^{T}X\}] &=& \mathbb{E}[\exp\{t^{T}(Z+W)\}]\\
 &=& \mathbb{E}[\exp\{t^{T}Z\}]\mathbb{E}[\exp\{t^{T}W\}] \\
&=&\exp \left\{t^{T}\mu + \frac{1}{2}t^{T}\Sigma t\right\}] \mathbb{E}\left[\exp\{t^{T}W\}\right] \\
& & \text{[by normality of $Z$]} \\
&=& \exp\left\{t^{T}a + \frac{1}{2}t^{T}\Sigma^{\dagger}t\right\}\exp\left\{t^{T}(\mu-a) + \frac{1}{2}t^{T}(\Sigma -\Sigma^{\dagger})t\right\}\mathbb{E}\left[\exp\{t^{T}W\}\right]
\end{eqnarray*}
so $X \stackrel{d}{=} Z + W \stackrel{d}{=} Z' + W'$, where $Z'\sim N(a,\Sigma^{\dagger})$, $W'\stackrel{d}{=}Z''+W$, $Z''\sim N((\mu-a),\Sigma - \Sigma^{\dagger})$ hence $f\in \mathcal{F}^{\dagger}$, proving that $\mathcal{F}^{\dagger}\subseteq \mathcal{F}_{d}^{\mathcal{G}}$ by the arbitrariness of $f\in \mathcal{F}_{d}^{\mathcal{G}}$. $\Sigma-\Sigma^{\dagger}$ being non-negative definite corresponds to $\Sigma \succ \Sigma^{\dagger}$ in the sense of L\"owner orderings. It is well known \citep[see e.g.][]{Mosler} that for any correlation matrix $\Omega$, $I\succ \Omega$, which proves our claim.
\end{proof}

\begin{proposition}\label{lemmaWeakConvergence}
There exists a discrete measure $G_{S}$, not necessarily unique such that
\[
\sup_{y\in\mathbb{R}^{d}}\int_{\mathcal{M}} \phi(y;\mu,\underline{q}I_{d})G_{S}(d\mu) = \sup_{y\in\mathbb{R}^{d}}\int_{\mathcal{M}} \phi(y;\mu,\underline{q}I_{d})G_{0}(d\mu) + O\left(S^{-1/d}\right).
\]
\end{proposition}

\begin{proof}

We first show that $\Phi:=\{\phi: \mu \mapsto \phi(y;\mu, \underline{q}\hspace{1pt}I_{d}); \mu\in[-M, M]^{d}\}\subset BL(\mathcal{M}):=\{f:\mathcal{M}\longrightarrow \mathbb{R}: \|f\|_{BL} < \infty\}$, the set of bounded lipschitz functions on $\mathcal{M}$,
where $\|f\|_{BL}=\|f\|_{L} + \|f\|_{\LL_{\infty}}$ with $\|f\|_{L}=\sup_{x\neq y; x,y \in \mathcal{M}}|f(x)-f(y)|/d(x,y)$. By the mean value theorem, it suffices to show that the elements of the gradient vector are bounded uniformly over $y\in\mathbb{R}^{d}$ and $\mu\in[-M,M]^{d}$. To this end, consider
\begin{eqnarray*}
& & \sup_{y\in\mathbb{R}^{d}, \mu\in \mathcal{M}} \left|\nabla_{\mu} \phi(y;\mu,\underline{q}I_{d})\right| \\
&=&  \sup_{y\in\mathbb{R}^{d}, \mu\in \mathcal{M}} \left| (\mu^{T}I_{d}^{-1})^{T}(2\underline{q}^{-(d/2+1)}) (2\pi)^{-d/2}\exp\{-(2\underline{q})^{-1}(y - \mu)^{T}(y-\mu)\}  \right| \\
&=&  \sup_{y\in\mathbb{R}^{d}, \mu\in \mathcal{M}} \left| 2\underline{q}^{-(d/2+1)} (2\pi)^{-d/2} \exp\{-(2\underline{q})^{-1}(y - \mu)^{T}(y-\mu)\} \mu \right|,
\end{eqnarray*}
which is bounded elementwise since $\underline{q}>0$ and $M<\infty$. We have shown that $\phi:\mu\mapsto \phi(y;\mu,\underline{q}\hspace{1pt}I_{d})$ is in $BL(\mathcal{M})$.

Suppose for a contradiction that no discrete measure $G_{S}$ exists such that
\[
\sup_{y\in\mathbb{R}^{d}}\int_{\mathcal{M}} \phi(y;\mu,\underline{q}I_{d})G_{S}(d\mu) \longrightarrow \sup_{y\in\mathbb{R}^{d}}\int_{\mathcal{M}} \phi(y;\mu,\underline{q}I_{d})G_{0}(d\mu)\quad \text{ as } S\longrightarrow \infty.
\]
Introduce the set of discrete measures
\[
\mathcal{G}_{s}^{\epsilon}:=\left\{G_{S}^{\epsilon}:G_{S}^{\epsilon}(A)=\sum_{s=1}^{S}\pi_{s}\delta_{\mu_{s}^{\epsilon}}(A): \pi_{1},\ldots, \pi_{S}\geq 0; \sum_{s=1}^{S}\pi_{s}=1; A\subset \mathcal{B}(\mathcal{M})\right\}
\]
where $\mathcal{B}(\mathcal{M})$ is the Borel sigma-algebra on $\mathcal{M}$,  $\delta_{\mu_{s}^{\epsilon}}(A)$ takes the value 1 if $\mu_{s}^{\epsilon}\in A$ and zero otherwise, and $\mu_{1}^{\epsilon},\ldots, \mu_{S}^{\epsilon}$ are $S$ elements of $\mathcal{M}=[-M,M]^{d}$ satisfying $\cup_{s=1}^{S}B(\mu_{s}^{\epsilon},\epsilon)=\mathcal{M}$ and $B(\mu_{s}^{\epsilon},\epsilon) \cap B(\mu_{j}^{\epsilon},\epsilon)=0$ for all $s\neq j$, $s,j \in\{1,\ldots,S\}$. Compactness of $\mathcal{M}$ guarantees the existence of such a collection of $\{\mu_{s}^{\epsilon}: s=1,\ldots,S\}$. Let $\mathcal{G}(\mathcal{M})$ be the set of all probability measures on the Borel sets of $\mathcal{M}$. Since $\mathcal{G}_{S}^{\epsilon}$ is dense in $\mathcal{G}(\mathcal{M})$ under the weak topology \citep[][Theorem 6.3]{Parth1967}, there exists a choice of weights $\{\pi_{s}: s=1,\ldots, S\}$ such that the sequence of weighted discrete measures converges to $G_{0}$ as $S\longrightarrow \infty$. It follows by Theorem 11.3.3 of \citet{Dudley2002} that $\int f dG_{s}^{\epsilon} \longrightarrow \int f dG_{0}$ for all $f\in BL(\mathcal{M})$, hence it cannot be true that $\phi:\mu\mapsto \phi(y;\mu,\underline{q}\hspace{1pt}I_{d})$ is in $BL(\mathcal{M})$, a contradiction.

We next prove that $\rho(G_{S}^{\epsilon},G_{0})=O(S^{-d/2})$ where $\rho(\cdot,\cdot)$ is the Prohorov metric defined below.

\begin{definition}\label{defProhorov}
Let $\mathbb{Q}$ and $\mathbb{P}$ be laws on $\mathcal{M}$. The \emph{Prohorov metric} is
\begin{equation}\label{eq:prokhorov2}
\rho(\mathbb{Q}, \mathbb{P}) := \inf \left\{\epsilon > 0: \mathbb{Q}(A) \leq \mathbb{P}(A^{\epsilon}) + \epsilon \; \; \text{for all Borel sets A} \right\} \\
\end{equation}
where $A^{\epsilon}:=\{y\in\mathcal{M}: d(x,y)<\epsilon \; \text{for some } x\in A\}$, i.e.~the ``$\epsilon$-enlargement'' of $A$.
\end{definition}

Fix $\epsilon>0$ and introduce the function $h(\mu)=0\vee\left(1-d(\mu,A)/\epsilon\right)$ where $d(\mu,A)=\inf_{r\in A}d(\mu,r)$. Then $h\in BL(\mathcal{M})$ \citep[][page 396 and Proposition 11.2.2]{Dudley2002} and $\ind\{\mu\in A\}\leq h(\mu) \leq \ind\{\mu\in A^{\epsilon}\}$. Cover $[-M,M]^{d}$ with disjoint open balls of radius $\epsilon/2$ around $\{\mu_{s}^{\epsilon}:s=1,\ldots,S\}$ and fix weights $\{\pi_{s}:s=1,\ldots,S\}$ such that $
\sum_{s=1}^{S}|\pi_{i}\delta_{\mu_{s}^{\epsilon}}-G_{0}(B_{s}^{\epsilon})|\leq \epsilon$, where $B_{s}^{\epsilon}:=B(\mu_{s}^{\epsilon},\epsilon)$. For an arbitrary Borel set, $A$,
\begin{eqnarray*}
G_{S}^{\epsilon}(A) & \leq & \int_{\mathcal{M}} h(\mu) G_{0}(d\mu) + \int_{\mathcal{M}} h(\mu) |G_{S}^{\epsilon}-G_{0}|(d\mu) \\
									  & \leq &  \int_{\mathcal{M}} \ind\{\mu \in A^{\epsilon}\} G_{0}(d\mu) + \int_{\mathcal{M}} h(\mu) |G_{S}^{\epsilon}-G_{0}|(d\mu) \\
									  & \leq & G_{0}(A^{\epsilon}) + \sup_{r\in \mathcal{M}}|h(r)|\sum_{s=1}^{S} |\pi_{s}\delta_{\mu_{s}^{\epsilon}}-G_{0}(B_{s}^{\epsilon})| + \sup_{r\in \mathcal{M}}|h(r)|G_{0}\left((\cup_{s=1}^{S}B_{s}^{\epsilon})^{c}\right) \\
									  &=& G_{0}(A^{\epsilon}) + \epsilon\;.
\end{eqnarray*}
Since we are allowed $S$ balls to cover $\mathcal{M}=[-M,M]^{d}$, the minimum $\epsilon$ for which the last line of the above display holds is $\epsilon = 2M/S^{1/d}$, hence the Prohorov metric converges at rate $S^{-1/d}$ when $G_{S}$ is taken as $G_{S}^{\epsilon}$.

\end{proof}

\section{Proofs}\label{sectionAppendixB}

\subsection{Definitions and preliminary lemmata}

\begin{definition}\label{defFechet}
For normed spaces $\mathbb{D}$ and $\mathbb{E}$ endowed with norms $\|\cdot\|_{\mathbb{D}}$ and $\|\cdot\|_{\mathbb{E}}$ respectively, and for some map $\Psi:\mathbb{D}\mapsto \mathbb{E}$, the Fr\'echet derivative \emph{(}if it exists\emph{)} is the linear continuous map $D\Psi_{g}:\mathbb{D}\mapsto \mathbb{E}$ such that 
\[
\left\|\Psi(g+h)-\Psi(g) - D\Psi_{g}(h)\right\|_{\mathbb{E}} = o(\|h\|_{\mathbb{D}}).
\]
\end{definition}

\begin{lemma}\label{objFunctionConv1}
Under Conditions \ref{conditionTrueDensity}, \ref{conditionMixtures}, and \ref{conditionPilot}, for any fixed $f_{\xi}\in\mathcal{F}_{d}^{S}$,
\[
\sup_{\xi\in\Xi} \left(\Phi_{P}(\widehat{f}^{P}-f_{\xi}) - \Phi_{P}(f_{0}-f_{\xi})\right)=O_{p}(v_{n}),
\]
where $v_{n}:=\max\{n^{-1/2},r_{n},s_{n}\}$ with $r_{n}$ and $s_{n}$ defined as in Condition \ref{conditionPilot}.
\end{lemma}

\begin{proof}
For notational simplicity, write $\Phi(b):=\Phi_{P}(b)$. Let $D\Phi_{b}(h)$ denote the Fr\'echet derivative of $\Phi(b)$ at $h$. Then by the definition of the Fr\'echet derivative and the functional $\Phi$,
\[
\left|\Phi(b+h)-\Phi(b)-D\Phi_{b}(h)\right|=\left|\left\langle b+h, b+h\right\rangle_{P} - \left\langle b, b\right\rangle_{P} - 2\left\langle b, h\right\rangle_{P}\right| = \|h\|_{\mathbb{L}_{2}(P)}^{2},
\]
where $\langle b, h\rangle_{P}=\int b(x)h(x)P(dx)$. Taking $b=f_{0}-f_{\xi}$ and $h=(\widehat{f}^{P}-f_{\xi})-(f_{0}-f_{\xi})=(\widehat{f}^{P}-f_{0})$, we have
\[
\mathbb{E}|\Phi(\widehat{f}^{P}-f_{\xi}) - \Phi(f_{0}-f_{\xi}) - D\Phi_{b}(\widehat{f}^{P}-f_{0})|=\mathbb{E}\|\widehat{f}^{P}-f_{0}\|_{\mathbb{L}_{2}(P)}^{2}\leq \|f_{0}\|_{\mathbb{L}_{\infty}}\mathbb{E}\|\widehat{f}^{P}-f_{0}\|_{\mathbb{L}_{2}(\text{Leb})}^{2}=O(r_{n})
\]
by Conditions \ref{conditionTrueDensity} and \ref{conditionPilot}, where the inequality follows because
\begin{eqnarray*}
\mathbb{E}\|\widehat{f}^{P}-f_{0}\|^{2}_{\mathbb{L}_{2}(P)} &=& \mathbb{E}\int |f_{0}(x)(\widehat{f}^{P}(x)-f_{0}(x))^{2}|dx = \mathbb{E}\|f_{0}(\widehat{f}^{P}-f_{0})^{2}\|_{\mathbb{L}_{1}(\text{Leb})}\\
&\leq& \|f_{0}\|_{\mathbb{L}_{\infty}}\mathbb{E}\|(\widehat{f}^{P}-f_{0})^{2}\|_{\mathbb{L}_{1}(\text{Leb})} = \|f_{0}\|_{\mathbb{L}_{\infty}}\EE\|(\widehat{f}^{P}-f_{0})\|^{2}_{\mathbb{L}_{2}(\text{Leb})}\\
\end{eqnarray*}
by non-negativity of $f_{0}$ and H\"older's inequality. To show that $(\Phi(\widehat{f}^{P}-f_{\xi}) - \Phi(f_{0}-f_{\xi}))=O_{p}(\max\{q_{n},r_{n}\})$, it thus suffices to show that $D\Phi_{b}(\widehat{f}^{P}-f_{0})=O_{p}(\max\{q_{n},r_{n}\})$. We have
\begin{eqnarray*}
D\Phi_{b}(\widehat{f}^{P}-f_{0}) &=& 2\left\langle f_{0}-f_{\xi}, \widehat{f}^{P}-f_{0}\right\rangle_{P}\\
 &=& 2 \left(\int \widehat{f}^{P}(x)(f_{0}-f_{\xi})(x)P(dx)- \int f_{0}(x)(f_{0}-f_{\xi})(x)P(dx)\right).
\end{eqnarray*}
We will next show that replacing $\int \widehat{f}^{P}(x)(f_{0}-f_{\xi})(x)P(dx)$ by a quantity depending on the empirical distribution rather than on $\widehat{f}^{P}$ only incurs a loss of $O(s_{n})$. Write $\int \widehat{f}^{P}(x)(f_{0}-f_{\xi})(x)P(dx)=\int f_{0}(x)\left(\widehat{f}^{P}(x)(f_{0}-f_{\xi})(x)\right)dx$. By Minkowski's inequality
\begin{eqnarray*}
& & \mathbb{E}\left|2\int f_{0}(x)[\widehat{f}^{P}(x)(f_{0}-f_{\xi})(x)]dx - \frac{2}{n}\sum_{i=1}^{n}f_{0}(X_{i})\left[f_{0}(X_{i})-f_{\xi}(X_{i})]\right)\right| \\
& \leq & \mathbb{E}\left|2\int f_{0}^{2}(x)\widehat{f}^{P}(x)dx - \frac{2}{n}\sum_{i=1}^{n}f_{0}^{2}(X_{i})\right| + \mathbb{E}\left|2\int f_{0}(x)\widehat{f}^{P}(x)f_{\xi}(x)dx - \frac{2}{n}\sum_{i=1}^{n}f_{0}(X_{i})f_{\xi}(X_{i})\right| \\
& = & I+II
\end{eqnarray*}
For the control over $I$ and $II$ we use the fact that for any function $g$ vanishing at infinity on $\RR^{d}$,
\[
g(x_{1},\ldots,x_{d})=\int_{-\infty}^{x_{i}}D_{i}g(x_{1},\ldots,x_{i-1},t_{i},x_{i+1},\ldots,x_{d})dt_{i}
\]
by the fundamental theorem of calculus. Iterating this argument, 
\begin{eqnarray*}
g(x)&=&\int_{-\infty}^{x_{d}}D_{d}\int_{-\infty}^{x_{d-1}}D_{d-1}\cdots\int_{-\infty}^{x_{1}}D_{1} g(t)dt_{1}dt_{2}\ldots dt_{d} \\
		&=&\int \ind\{t_{d}\in(-\infty,x_{d})\}D_{d} \cdots 	\int\ind\{t_{1}\in(-\infty,x_{1})\}D_{1} g(t)dt_{1}dt_{2}\ldots dt_{d}.
\end{eqnarray*}
With abuse of notation, we will write the right hand side as $T\bigl((g\circ \ind_{\{(-\infty,x)\}})(t)\bigr)$. $\ind\{t\in(-\infty,x)\}$ means $\ind\{t_{1}\in(-\infty,x_{1})\}\ind\{t_{2}\in(-\infty,x_{2})\},\ldots, \ind\{t_{d}\in(-\infty,x_{d})\}$, $\ind_{\{(-\infty,x)\}}(t)$ means $\ind\{t\in(-\infty,x)\}$, and $dt$ is shorthand for $dt_{1}\ldots dt_{d}$.\\

\noindent \textbf{Control over I}. We have
\begin{eqnarray*}
&   & \mathbb{E}\left|2\int f_{0}^{2}(x)\widehat{f}^{P}(x)dx - \frac{2}{n}\sum_{i=1}^{n}f_{0}^{2}(X_{i})\right| \\
& = & \mathbb{E}\left|2\int T\left((f_{0}^{2}\circ \textstyle{\ind_{\{(-\infty,x)\}}})(t)\right)\widehat{f}^{P}(x)dx -\frac{2}{n}\sum_{i=1}^{n} T \bigl((f_{0}^{2}\circ \textstyle{\ind_{\{(-\infty,X_{i})\}}})(t)\bigr)\right| \\
& = & \mathbb{E}\left|2\int_{-\infty}^{x_{d}}D_{d}\int_{-\infty}^{x_{d-1}}D_{d-1}\cdots\int_{-\infty}^{x_{1}}D_{1} f_{0}^{2})(t)\left(\int \ind\{x\in(t,\infty)\}\widehat{f}^{P}(x)dx - \frac{1}{n}\sum_{i=1}^{n}\ind\{X_{i}\in(t,\infty)\}\right)dt\right| \\
&   & \text{[by Fubini's Theorem]}  \\
& \leq & 2 \mathbb{E} \sup_{t\in\mathbb{R}^{d}}|\widehat{F}^{P}(t)-F_{n}(t)|\sum_{|\alpha|=1}\|D^{\alpha}f_{0}^{2}\|_{\mathbb{L}_{1}},
\end{eqnarray*}
where $\widehat{F}^{P}$ is the distribution function corresponding to the density function $\widehat{f}^{P}$. Since $\|D^{\alpha}f_{0}^{2}\|_{\mathbb{L}_{1}}<\infty$ by Condition \ref{conditionTrueDensity}, the last line of the above display is $O(s_{n})$ by Condition \ref{conditionPilot}. 

\vspace{8pt}

\noindent \textbf{Control over II}. Write 
\[
h(x)=f_{0}(x)f_{\xi}(x)=\int \ind\{t_{d}\in(-\infty,x_{d})\}D_{d} \cdots 	\int\ind\{t_{1}\in(-\infty,x_{1})\}D_{1} h(t)dt_{1}dt_{2}\ldots dt_{d}.
\]
Then, as in the control over $I$,
\begin{eqnarray*}
&   & \mathbb{E}\left|2\int f_{0}(x)\widehat{f}^{P}(x)f(x)dx - \frac{2}{n}\sum_{i=1}^{n}f_{0}(X_{i})f_{\xi}(X_{i})\right| \\
& = & \mathbb{E}\left|2\int T\left((h\circ \textstyle{\ind_{\{(-\infty,x)\}}})(t)\right)\widehat{f}^{P}(x)dx -\frac{2}{n}\sum_{i=1}^{n} T \bigl((h\circ \textstyle{\ind_{\{(-\infty,X_{i})\}}})(t)\bigr)\right| \\
& = & \mathbb{E}\left|2\int_{-\infty}^{x_{d}}D_{d}\int_{-\infty}^{x_{d-1}}D_{d-1}\cdots\int_{-\infty}^{x_{1}}D_{1} h(t)\left(\int \ind\{x\in(t,\infty)\}\widehat{f}^{P}(x)dx - \frac{1}{n}\sum_{i=1}^{n}\ind\{X_{i}\in(t,\infty)\}\right)dt\right| \\
& \leq & 2 \mathbb{E} \sup_{t\in\mathbb{R}^{d}}|\widehat{F}^{P}(t)-F_{n}(t)|\sum_{|\alpha|=1}\|D^{\alpha}h\|_{\mathbb{L}_{1}}.
\end{eqnarray*}
But $D^{\alpha} h = D^{\alpha}(f_{0}f_{\xi}) = (D^{\alpha}f_{0})f_{\xi} + f_{0}(D^{\alpha}f_{\xi})$, so
\begin{eqnarray*}
\|D^{\alpha}h\|_{\LL_{1}}&=&\|(D^{\alpha}f_{0})f_{\xi} + f_{0}(D^{\alpha}f_{\xi})\|_{\LL_{1}} \|(D^{\alpha}f_{0})f_{\xi}\|_{\LL_{1}} + \|f_{0}(D^{\alpha}f_{\xi})\|_{\LL_{1}} \\
												 &\leq & \|D^{\alpha}f_{0}\|_{\LL_{1}}\|f_{\xi}\|_{\infty} + \|f_{0}\|_{\infty}\|D^{\alpha}f_{\xi}\|_{\LL_{1}}<\infty
\end{eqnarray*}
by H\"older's inequality and Conditions \ref{conditionTrueDensity} and \ref{conditionMixtures}. Hence the last line of the previous display is $O(s_{n})$ by Condition \ref{conditionPilot}.

We conclude that $|D\Phi_{b}(\widehat{f}^{P}-f_{0}) - Z_{\xi}|=O_{p}(s_{n})$, where
\begin{eqnarray*}
\sqrt{n}Z_{\xi} &= & \frac{2}{\sqrt{n}}\sum_{i=1}^{n}f_{0}(X_{i})\left(f_{0}(X_{i})-f_{\xi}(X_{i})\right) - 2 \int f_{0}^{2}(x)\left(f_{0}-f_{\xi}\right)(x)dx \\
&=& \frac{2}{\sqrt{n}}\sum_{i=1}^{n}\Biggl(f_{0}(X_{i})\bigl(f_{0}(X_{i})-f_{\xi}(X_{i})\bigr) - \mathbb{E}\Bigl[f_{0}(X_{i})\bigl(f_{0}(X_{i})-f_{\xi}(X_{i})\bigr)\Bigr]\Biggr) = O_{p}(1),
\end{eqnarray*}
by the central limit theorem thus,
\[
\left(\Phi_{P}(\widehat{f}^{P}-f_{\xi}) - \Phi_{P}(f_{0}-f_{\xi})\right) = O_{p}\bigl(\max\{n^{-1/2},r_{n},s_{n}\}\bigr) \quad \text{for all } \xi \in\Xi.
\]
It remains to show that the rate holds uniformly over $\Xi$, i.e.~$\sup_{\xi\in\Xi}|Z_{\xi}|=O_{p}\bigl(\max\{n^{-1/2},r_{n},s_{n}\}\bigr)$. Introduce the class of functions $\mathcal{G}_{\Xi}=\{g_{\xi}=f_{0}(f_{0}-f_{\xi}):\;\xi\in\Xi\}$. By Minkowski's inequality and a double application of H\"older's inequality, for a $g\in\mathcal{G}_{\Xi}$
\[
\|g\|_{\mathbb{L}_{1}}=\|f_{0}(f_{0}-f_{\xi}) \|_{\mathbb{L}_{1}} \leq \|f_{0}f_{0}\|_{\mathbb{L}_{1}}+\|f_{0}f_{\xi}\|_{\mathbb{L}_{1}} \leq \|f_{0}\|_{\mathbb{L}_{2}}^{2} + \|f_{0}\|_{\mathbb{L}_{2}}\|f_{\xi}\|_{\mathbb{L}_{2}} < \infty,
\]
hence $g\in\mathbb{L}_{1}(\mathbb{R}^{d},\text{Leb})$ and $(g-g')\in\mathbb{L}_{1}(\mathbb{R}^{d},\text{Leb})$ for all $g, g'\in\mathcal{G}_{\Xi}$. Moreover, by the Lipschitz requirement of Condition \ref{conditionMixtures}, there exists a $K<\infty$ such that
\[
\|g-g'\|_{\mathbb{L}_{1}}=\|g_{\xi}-g_{\xi'}\|_{\mathbb{L}_{1}}=\|f_{0}(f_{\xi'}-f_{\xi})\|_{\mathbb{L}_{1}}\leq \|f_{0}\|_{\mathbb{L}_{2}}\|(f_{\xi'}-f_{\xi})\|_{\mathbb{L}_{2}} \leq \|f_{0}\|_{\mathbb{L}_{2}} K \|\xi-\xi'\|_{\ell_{1}}.
\]
Compactness of $\Xi$ (Condition \ref{conditionMixtures}) implies that $(\Xi,d)$ is totally bounded where $d(\xi,\xi')=\|\xi-\xi'\|_{\ell_{1}}$, and the previous display implies a mapping from finitely many $\epsilon$-balls on $(\Xi,d)$ to finitely many $\delta(\epsilon)$-balls on $(\mathcal{G}_{\Xi},\rho)$ with $\rho$ the $\mathbb{L}_{1}(\mathbb{R}^{d},\text{Leb})$ norm; hence $(\mathcal{G}_{\Xi},\rho)$ is totally bounded. 
Introduce the sets $\{B_{\LL_{1}}(g_{j},\delta): j=1,\ldots, J\}$, $J<\infty$ such that $\mathcal{G}_{\Xi}=\cup_{j=1}^{J}B_{\LL_{1}}(g_{j},\delta)$ and the corresponding sets $\{B_{\ell_{1}}(\xi_{j},\bar{\delta}): j=1,\ldots, J\}$, where $\bar{\delta}=\delta/\|f_{0}\|_{\mathbb{L}_{2}}K$. Introduce $\mathbb{M}_{n}=(P_{n}-P)$, the unscaled empirical process on $\mathcal{G}_{\Xi}$, then
\begin{eqnarray*}
& & P\Bigl(\sup_{\xi\in\Xi}|Z_{\xi}|>Cn^{-1/2}\Bigr)= P\Bigl(\sup_{g\in\mathcal{G}_{\Xi}}|\mathbb{M}_{n}g|>Cn^{-1/2}\Bigr)\\
& \leq &P\Bigl(\max_{1\leq j\leq J}\sup_{g\in B_{\LL_{1}}(g_{j},\delta)}\left(|\mathbb{M}_{n}g - \mathbb{M}_{n}g_{j}| + |\mathbb{M}_{n}g_{j}|\right) >Cn^{-1/2} \Bigr) \\
& \leq &P\Bigl(\sup_{g'\in\mathcal{G}_{\Xi}}\sup_{g\in B_{\LL_{1}}(g',\delta)}|\mathbb{M}_{n}g - \mathbb{M}_{n}g'|> \frac{C}{2}n^{-1/2} \Bigr) + P\Bigl(\max_{1\leq j\leq J}|\mathbb{M}_{n}g_{j}|>\frac{C}{2}n^{-1/2} \Bigr).
\end{eqnarray*}
By the pointwise convergence already established, we know that for all $\epsilon>0$, there exists a $\bar{K}=\bar{K}(\epsilon)<\infty$ and an $n_{0}(\epsilon)>0$ such that
\[
P\Bigl(\max_{1\leq j\leq J}|\mathbb{M}_{n}g_{j}|>\bar{K} v_{n} \Bigr) < \frac{\epsilon}{2} \quad \text{for all } n>n_{0}. 
\]
We also know by Condition \ref{conditionMixtures} that for all $\epsilon>0$, there exists a $\bar{\bar{K}}=\bar{\bar{K}}(\epsilon)<\infty$, a $\bar{\bar{\delta}}=\bar{\bar{\delta}}(\epsilon)>0$ and an $n'_{0}(\epsilon)>0$ such that
\begin{equation}\label{stochEqui}
P\Bigl(\sup_{\xi,\xi': \|\xi-\xi'\|_{\ell_{1}}<\bar{\bar{\delta}}}|\mathbb{M}_{n}f_{0}(f_{\xi'}-f_{\xi})|>\bar{\bar{K}} v_{n} \Bigr) < \frac{\epsilon}{2} \quad \text{for all } n>n_{0}'.
\end{equation}
Taking $\delta=\bar{\bar{\delta}}(\epsilon)\|f_{0}\|_{\mathbb{L}_{2}}K$ and $C=2\max\{\bar{K}(\epsilon),\bar{\bar{K}}(\epsilon)\}$ guarantees that for all $\epsilon$, $\displaystyle{P\bigl(\sup_{\xi\in\Xi}|Z_{\xi}|>C v_{n}\bigr)<\epsilon}$.
\end{proof}

\begin{lemma}\label{lemmaObjFuncConv2}
Under Conditions \ref{conditionTrueDensity} and \ref{conditionMixtures}, for a fixed $f_{\xi}\in \mathcal{F}_{d}^{S}$,
\begin{equation}\label{eqNormalLim2}
\sqrt{n}\left(\Phi_{P_{n}}(f_{0}-f_{\xi}) - \Phi_{P}(f_{0}-f_{\xi})\right) \longrightarrow_{d} N(0,\sigma(f_{0},f_{\xi})),
\end{equation}
where
\begin{eqnarray*}
\sigma(f_{0},f_{\xi})&=&\int f_{0}^{3}(x)dx + \int f_{\xi}^{2}(x)f_{0}(x)dx + 2\int f_{0}^{2}(x)dx\int f_{\xi}(x) f_{0}(x)dx \\
&  & \quad - \;\; 2 \int f_{\xi}(x)f_{0}^{2}(x)dx - \bigl(\int f_{0}^{2}(x)\bigr)^{2} - \bigl(\int f_{\xi}(x)f_{0}(x)dx\bigr)^{2}.
\end{eqnarray*}
Moreover, $\displaystyle{\sup_{\xi\in\Xi}\left(\Phi_{P_{n}}(f_{0}-f_{\xi}) - \Phi_{P}(f_{0}-f_{\xi})\right) =O_{p}(v_{n})}$, where $v_{n}$ is as defined in Lemma \ref{objFunctionConv1}.
\end{lemma}

\begin{proof}
Noting that every term in $\sigma(f_{0},f_{\xi})$ is bounded by Conditions \ref{conditionTrueDensity} and \ref{conditionMixtures} together with H\"older's inequality, convergence to the limit distribution in \eqref{eqNormalLim2} follows by the central limit theorem, which implies pointwise convergence at rate $\sqrt{n}$, and a fortiori at rate $v_{n}=\max\{n^{-1/2},r_{n},s_{n}\}$. The extension to uniform convergence at rate $v_{n}$ follows by an analogous argument to that used in the proof of Lemma \ref{objFunctionConv1}.
\end{proof}

\begin{lemma}\label{lemmaVdVWEq3}
Under Conditions \ref{conditionTrueDensity}-\ref{conditionPilot}, $\displaystyle{\sup_{\xi\in\Xi}P\left((\widehat{f}^{P}-f_{\xi})^{2}-(f_{0}-f_{\xi})^{2}\right)^{2}\longrightarrow_{p} 0}$.
\end{lemma}

\begin{proof}
For typographical convenience, we write $\widehat{f}^{P}$ as $\widehat{f}$. By non-negativity of $f_{0}$ and $f_{\xi}$ on $\RR^{d}$ and H\"older's inequality,
\begin{eqnarray*}
				&    &P\left((\widehat{f}^{P}-f_{\xi})^{2}-(f_{0}-f_{\xi})^{2}\right)^{2}\\
 				&\leq& \int|\widehat{f}^{2}-f_{0}^{2}|dP + 2\int f_{\xi}|f_{0}-\widehat{f}|dP \\
																																		&  = &	\int|f_{0}(x)\bigl((\widehat{f}(x)-f_{0}(x))(\widehat{f}(x)+f_{0}(x))\bigr)|dx + 2\int |f_{0}(x)f_{\xi}(x)(f_{0}(x)-\widehat{f}(x))|dx \\
																																		&\leq& \|\widehat{f}-f_{0}\|_{\LL_{2}(\text{Leb})}\|f_{0}(\widehat{f}-f_{0})\|_{\LL_{2}(\text{Leb})}+ \|f_{0}f_{\xi}\|_{\LL_{2}(\text{Leb})}\|f_{0}f_{\xi}\|_{\LL_{2}(\text{Leb})}\|f_{0}-\widehat{f}\|_{\LL_{2}(\text{Leb})}\\
																																		& = & I+II(\xi).
\end{eqnarray*}
Hence it suffices to prove that $I\longrightarrow_{p}0$ and $\sup_{\xi\in\Xi}II(\xi)\longrightarrow_{p}0$. By Condition \ref{conditionPilot} and Markov's inequality $\|f_{0}-\widehat{f}\|_{\LL_{2}(\text{Leb})}\longrightarrow_{p}0$. We also have 
\[
\|f_{0}f_{\xi}\|_{\LL_{2}(\text{Leb})}=|\langle f_{0}^{2}f_{\xi}^{2}\rangle|^{1/2}\leq \bigl(\|f_{0}^{2}\|_{\LL_{2}}\|f_{\xi}^{2}\|_{\LL_{2}}\bigr)^{1/2}
\] 
therefore $\sup_{\xi\in\Xi}\|f_{0}f_{\xi}\|_{\LL_{2}(\text{Leb})}<\infty$ by Conditions \ref{conditionTrueDensity} and \ref{conditionMixtures}, and $\sup_{\xi\in\Xi}II(\xi)\longrightarrow_{p}0$. $\|f_{0}(\widehat{f}-f_{0})\|_{\LL_{2}(\text{Leb})}\leq \|\widehat{f}^{2}\|_{\LL_{2}(\text{Leb})}\|f_{0}^{2}\|_{\LL_{2}(\text{Leb})}2\max\{\|\widehat{f}^{2}\|_{\LL_{2}(\text{Leb})},\|f_{0}^{2}\|_{\LL_{2}(\text{Leb})}\}<\infty$ by Conditions \ref{conditionTrueDensity} and \ref{conditionPilot}, hence $I\longrightarrow_{p}0$.
\end{proof}

The following Lemma, which is stated here for ease of reference, is Theorem 2.1 of \citet{vdVW2007}.

\begin{lemma}\label{lemmaVdVWThm1}
Write $Q_{g,\xi}:=(g-f(\xi))^{2}$. Let $\mathcal{H}$ be such that $f_{0}\in\mathcal{H}$ and suppose that $\mathcal{H}_{0}$ is a fixed subset of $\mathcal{H}$ such that $\Pr(\widehat{f}_{n}\in\mathcal{H}_{0})\longrightarrow 1$, where $\widehat{f}_{n}$ is a sequence of estimators for $f_{0}$. If $\{Q_{f,\xi}:\xi\in\Xi, \; f\in\mathcal{H}_{0}\}$ is $P$-Donsker and $\sup_{\xi\in \Xi}P\bigl(Q_{\widehat{f}_{n},\xi}-Q_{f_{0},\xi}\bigr)^{2}\longrightarrow_{p} 0$, then $\sup_{\xi\in\Xi}\bigl|\sqrt{n}(P_{n}-P)(Q_{\widehat{f}_{n},\xi}-Q_{f_{0},\xi})\bigr|\longrightarrow_{p} 0$.
\end{lemma}

The proof of Lemma \ref{minimumExists}, which is proved below, is based on several preliminary results from convex analysis, which are stated here for ease of reference. For the proofs, see the corresponding references.

\begin{definition}\emph{\citep[][Theorem III.D]{RobertsVarberg1973}.}
Suppose $\emph{(}\pi_{1},\ldots,\pi_{S}\emph{)}\in \Delta^{S}$ where $\Delta^{S}$ is the unit $S$ simplex and $S<\infty$. Then $x=\sum_{s=1}^{S}\pi_{s}x_{s}$ is a \emph{convex combination} of $\{x_{1},\ldots,x_{S}\}$, the latter being elements of the linear space $L$.
\end{definition}

\begin{lemma}\label{lemmaRVIIID}\emph{\citep[][Theorem III.D]{RobertsVarberg1973}.}
Let $U\subseteq L$, where $L$ is a linear space. The convex hull of $U$, $\text{\emph{conv}}(U)$ consists precisely of all convex combinations of elements of $U$.
\end{lemma}

\begin{lemma}\label{lemmaRVIIIE}\emph{\citep[][Theorem III.E]{RobertsVarberg1973}.}
Let $U\subseteq L$, where $L$ is a linear space. If the convex hull $\text{\emph{conv}}(U)$ has dimension $n$, then for each $x\in \text{\emph{conv}}(U)$, there exist $n+1$ points $x_{1}\ldots, x_{n+1}\in U$ such that $x$ is a convex combination of those points.
\end{lemma}

\begin{definition}\emph{\citep[][page 1]{Phelps1966}.}
Suppose that $X$ is a nonempty compact subset of a locally convex space $E$ and that $\nu$ is a probability measure on $X$. A point $x\in E$ is said to be \emph{represented} by $\nu$ if $h(x)=\int_{X}h d\nu$ for every continuous linear functional $h$ on $E$.
\end{definition}

\begin{lemma}\label{lemmaPhelps}\emph{\citep[][Proposition 1.2]{Phelps1966}}.
Suppose that $U$ is a compact subset of a locally convex space $E$. A point $x\in E$ is in the closed convex hull of $U$ if and only if there exists a probability measure $\nu$ on $U$ that represents $x$.
\end{lemma}


\subsection{Proofs of main results}

\begin{proof} \text{[}Lemma \ref{minimumExists}\text{]}
Introduce the notation $f_{G}=\int_{\Theta} f_{\theta}G(d\theta)$. The set $\mathcal{G}$ of all probability measures on $\Theta$ is convex by compactness of $\Theta$. Since the feasible region, $\mathcal{G}$ and the objective function, $\Phi_{P_{n}}(\widehat{f}^{P}-f_{G})$, are convex, a minimum exists. Define the atomic and mixture vectors as $\mathbf{f}_{\theta}:=\left(f_{\theta}(Y_{1}),\ldots, f_{\theta}(Y_{n})\right)\in \RR^{n}$ and $\mathbf{f}_{G}:=\left(f_{G}(Y_{1}),\ldots, f_{G}(Y_{n})\right)\in \RR^{n}$ respectively. $\Gamma:=\{\mathbf{f}_{\theta}:\theta\in\Theta\}\subset \RR^{n}$ represents all possible fitted values of the atomic vector. The convex hull of $\Gamma$, written $\text{conv}(\Gamma)$, is the intersection of all convex sets containing $\Gamma$, and is itself a convex set. By Lemma \ref{lemmaRVIIID}, $\text{conv}(\Gamma)=\{\mathbf{f}_{G}: G \in\mathcal{G}^{S}\}\subset \RR^{n}$, where $\mathcal{G}^{S}$ is the set of probability measures on $\Theta$ with support on $S<\infty$ points in $\Theta$. Since compactness of $\Theta$ implies compactness of $\Gamma$ Lemma \ref{lemmaPhelps} delivers the stronger result that $\text{conv}(\Gamma)=\{\mathbf{f}_{G}:G\in\mathcal{G}\}\subset \RR^{n}$ under measurability of the map $\theta\mapsto f_{\theta}$, which means that any probability measure $G$ on $\Theta$ corresponds to a probability measure on $\Gamma$. Thus ensuring that the infinite dimensional minimisation problem is equivalent to a finite dimensional one. Finally, by Lemma \ref{lemmaRVIIIE}, for any $\mathbf{f}\in\text{conv}(\Gamma)$ there exist points $\mathbf{f}_{1},\ldots,\mathbf{f}_{S}\in\Gamma$ with $S\leq n+1$ such that $\mathbf{f}$ is a convex combination of these points. We conclude that $\text{conv}(\Gamma)=\{\mathbf{f}_{G}:G\in\mathcal{G}\}=\{\mathbf{f}_{G}:G\in\mathcal{G}^{S}\}$ with $S\leq n+1$.

Any $f_{\widehat{G}}$ of the form $f_{\widehat{G}}=\int f_{\theta} \widehat{G}(d\theta)$ that minimises $\Phi_{P_{n}}(\widehat{f}^{P}-f_{G})=\frac{1}{n}\sum_{i=1}^{n}(\widehat{f}^{P}(Y_{i})-f_{G}(Y_{i}))^{2}$ for a fixed pilot estimate $\widehat{f}^{P}$ corresponds to a $\mathbf{f}_{\widehat{G}}=(f_{\widehat{G}}(Y_{1}),\ldots,f_{\widehat{G}}(Y_{n}))$ that minimises $\Phi_{P_{n}}(\widehat{f}^{P}-f_{G})$. Since $\mathbf{f}_{G}\in \text{conv}(\Gamma)$, it is a convex combination of at most $n+1$ points in $\Gamma$. We conclude that $\widehat{G}$ has at most $n+1$ points of support.
\end{proof}

\begin{proof}\text{[}Proposition \ref{propositionKernel}\text{]}
The proof is a corollary of Corollary 1.2 of \citet{MasonSwanepoel2013} (see also Corollary 1.3 op.~cit.) as long as we show that any $0<h<b_{n}$ delivers $\|\EE\widehat{F}^{k}_{n,h}-F_{0}\|_{\infty}=o(\sqrt{n^{-1}\log\log n})$. To this end, write, for any $x\in\RR^{d}$,
\begin{eqnarray*}
& & |\EE\widehat{F}^{k}_{n,h}(x)-F_{0}(x)|= |k_{h}\ast F_{0}(x)-F_{0}(x)| \\
&=& \bigl|\int_{\RR^{d}}F_{0}(x-zh)k(z)dz - F_{0}(x)\bigr|	= \bigl|\int_{\RR^{d}}\bigl(F_{0}(x-zh) - F_{0}(x)\bigr)k(z)dz\bigr| 
\end{eqnarray*}
where we have used the substitution $(u_{j}-y_{j})/h\mapsto z_{j}$ in $k_{h}\ast F_{0}(x)=\int_{-\infty}^{x_{1}} \cdots \int_{-\infty}^{x_{d}}\frac{1}{h^{d}}k\bigl(\frac{u-y}{h}\bigr)P(dy)du$ and the fact that $\int_{\RR^{d}}k(z)dz=1$. Expanding $F_{0}(x-hz)$ into a Taylor series around $x$ with Laplacian representation for the remainder \citep[][page 126]{Ziemer1989}, we have, in the notation defined in Section \ref{subSecNotation},
\begin{eqnarray*}
F_{0}(x-hz)&=& F_{0}(x)\; + \sum_{0\leq|\alpha|\leq \ell}\frac{1}{\alpha!}D^{\alpha}F_{0}(x)(-hz)^{\alpha} \\
					 & & \quad + \;\;(\ell+1)\sum_{|\alpha|=\ell+1}\frac{1}{\alpha!}\left[\int_{0}^{1}(1-t)^{\ell}D^{\alpha}F_{0}\left[(1-t)x+t(x-hz)\right]dt\right](-hz)^{\alpha}.
\end{eqnarray*}
Therefore, $\|\EE\widehat{F}^{k}_{n,h}-F_{0}\|_{\infty}\leq h^{\ell+1}\int_{\RR^{d}}|k(z)||z^{\alpha}|dz\sum_{|\alpha|=\ell+1}\|D^{\alpha}F_{0}\|_{\infty}$, where we have used the requirement on the kernel that $\int_{\RR^{d}}z^{\alpha}k(z)dz=0$. Hence any $h=o(n^{-1/2(\ell+1)}\sqrt{\log\log n})$ delivers $\|\EE\widehat{F}^{k}_{n,h}-F_{0}\|_{\infty}=o(n^{-1/2}\sqrt{\log\log n})$.
\end{proof}

\begin{proof}\text{[}Propostion \ref{propositionDonsker}\text{]}
\begin{eqnarray*}
D^{\alpha}f_{n,h_{n}}^{k}(x) &=& \frac{1}{n h_{n}^{d}} \sum_{i=1}^{n} D^{\alpha} \prod_{j=1}^{d} k\Bigl(\frac{x_{j}-X_{ij}}{h_{n}}\Bigr) \\
														 &=& \frac{1}{\sqrt{2\pi}} \frac{1}{n h_{n}^{d}} \sum_{i=1}^{n} \exp\Bigl\{-\frac{(x_{j}-X_{ij})^{2}}{h_{n}^{2}}\Bigr\} D_{j}^{\alpha_{j}} \Bigl(\frac{(x_{j}-X_{ij})^{2}}{h_{n}^{2}}\Bigr)
\end{eqnarray*}
and we see that, for any $\alpha$,
\begin{eqnarray*}
& & \|D^{\alpha}\widehat{f}_{n,h_{n}}^{k} - D^{\alpha}(k_{h_{n}}\ast f_{0})\|_{\LL_{1}} \\
&=& \int_{\RR^{d}} \left|\frac{1}{\sqrt{2\pi}}\frac{1}{h_{n}^{d}}\int_{\RR^{d}}\prod_{j=1}^{d} \exp\Bigl\{-\frac{(x_{j}-y_{j})^{2}}{h_{n}^{2}}\Bigr\}D_{j}^{\alpha_{j}}\Bigl(\frac{(x_{j}-y_{j})^{2}}{h_{n}^{2}}\Bigr)(P_{n}-P)(dy)\right|dx \longrightarrow_{a.s.} 0 \;\; \text{as } n\rightarrow \infty
\end{eqnarray*}
as long as $h_{n}^{d}\searrow 0$ slower than $O(n^{-1/2})$, i.e.~as long as $h_{n}\searrow 0$ slower than $O(n^{-1/2d})$.  Since $D^{\alpha}f_{n,h_{n}}^{K}\longrightarrow_{a.s.} D^{\alpha}(K_{h_{n}}\ast f_{0})$ in $\LL_{1}(\RR^{d},\text{Leb})$ for any $\alpha$, $P(\widehat{f}_{n,h_{n}}^{K}\in \mathfrak{W}_{\infty,2})\longrightarrow 1$ as $n\longrightarrow \infty$. The conclusion follows because $\mathfrak{W}_{\ell,2}\subset \mathcal{D}$ for $\ell>d/2$ by Theorem 1.3 of \citet{Marcus1985}. 
\end{proof}

\begin{proof} \text{[}Proposition \ref{propositionNormalMixtures}\text{]}.
Since $M<\infty$ and $\underline{q}>0$, every mixture density belongs to $\mathbb{L}_{p}(\mathbb{R}^{d},\text{Leb})$ for all $p<\infty$. Since vector spaces are closed under addition, we conclude that $f_{\xi}$ belongs to $\mathbb{L}_{p}(\mathbb{R}^{d},\text{Leb})$ for all $p<\infty$ as well. $\bar{\mathcal{F}}_{d}^{S}\subset \mathfrak{W}_{1,1}$ is equivalent to the statement, for all $f_{\xi}\in\bar{\mathcal{F}}_{d}^{S}$ and for all $j\in\{1,\ldots,d\}$, $Df_{\xi}=\sum_{|\alpha|=1}D^{\alpha}f_{\xi}\in\LL_{1}(\RR^{d},\text{Leb})$ for $|\alpha|=\sum_{j=1}^{d}\alpha_{j}=1$, where $D^{\alpha}$ is defined in Section \ref{subSecNotation}. To show this, consider 
\begin{eqnarray*}
\frac{\partial}{\partial y_{j}}f_{\xi}&=&\sum_{s=1}^{S}\pi_{s}\frac{\partial}{\partial y_{j}} \phi_{s} =	(-1)\underline{q}^{-d}(2\pi)^{-d/2}\sum_{s=1}^{S}\pi_{s}\exp\Bigl\{-\frac{1}{2\underline{q}}\sum_{j=1}^{d}(y_{j}-\mu_{s,j})^{2}\Bigr\}(y_{j}-\mu_{s,j})\\
																			&=:&(-1)\underline{q}^{-d}(2\pi)^{-d/2}\sum_{s=1}^{S}\Delta_{s,j}.
\end{eqnarray*}
Since the tails of $\Delta_{s,j}$ are subexponential in $y$ for any $s\in\{1,\ldots, S\}$ and for any $j\in\{1,\ldots,d\}$, the previous display is in $\mathbb{L}_{1}(\mathbb{R},\text{Leb})$, therefore so too is $Df_{\xi}$.

To verify the Lipschitz requirement of Condition \ref{conditionMixtures}, it suffices by the discussion following Condition \ref{conditionMixtures} to show that $A_{f_{\xi}}\in\mathbb{L}_{2}(\mathbb{R}^{d},\text{Leb})$, where $A_{f_{\xi}}$ is defined in equation \eqref{suffLipschitz}. Let $\xi_{\ell}=\pi_{s}$ for $\ell=s<S$ and $\xi_{\ell}=\mu_{s}$ for $\ell=S+s$, $s<S$ arbitrary. We show that for any $j\in\{1,\ldots 2S\}$, $\Bigl\|\sup_{\xi\in\Xi}\frac{\partial}{\partial \xi_{j}}f_{\xi}\Bigr\|_{\ell_{\infty}} \in \mathbb{L}_{2}(\mathbb{R}^{d},\text{Leb})$. Let $j\leq S$, then
\[
\Bigl\|\sup_{\xi\in\Xi}\frac{\partial}{\partial \xi_{j}}f_{\xi}(y)\Bigr\|_{\ell_{\infty}}=\sup_{\pi,\mu\in\Xi}\bigl|\frac{\partial}{\partial \pi_{s}}f_{\pi,\mu}(y)\bigr|=\sup_{\mu_{s}\in\mathcal{M}}\phi(y,\mu_{s},\underline{q}I_{d})
\]
hence $\Bigl\|\sup_{\xi\in\Xi}\bigl|(\partial/\partial \xi_{j})f_{\xi}\Bigr\|_{\ell_{\infty}}\in\mathbb{L}_{2}(\mathbb{R}^{d},\text{Leb})$ for any $j\leq S$. Let $j>S$, then
\begin{eqnarray*}
\Bigl\|\sup_{\xi\in\Xi}\frac{\partial}{\partial \xi_{j}}f_{\xi}(y)\Bigr\|_{\ell_{\infty}} &=& \Bigl\|\sup_{\pi,\mu\in\Xi}\frac{\partial}{\partial \mu_{s}}f_{\pi,\mu}(y)\Bigl\|_{\ell_{\infty}}\\
&=&\Bigl\|\sup_{\pi,\mu\in\Xi}\pi_{s}2\underline{q}^{-(d/2 + 1)}(2\pi)^{-d/2}\exp\{-(2\underline{q})^{-1}w_{s}^{T}w_{s}\}\mu_{s}\Bigr\|_{\ell_{\infty}}\\
																			& = & \Bigl\|\sup_{\pi,\mu\in\Xi}\pi_{s}2\underline{q}^{-1}\phi(y,\mu_{s},\underline{q}I_{d})\mu_{s}\Bigr\|_{\ell_{\infty}},
\end{eqnarray*}
where $w_{s}=(y-\mu_{s})$. Since $\underline{q}>0$ and $M<\infty$, $\|\sup_{\pi,\mu\in\Xi}\pi_{s}2\underline{q}^{-1}\phi(\cdot,\mu_{s},\underline{q}I_{d})\mu_{s}\|_{\ell_{\infty}} \in \mathbb{L}_{2}(\mathbb{R}^{d},\text{Leb})$, hence the Lipschitz requirement is fulfilled. Finally,
\begin{eqnarray*}
& & \sup_{\xi,\xi'\in\Xi: \; \|\xi-\xi'\|_{\ell_{1}}<\delta}\Bigl|\frac{1}{n}\sum_{i=1}^{n}\bigl[f_{0}(Y_{i})(f_{\xi'}(Y_{i})-f_{\xi}(Y_{i})\bigr)\bigr]-\mathbb{E}\bigl[f_{0}(Y_{i})\bigl(f_{\xi'}(Y_{i})-f_{\xi}(Y_{i})\bigr)\bigr]\Bigr| \\
&=& \sup_{\xi,\xi'\in\Xi: \; \|\xi-\xi'\|_{\ell_{1}}<\delta}\Bigl|\bigl(\frac{1}{n}\sum_{i=1}^{n} - \mathbb{E}\bigr) f_{0}(Y_{i}) \left(\nabla_{\xi}^{T}f_{\bar{\xi}}(Y_{i})(\xi-\xi')\right)\Bigr| \quad \bar{\xi}\in\text{\emph{conv}}\{\xi,\xi'\} \\
&\leq & \delta \bigl\|\sup_{y\in\RR^{d}}\sup_{\xi\in\Xi}\nabla_{\xi}f_{\xi}(y)\bigr\|_{\ell_{\infty}}\Bigl|\frac{1}{n}\sum_{i=1}^{n}f_{0}(Y_{i}) - \mathbb{E}f_{0}(Y_{i})\Bigr|
\end{eqnarray*}
by H\"older's inequality. From the previous calculations we see that all entries of $\bigl\|\sup_{\xi\in\Xi}\sup_{y\in\mathbb{R}^{d}}\nabla_{\xi}f_{\xi}(y)\bigr\|_{\ell_{\infty}}$ are bounded. Now $\mathbb{E}f_{0}(Y_{i})=\int f_{0}(y)P(dy)=\int f_{0}^{2}(y)dy <\infty$ because $f_{0}\in \mathbb{L}_{2}(\mathbb{R}^{d},\text{Leb})$ by Condition \ref{conditionTrueDensity}. The latter implies that $\bigl[\mathbb{E}f_{0}(Y_{i})\bigr]^{2}<\infty$, which in turn implies that $\text{Var}(f_{0}(Y_{i}))=\mathbb{E}\bigl[f_{0}^{2}(Y_{i})\bigr] - \bigl[\mathbb{E}f_{0}(Y_{i})\bigr]^{2} < \infty$ because $\mathbb{E}\bigl[f_{0}^{2}(Y_{i})\bigr] = \int f_{0}^{3}(y)dy <\infty$, again by Condition \ref{conditionTrueDensity}. The $\bigl(f_{0}(Y_{i})\bigr)_{i=1}^{n}$ are clearly i.i.d., hence invoking the central limit theorem, we have $\bigl|\frac{1}{n}\sum_{i=1}^{n}f_{0}(Y_{i}) - \mathbb{E}f_{0}(Y_{i})\bigr|=O_{p}(n^{-1/2})$, which a fortiori is $O_{p}(v_{n})$.
\end{proof}

\begin{proof} \text{[}Theorem \ref{thmRate}\text{]}. Write $P_{n}Q_{\widehat{f},\xi}:=\widehat{\mathbb{M}}_{n}(\xi)$ and $PQ_{f_{0},\xi}:=\mathbb{M}_{0}(\xi)$. With this notation, $Q_{g,\xi}$ is $(g-f(\xi))^{2}$ as in Lemmata \ref{lemmaVdVWEq3} and \ref{lemmaVdVWThm1}. We have the decomposition
\begin{equation}\label{eqDecomposition}
\sqrt{n}\bigl(P_{n}Q_{\widehat{f},\xi} - PQ_{f_{0},\xi}\bigr)=\mathbb{G}_{n}\bigl(Q_{\widehat{f},\xi}-Q_{f_{0},\xi}\bigr) + \mathbb{G}_{n}Q_{f_{0},\xi} + \sqrt{n}P\bigl(Q_{\widehat{f},\xi}-Q_{f_{0},\xi}\bigr),
\end{equation}
where $\mathbb{G}_{n}Q=\sqrt{n}(P_{n}-P)Q$ is the empirical process at $Q$. Noting that $PQ_{g,\xi}=\Phi_{P}(g-f_{\xi})$ and $\mathbb{G}_{n}Q_{f_{0},\xi}=\sqrt{n}(\Phi_{P_{n}}(f_{0}-f_{\xi})-\Phi_{P}(f_{0}-f_{\xi}))$, Lemmata \ref{objFunctionConv1} and \ref{lemmaObjFuncConv2} provide the required control over the third term and the second term respectively. The required control over the first term comes from an application of Lemma \ref{lemmaVdVWEq3} followed by an application of \ref{lemmaVdVWThm1}, noting that, since $\Xi$ is a finite dimensional parameter space, the Donsker condition on $\mathcal{H}_{0}$ guarantees that $\{Q_{f,\xi}:\; \xi\in\Xi, f\in\mathcal{H}_{0}\}$ is $P$-Donsker as well. 
\end{proof}

\begin{proof} \text{[}Theorem \ref{thmConsistency} \text{]}
Let $\widehat{\mathbb{M}}_{n}=\int(\widehat{f}^{P}-f)^{2}dP_{n}$ and $\mathbb{M}_{0}=\int (f_{0}-f)^{2}dP$ be stochastic processes indexed by $\Xi$. By Theorem \ref{thmRate}, $\sup_{\xi\in\Xi}|\widehat{\mathbb{M}}_{n}(\xi) - \mathbb{M}_{0}(\xi)|\longrightarrow_{p} 0$. By the unique minimiser assumption, and the assumption that $\xi^{*}_{0}$ belongs to the interior of $\Xi$, there exists a $\xi^{*}_{0}$ such that $\mathbb{M}_{0}(\xi^{*}_{0})<\inf_{\xi \notin G}\mathbb{M}_{0}(\xi)$ for every open set $G$ that contains $\xi^{*}_{0}$. Since, by the statement of the theorem, there exists a sequence $\widehat{\xi}_{n}^{*}$ such that $\widehat{\mathbb{M}}_{n}(\xi_{n}^{*})\leq \inf \widehat{\mathbb{M}}_{n}(\xi)+o_{p}(1)$, $\widehat{\xi}^{*}_{n}\longrightarrow \xi_{0}^{*}$ in outer probability by Corollary 3.2.3 of \citet{vdVW1996}.
\end{proof}

\begin{proof} \text{[}Proposition \ref{propositionProjectEmpirical}\text{]} \; Since $\pi_{0}^{*}=(\pi_{0,1}^{*},\ldots, \pi_{0,S}^{*})$ are known, the projection $f_{n,g}^{*}=f(\mu_{n,g,1}^{*},\ldots,\mu_{n,g,S}^{*}, \pi_{0}^{*})$ of $g$ on $\bar{\mathcal{F}}_{d}^{S}$ satisfies
\begin{equation}\label{FOC}
0=\frac{\partial L_{n} (f(\pi_{0},\mu);g)}{\partial \mu_{s}} (\mu_{n,g}^{*},\pi_{0}^{*}) \quad \forall s\in\{1,\ldots,S\}, \;\;\; \mu_{n,g}^{*}=(\mu_{n,g,1}^{*},\ldots,\mu_{n,g,S}^{*}),
\end{equation}
i.e.
\begin{equation}\label{FOC2}
0= 2\sum_{i=1}^{n}(g(Y_{i})-\sum_{s=1}^{S}\pi_{0,s}^{*}\phi(Y_{i};\mu_{n,g,s}^{*},\underline{q}I_{d}))V_{s}(Y_{i};\pi_{0,s}^{*},\mu_{n,g,s}^{*})\quad \forall s\in\{1,\ldots,S\}
\end{equation}
where $\displaystyle{V_{s}(Y_{i};\pi_{0,s}^{*},\mu_{n,g,s}^{*})=\pi_{0,s}^{*}\phi(Y_{i};\mu_{n,g,s}^{*},\underline{q}I_{d})\underline{q}^{-1}I_{d}(Y_{i}-\mu_{n,g,s}^{*})}$.

From \eqref{FOC2} we see that, at a minimising $\{\mu_{n,g,1}^{*},\ldots, \mu_{n,g,S}^{*}\}$, $\{Y_{1},\ldots,Y_{n}\}\subset \mathcal{I}^{c}\cup \mathcal{I}$, where $\mathcal{I}=\mathcal{J}\cup\mathcal{E}$,
\begin{equation}\label{setIc}
\mathcal{I}^{c}:=\left\{Y_{i}:\; \sum_{s=1}^{S}\pi_{0,s}^{*}\phi(Y_{i};\mu_{n,g,s}^{*},\underline{q}I_{d})=w_{g,i}\delta(Y_{i})\right\},
\end{equation}
\begin{equation}\label{setJ}
\mathcal{J}:=\left\{Y_{i}:\; V_{s}(Y_{i},\pi_{0,s}^{*},\mu_{n,g,s}^{*})=0 \; \forall s\in\{1,\ldots,S\}\right\},
\end{equation}
and
\begin{equation}\label{setE}
\mathcal{E}:= \left\{Y_{i}:\; \sum_{j: Y_{j}\in\mathcal{E}}\left(g(Y_{i})-\sum_{s=1}^{S}\pi_{0,s}^{*}\phi(Y_{i};\mu_{n,g,s}^{*},\underline{q}I_{d})\right)=0\right\}.
\end{equation}
Notice that $Y_{i}\in\mathcal{J}$ if and only if $\|Y_{i}-\mu_{n,g,s}^{*}\|_{\ell_{2}}=\infty$ for all $s\in \{1,\ldots,S\}$, thus if there exists a $Y_{i}$ in $\mathcal{J}$, then $\mathcal{J}=\{Y_{1},\ldots,Y_{n}\}$; this is scenario (ii). We have $\mathcal{J}=\{Y_{1},\ldots,Y_{n}\}$ or $\mathcal{J}=\emptyset$.

Consider $\mathcal{J}=\emptyset$, then $\mathcal{J}^{c}=\mathcal{I}^{c}\cup \mathcal{E}$. Since $S<n$, there exists a set $\mathcal{S}\neq \emptyset$ such that $Y_{i}\notin\mathcal{I}^{c}$ for all $Y_{i}\in\mathcal{S}$. Since $\mathcal{J}=\emptyset$, $\mathcal{S}=\mathcal{E}$ and therefore $\mathcal{E}\neq \emptyset$.

Introduce the sets 
\[
\mathcal{E}_{g}^{-}:=\left\{Y_{i}: \; \sum_{s=1}^{S}\pi_{0,s}^{*}\phi(Y_{i};\mu_{n,g,s}^{*},\underline{q}I_{d}) < w_{g,i}\right\}\; \text{ and } \;
\mathcal{E}_{g}^{+}:=\left\{Y_{i}: \; \sum_{s=1}^{S}\pi_{0,s}^{*}\phi(Y_{i};\mu_{n,g,s}^{*},\underline{q}I_{d}) > w_{g,i}\right\},
\] 
and notice that at a minimum $\mathcal{E}=\mathcal{E}_{g}^{+}\cup \mathcal{E}_{g}^{-}$, with $\mathcal{E}_{g}^{-}\neq \emptyset$ and $\mathcal{E}_{g}^{+}\neq \emptyset$ by the definition of $\mathcal{E}$ and the fact that $\mathcal{E}\neq \emptyset$. Re-writing the clause in $\mathcal{E}$ more explicitly in terms of $\mathcal{E}_{g}^{-}$ and $\mathcal{E}_{g}^{+}$, we have
\begin{equation}\label{eqEquality}
\begin{aligned}
& \; |\mathcal{E}_{g}^{-}|\left(\frac{1}{|\mathcal{E}_{g}^{-}|}\sum_{i:Y_{i}\in\mathcal{E}^{-}_{g}}\left(w_{g,i}-\sum_{s=1}^{S}\pi_{0,s}^{*}\phi(Y_{i};\mu_{n,g,s}^{*},\underline{q}I_{d})\right)\right) \\
 = & \; -|\mathcal{E}_{g}^{+}|\left(\frac{1}{|\mathcal{E}_{g}^{+}|}\sum_{i:Y_{i}\in\mathcal{E}^{+}_{g}}\left(w_{g,i}-\sum_{s=1}^{S}\pi_{0,s}^{*}\phi(Y_{i};\mu_{n,g,s}^{*},\underline{q}I_{d})\right)\right).
\end{aligned}
\end{equation}
By log concavity of $f_{0}$, we can define a nested sequence of closed convex sets $\mathcal{R}_{1}^{f_{0}}\subset \mathcal{R}_{2}^{f_{0}} \subset \cdots$ such that $\mathcal{R}_{\ell}^{f_{0}}:=\left\{y\in\mathbb{R}^{d}: f_{0}(y)\geq r_{\ell}\right\} \quad r_{\ell}>r_{k} \;\; \forall k>\ell$, with $\limsup_{n\longrightarrow \infty} \mathcal{R}_{k(n)}^{f_{0}}=E^{f_{0}}:=\left\{y\in\mathbb{R}^{p}: f_{0}> 0\right\}=\text{supp}(f_{0})$. We can similarly define $\mathcal{R}_{n,\ell}^{f_{0}}:=\left\{y\in\{Y_{1},\ldots,Y_{n}\}: f_{0}(y)\geq r_{\ell}\right\} \quad r_{\ell}>r_{k} \;\; \forall k>\ell$, hence $\mathcal{R}_{n,k(n)}^{f_{0}}$ is the empirical analogue of $\mathcal{R}_{k(n)}^{f_{0}}$ and $E_{n}^{f_{0}}=\limsup_{n\longrightarrow \infty} \mathcal{R}_{n,k(n)}^{f_{0}}$. Since $f_{0}\in\mathcal{F}_{d}^{LC}$, there exists an $n'$ such that for all $n>n'$, $|\mathcal{R}_{n,k(n)}^{f_{0}}|>|E_{n}^{f_{0}}\backslash \mathcal{R}_{n,k(n)}^{f_{0}}|$ with probability 1. Suppose for a contradiction that $\mu_{n,g}^{*}:=\{\mu_{n,g,1}^{*},\ldots,\mu_{n,g,S}^{*}\}\subset \mathcal{R}_{k(n)}^{f_{0}}$, then $|\mathcal{R}_{n,k(n)}^{f_{0}} \cap \mathcal{E}_{g}^{+}|=|\mathcal{R}_{n,k(n)}^{f_{0}}|$ with probability 1. Since
$|\mathcal{R}_{n,k(n)}^{f_{0}}|+|E_{n}^{f_{0}}\backslash \mathcal{R}_{n,k(n)}^{f_{0}}|=n$, we conclude that $|\mathcal{E}_{g}^{+}|>|\mathcal{E}_{g}^{-}|$ for all $n>n'$ with probability 1. Moreover, since $w_{g,i}=n^{-1}$,
\[
\frac{1}{|\mathcal{E}_{g}^{-}|}\sum_{i:Y_{i}\in\mathcal{E}^{-}_{g}}\left(w_{g,i}-\sum_{s=1}^{S}\pi_{0,s}^{*}\phi(Y_{i};\mu_{n,g,s}^{*},\underline{q}I_{d})\right) < \frac{|\mathcal{E}_{g}^{-}|}{|\mathcal{E}_{g}^{-}|}\left(\frac{1}{n}\right)
\]
whilst
\[
-\frac{1}{|\mathcal{E}_{g}^{+}|}\sum_{i:Y_{i}\in\mathcal{E}^{+}_{g}}\left(w_{g,i}-\sum_{s=1}^{S}\pi_{0,s}^{*}\phi(Y_{i};\mu_{n,g,s}^{*},\underline{q}I_{d})\right) \nearrow C>0 \; \text{ as } n\longrightarrow \infty.
\]
Therefore there exists an $n_{0}\geq n'$ such that for all $n>n_{0}$
\begin{eqnarray*}
& & \; -|\mathcal{E}_{g}^{+}|\left(\frac{1}{|\mathcal{E}_{g}^{+}|}\sum_{i:Y_{i}\in\mathcal{E}^{+}_{g}}\left(w_{g,i}-\sum_{s=1}^{S}\pi_{0,s}^{*}\phi(Y_{i};\mu_{n,g,s}^{*},\underline{q}I_{d})\right)\right)\\
& > & |\mathcal{E}_{g}^{-}|\left(\frac{1}{|\mathcal{E}_{g}^{-}|}\sum_{i:Y_{i}\in\mathcal{E}^{-}_{g}}\left(w_{g,i}-\sum_{s=1}^{S}\pi_{0,s}^{*}\phi(Y_{i};\mu_{n,g,s}^{*},\underline{q}I_{d})\right)\right).
\end{eqnarray*}
This proves that there exists at least one $Y_{i}\in\mathcal{S}$ such that $Y_{i}\notin \mathcal{E}$, which is a contradiction to $\mathcal{S}=\mathcal{E}$. We conclude that a minimum is unobtainable with $\mathcal{J}=\emptyset$ and $\mu_{n,g}^{*}\subset \mathcal{R}_{n,k(n)}^{f_{0}}$; if $\mathcal{J}=\emptyset$ at a minimum, then $\mu_{n,g}^{*}\subset \mathbb{R}^{p}\backslash\mathcal{R}_{n,k(n)}^{f_{0}}$.
\end{proof}

\begin{proof} \text{[}Proposition \ref{propositionProjectLogConcave}\text{]} \; The sets $\mathcal{J}$, $\mathcal{I}^{c}$ and $\mathcal{E}$ are those of equations \eqref{setJ}, \eqref{setIc} and \eqref{setE} respectively. As in the proof of Proposition \ref{propositionProjectLogConcave}, either $\mathcal{J}=\{Y_{1},\ldots,Y_{n}\}$ or $\mathcal{J}=\emptyset$. If $\mathcal{J}=\{Y_{1},\ldots,Y_{n}\}$, then the minimum is the one with $\|\mu_{n,g,s}\|_{\ell_{2}}=\infty$ for all $s\in\{1,\ldots,S\}$.

Consider minima achieved with $\mathcal{J}=\emptyset$. Then for any $Y_{i}\in\{Y_{1},\ldots,Y_{n}\}$, $Y_{i}\in\mathcal{J}^{c}=\mathcal{I}^{c}\cup \mathcal{E}$ at a minimum. Since $f_{0}\notin \bar{\mathcal{F}}_{d}^{S}$ and $\widehat{f}^{P}\notin \bar{\mathcal{F}}_{d}^{S}$, with probability 1 there exists a $\mathcal{S}\neq \emptyset$ such that $Y_{i}\notin \mathcal{I}^{c}$ for all $Y_{i}\in\mathcal{S}$. Hence $\mathcal{S}\neq \emptyset$ with probability 1, and at a minimum $\mathcal{S}=\mathcal{E}$.

As in the proof of Proposition \ref{propositionProjectEmpirical}, introduce the sets $\mathcal{E}_{g}^{-}$ and $\mathcal{E}_{g}^{+}$ 
and notice that at a minimum $\mathcal{E}=\mathcal{E}_{g}^{+}\cup \mathcal{E}_{g}^{-}$, with $\mathcal{E}_{g}^{-}\neq \emptyset$ and $\mathcal{E}_{g}^{+}\neq \emptyset$ by the definition of $\mathcal{E}$ and the fact that $\mathcal{E}\neq \emptyset$.

Clearly, when $\mu_{n,g}^{*}$ violates condition (i),
\begin{eqnarray*}
& & \; 
\sum_{i:Y_{i}\in\mathcal{E}^{-}_{g}}\left(w_{g,i}-\sum_{s=1}^{S}\pi_{0,s}^{*}\phi(Y_{i};\mu_{n,g,s}^{*},\underline{q}I_{d})\right)\\
& > & -\sum_{i:Y_{i}\in\mathcal{E}^{+}_{g}}\left(w_{g,i}-\sum_{s=1}^{S}\pi_{0,s}^{*}\phi(Y_{i};\mu_{n,g,s}^{*},\underline{q}I_{d})\right)
\end{eqnarray*}
with probability 1 for sufficiently large $n$, thereby contradicting the requirement for a minimum that $Y_{i}\in\mathcal{E}$ whenever $Y_{i}\notin \mathcal{I}^{c}$.

\end{proof}

\noindent \textbf{Acknowledgements:} The first author is grateful to Princeton University for generous hospitality from March-April 2013 and to the EPSRC for financial support under grants EP/D063485/1 and EP/K004581/1. We thank Matthew Arnold, John Aston, Alicia Nieto-Reyes, Victor Panaretos, Jonty Rougier and Richard Samworth for constructive conversations related to aspects of this work. 

\bibliographystyle{ims}
\bibliography{mvMixBib}
\end{document}